\pdfoutput=1
\RequirePackage{ifpdf}
\ifpdf 
\documentclass[pdftex]{sigma}
\else
\documentclass{sigma}
\fi

\usepackage{longtable}
\usepackage{tikz}
\usetikzlibrary{shapes.geometric, arrows}
\graphicspath{{./images/}}

\NewDocumentCommand{\binomial}{omm}
 {%
 \genfrac(){0pt}{}{#2}{#3}%
 \IfValueT{#1}{_{\!#1}}%
 }
\NewDocumentCommand{\eulerian}{omm}
 {%
 \genfrac<>{0pt}{}{#2}{#3}%
 \IfValueT{#1}{_{\!#1}}%
 }

\def \s {\sigma}

\newtheorem{thm}{Theorem}[section]
\newtheorem{prop}[thm]{Proposition}

\newtheorem{conj}[thm]{Conjecture}
\newtheorem*{claim}{Claim}

\theoremstyle{definition}

\newtheorem{defn}[thm]{Definition}
\newtheorem{question}[thm]{Question}
\newtheorem{rem}[thm]{Remark}

\newcommand{\sfR}{\mathsf{R}}
\newcommand{\sfW}{\mathsf{W}}

\usepackage{diagbox}

	\def \be {\begin{equation}}
		\def \ee {\end{equation}}
	\def \ba {\begin{eqnarray}}
		\def \ea {\end{eqnarray}}
	\def \bg {\begin{gather}}
		\def \eeg {\end{gather}}

	\def \be {\begin{equation}}
		\def \en {\end{equation}}
	\def \bes {\begin{eqnarray}}
		\def \ens {\end{eqnarray}}

	\def \s {\mathfrak{s}}
	\def \t {\mathfrak{t}}

 \def \min {{\rm minimal}}
	
	\numberwithin{equation}{section}

\usepackage[labelsep=period,labelfont=bf,font=small]{caption}
	
\begin{document}

\newcommand{\arXivNumber}{2212.11243}

\renewcommand{\PaperNumber}{016}

\FirstPageHeading

\ShortArticleName{Color-Dressed Generalized Biadjoint Scalar Amplitudes: Local Planarity}

\ArticleName{Color-Dressed Generalized Biadjoint Scalar\\ Amplitudes: Local Planarity}

\Author{Freddy CACHAZO~$^{\rm a}$, Nick EARLY~$^{\rm b}$ and Yong ZHANG~$^{\rm a}$}

\AuthorNameForHeading{F.~Cachazo, N.~Early and Y.~Zhang}

\Address{$^{\rm a)}$~Perimeter Institute for Theoretical Physics, Waterloo, ON N2L 2Y5, Canada}
\EmailD{\href{mailto:fcachazo@pitp.ca}{fcachazo@pitp.ca}, \href{mailto:yzhang@pitp.ca}{yzhang@pitp.ca}}

\Address{$^{\rm b)}$~Max Planck Institute for Mathematics in the Sciences, Leipzig, Germany}
\EmailD{\href{mailto:nick.early@mis.mpg.de}{nick.early@mis.mpg.de}}

\ArticleDates{Received August 25, 2023, in final form February 04, 2024; Published online February 21, 2024}

\Abstract{The biadjoint scalar theory has cubic interactions and fields transforming in the biadjoint representation of ${\rm SU}(N)\times {\rm SU}\big({\tilde N}\big)$. Amplitudes are ``color'' decomposed in terms of partial amplitudes computed using Feynman diagrams which are simultaneously planar with respect to two orderings. In 2019, a generalization of biadjoint scalar amplitudes based on generalized Feynman diagrams (GFDs) was introduced. GFDs are collections of Feynman diagrams derived by incorporating an additional constraint of ``local planarity'' into the construction of the arrangements of metric trees in combinatorics. In this work, we propose a natural generalization of color orderings which leads to color-dressed amplitudes. A generalized color ordering (GCO) is defined as a collection of standard color orderings that is induced, in a precise sense, from an arrangement of projective lines on $\mathbb{RP}^2$. We present results for $n\leq 9$ generalized color orderings and GFDs, uncovering new phenomena in each case. We discover generalized decoupling identities and propose a definition of the ``colorless'' generalized scalar amplitude. We also propose a notion of GCOs for arbitrary~$\mathbb{RP}^{k-1}$, discuss some of their properties, and comment on their GFDs. In a companion paper, we explore the definition of partial amplitudes using CEGM integral formulas.}

\Keywords{generalized color orderings; generalized Feynman diagrams; generalized decoupling identities; CEGM integral formulas; generalized biadjoint amplitudes}

\Classification{14M15; 05E99; 14T99}

\section{Introduction}

Tree-level scattering amplitudes of gluons are organized in terms of partial amplitudes and color structures as follows (see \cite{Mangano:1990by} for a review):
\begin{equation}\label{coYM}
 {\mathcal A}_n(\{ k_i,\epsilon_i,a_i\} ) = \sum_{\sigma \in S_{n}/{\mathbb{Z}_n}} \operatorname{tr}\big( T^{a_{\sigma (1)}}T^{a_{\sigma (2)}} \cdots T^{a_{\sigma (n)}}\big)A(\sigma(1),\sigma(2),\dots ,\sigma(n)).
\end{equation}
Here $T^a$ are generators of ${\rm SU}(N)$, which is traditionally called the ``color'' group.

When gluons are replaced by scalars, ${\rm SU}(N)$ is called a ``flavor'' group but in this work we do not make a distinction and uniformly use ``color''. We are mainly interested in the biadjoint cubic scalar theory \cite{Cachazo:2013iea}. This theory carries a group ${\rm SU}(N)\times {\rm SU}\big(\tilde N\big)$ and its tree amplitudes can be organized in terms of partial amplitudes as follows:
\begin{align}
& {\mathcal M}_n(\{ k_i,a_i,{\tilde a}_i\} ) = \sum_{\alpha,\beta \in S_{n}/{\mathbb{Z}_n}} \operatorname{tr}\left( T^{a_{\alpha (1)}}T^{a_{\alpha (2)}} \cdots T^{a_{\alpha (n)}}\right)\nonumber\\
 &\hphantom{{\mathcal M}_n(\{ k_i,a_i,{\tilde a}_i\} ) =}{}
 \times \operatorname{tr}\left( T^{{\tilde a}_{\beta (1)}}T^{{\tilde a}_{\beta (2)}} \cdots T^{{\tilde a}_{\beta (n)}}\right)m_n(\alpha,\beta).\label{coBS}
\end{align}
The expressions \eqref{coYM} and \eqref{coBS} are {\it color decompositions} of the corresponding amplitudes.

Biadjoint partial amplitudes $m_n(\alpha,\beta)$ have a simple formula, as the sum over tree-level $\phi^3$ Feynman diagrams which are planar with respect to both orderings. Alternatively, $m_n(\alpha,\beta)$ has a Cachazo--He--Yuan (CHY) formulation in terms of an integral over the configuration space of~$n$ points on $\mathbb{CP}^1$ with integrands that depend on the orderings in a simple manner \cite{Cachazo:2013hca,Cachazo:2013iea}.

In 2019, Guevara, Mizera, and the first two authors proposed a generalization of the CHY construction to integrals over the configuration space of $n$ points on $\mathbb{CP}^{k-1}$ \cite{Cachazo:2019ngv} (see \cite{Arkani-Hamed:2019mrd,Drummond:2019qjk,Drummond:2020kqg, Gates:2021tnp,He:2020ray,He:2021zuv} for related work and connections to cluster algebras). The standard biadjoint theory corresponds to $k=2$ while $k>2$ leads to Cachazo--Early--Guevara--Mizera (CEGM) generalized biadjoint amplitudes. Also in \cite{Cachazo:2019ngv}, a connection to the tropical Grassmannian \cite{speyer2004tropical}, $\operatorname{Trop} G(k,n)$, was proposed and proven for the positive part $\operatorname{Trop}^+G(3,6)$ \cite{speyer2005tropical}. Tropical Grassmannians are closely related to the space of arrangements of metric trees \cite{herrmann2008draw}, which provide the formulation of a special class of generalized biadjoint amplitudes in terms of {\it planar} generalized Feynman diagrams (GFDs) as defined in \cite{Borges:2019csl}.

While the first steps towards generalizing Feynman diagrams to $k>2$ have been taken, the notion of color factors has been missing completely. In this work, we fill this gap by proposing a~notion of $k>2$ color factors, using techniques inspired by the theory of oriented matroids~\cite{bjorner1999oriented}. We first concentrate on $k=3$, and define a $k=3$ color ordering as a collection of $n$, $k=2$ orderings on $n-1$ labels which can be derived from a generic arrangement of projective lines on $\mathbb{RP}^2$. For example,
\begin{gather}\label{introEx}
\Sigma := \big(\sigma^{(1)},\sigma^{(2)},\sigma^{(3)},\sigma^{(4)},\sigma^{(5)}\big)=((2345),(1345),(1245),(1235),(1234))
\end{gather}
is one of the twelve possible color orderings for $(k,n)=(3,5)$ amplitudes.

A convenient way to draw the arrangement of $n$ lines is by taking a chart of $\mathbb{RP}^2$ as a plane with a circle at infinity with its antipodal points identified. The genericity assumption guarantees that each line is intersected by the others in a way that defines a $(k=2,n-1)$ color ordering. Figure~\ref{fiveNaive} shows a representation of $\Sigma$.

\begin{figure}[t]\centering
\begin{tikzpicture}

\begin{scope}[xshift=5.3cm,yshift=0cm, scale=3]
 \draw[ dashed](0,0) circle (1);
 \draw[blue,very thick] (90+2:1) node [above]{\color{black}1} -- (-90:1) node [below]{\color{black}1} ;

 \draw[olive,very thick] (45+2:1) node [above]{\color{black}2} -- (-115-20:1) node [below]{\color{black}2} ;

 \draw[red,very thick] (35+2:1) node [right]{\color{black}4} -- (-180-10:1) node [left]{\color{black}4} ;

 \draw[cyan,very thick] (0+2:1) node [right]{\color{black}5} -- (-180-30:1) node [left]{\color{black}5} ;

 \draw[violet,very thick] (-65:1) node [below right=-2pt]{\color{black}3} -- (-180-45:1) node [ above left=-2pt]{\color{black}3} ;

\end{scope}

\begin{scope}[xshift=-2cm,yshift=-4cm, scale=3]
\draw[gray, thick] (0,0) -- (5,0) -- (5,-1) -- (0,-1) -- (0,0)
(1,0) -- (1,-1)
(2,0) -- (2,-1)
(3,0) -- (3,-1)
(4,0) -- (4,-1);
\draw[blue,very thick] (0.5,-0.5)
node[below=.9cm] {\color{black}2}
node[left=.9cm] {\color{black}3}
node[above=.9cm] {\color{black}4}
node[right=.9cm] {\color{black}5}
circle (.3) ;

\draw[olive,very thick] (1.5,-0.5)
node[below=.9cm] {\color{black}1}
node[left=.9cm] {\color{black}3}
node[above=.9cm] {\color{black}4}
node[right=.9cm] {\color{black}5}
circle (.3) ;

\draw[red,very thick] (3.5,-0.5)
node[below=.9cm] {\color{black}5}
node[left=.9cm] {\color{black}3}
node[above=.9cm] {\color{black}2}
node[right=.9cm] {\color{black}1}
circle (.3) ;

\draw[cyan,very thick] (4.5,-0.5)
node[below=.9cm] {\color{black}4}
node[left=.9cm] {\color{black}3}
node[above=.9cm] {\color{black}2}
node[right=.9cm] {\color{black}1}
circle (.3) ;

\draw[violet,very thick] (2.5,-0.5)
node[below=.9cm] {\color{black}4}
node[left=.9cm] {\color{black}5}
node[above=.9cm] {\color{black}1}
node[right=.9cm] {\color{black}2}
circle (.3) ;
\end{scope}

\end{tikzpicture}

\caption{Top: Arrangement of lines corresponding to the generalized color ordering $\Sigma = ((2345),\allowbreak(1345),(1245),(1235),(1234))$. The dashed circle is at infinity, points on it with the same label are identified. Each line defines a $k=2$ color ordering by identifying its points on the boundary to make a~circle. Bottom: Five $(k,n)=(2,4)$ color orderings obtained from $\Sigma$ by the order in which lines intersect a given one.}	\label{fiveNaive}
\end{figure}
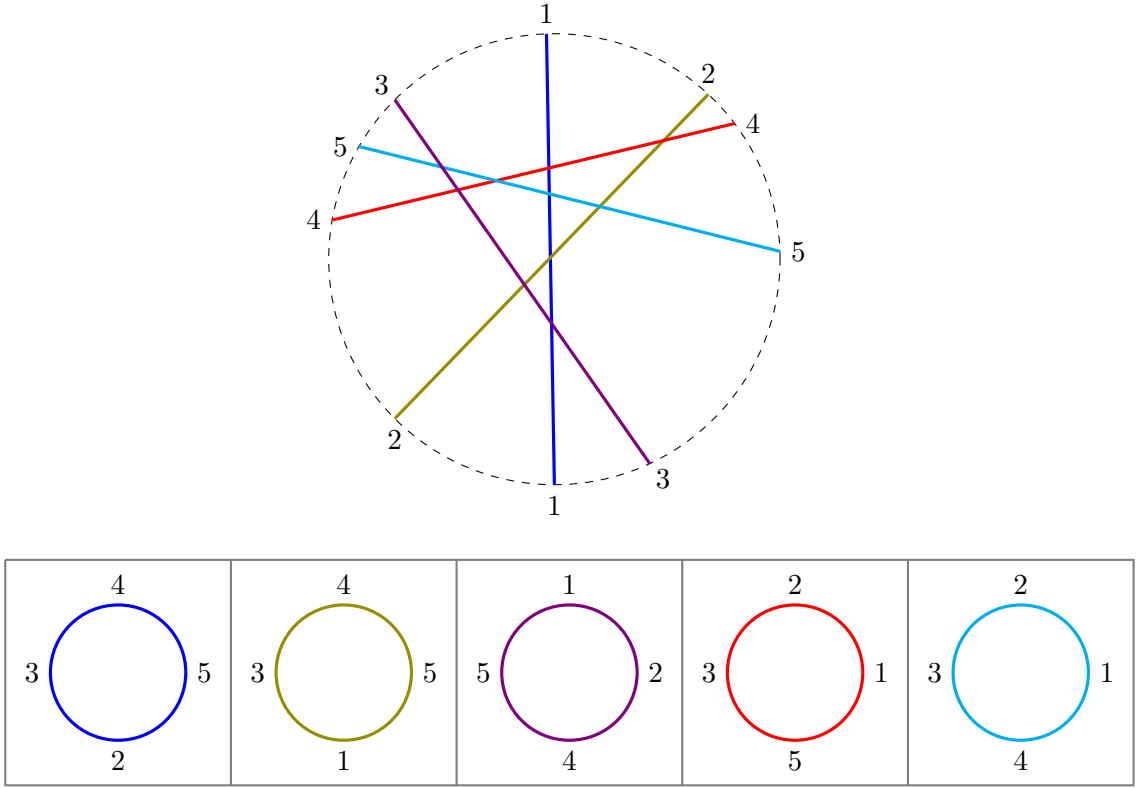

Associated with each ordering there is a color factor
${\bf c}(\Sigma)$\footnote{See \cite{Early:2023dlt} for a recent study of the Lie-theoretic
realization of these color factors.} so that the color dressed generalized biadjoint amplitude is given by
\begin{equation*}
 {\mathcal M}^{(3)}_n = \sum_{I,J=1}^{{\cal N}_{3,n}} {\bf c}(\Sigma_I){\bf c}(\Sigma_J) m_n^{(3)}(\Sigma_I,\Sigma_J),
\end{equation*}
where ${\cal N}_{3,n}$ is the number of $k=3$ color orderings for $n$ points.

Most results in the literature only deal with partial amplitudes associated with the \textit{positive} tropical Grassmannian (cf.\ \cite{Arkani-Hamed:2019rds,Arkani-Hamed:2020cig,MR4449810,Early:2019eun,Henke:2021ity,
Lukowski:2020dpn,speyer2005tropical,speyer2021positive}). As we show in Section \ref{three}, these correspond to partial amplitudes in which both color orderings satisfy a very restrictive property, a ``global'' notion of planarity.

Having a precise definition of generalized color orderings (GCOs) allows us to provide a~complete characterization of the generalized Feynman diagrams needed to fully characterize and compute \smash{$m_n^{(3)}(\Sigma_I,\Sigma_J)$}. We define a GFD as an arrangement of metric trees which is compatible with at least one generalized color ordering, in the sense that each Feynman diagram in the collection is planar with respect to the corresponding ordering. This is the notion of {\it local planarity}.

In general, we then define \smash{$m_n^{(3)}(\Sigma_I,\Sigma_J)$} as a sum over GFDs that are compatible, or locally planar, with both color orderings $\Sigma_I$ and $\Sigma_J$.

In order to illustrate the concepts, we study \smash{${\mathcal M}^{(3)}_6$} in detail. We find ${\cal N}_{3,6}=372$ distinct color orderings. Unlike the standard $k=2$ case in which all $(n-1)!/2$ color orderings are related by relabelings, for $(k,n)=(3,6)$ there are four types of orderings, with $60$, $180$, $120$, and $12$ orderings in each type respectively. For $(k,n)=(3,7)$, we find ${\cal N}_{3,7}=27\, 240$ color orderings that fall into eleven types. Now $(k,n)=(3,7)$ also provides the first example of an arrangement of metric trees which is {\it not} a generalized Feynman diagram since it is not compatible with any color ordering. This also leads to the first examples of GFDs which contain trees with both degree-four and degree-three vertices in their collections.

We provide a list of representatives for each of the $135$ types of color ordering for ${(k,n)=(3,8)}$ in an appendix, whose permutations give ${\cal N}_{3,8}= 4\, 445\, 640$ color orderings in total.
We present the $4381$ types of $(3,9)$ color orderings in an ancillary file and using relabelling one finds ${{\cal N}_{3,9} = 1 \,553\, 388\, 480}$ color orderings in total.

The reader familiar with the configuration space of points in the projective plane, $X(3,n)$, would recognize the numbers of color orderings for $n=5,6,7$ and their partition into types as the same as the number of chambers and their types, as well as with numbers of reorientation classes of oriented uniform matroids (for a detailed discussion of $n=6$ see \cite{yoshida2013hypergeometric} and for $n=7$ see \cite{sekiguchi1999configurations}). This is not an accident. In fact, the CEGM construction directly computes partial amplitudes as integrals over $X(3,n)$. In a companion paper \cite{Cachazo:2023ltw}, we explore color-dressed amplitudes from the CEGM integral viewpoint, uncovering fascinating connections to canonical forms \cite{Arkani-Hamed:2020cig}, reorientation classes of uniform oriented matroids, the tropical Grassmannian, and the hypersimplex.

In Section \ref{sec7}, we arrive at the first non-trivial application of generalized color orderings by introducing the notion of decoupling and its corresponding identities among partial amplitudes, in analogy with the famous ${\rm U}(1)$ decoupling identifies in gauge theories (see \cite{Mangano:1990by} for a review).

In Section \ref{sec8}, we introduce a generalization of the single scalar field $\phi^3$ theory. This theory has two natural definitions, which we claim are equivalent. The first is as a sum over all~$(k,n)$ generalized Feynman diagrams while the second is as a sum over all diagonal generalized biadjoint amplitudes.

In Section \ref{sechigherkGCO}, we discuss the generalization of GCOs beyond $k=3$; having moved from $k=2$ color orderings to $k=3$ GCOs, the further step to generalize to $k=4$ and beyond is relatively straightforward: arrangements of projective lines in $\mathbb{RP}^{2}$ are replaced by arrangements of $\mathbb{RP}^{k-2}$'s in $\mathbb{RP}^{k-1}$. We will also discuss GCOs for general $k$, their duality, their GFDs and their partial amplitudes in the later part of this paper.

There are of course new phenomena for higher $k$ GCOs and GFDs, which requires exploration. A particular one concerns relations between GCOs and GFDs for different values of $k$. For example, in \cite{Cachazo:2019xjx}, making use of the property that every column or row of a $k=4$ planar matrix of Feynman diagrams, i.e., GFDs contributing to a $k=4$ type 0 GCO, must be a $k=3$ planar arrangement of metric trees, a new bootstrap algorithm was developed to find all such GFDs.

Our goal in this paper is to find all GFDs needed to give the combinatorial construction of the partial amplitude for an arbitrary pair of GCOs. A companion paper \cite{Cachazo:2023ltw} is devoted to the development of a method to find all integrands needed in the CEGM integral in order to produce any such partial amplitude. In that paper, we verify that partial amplitudes obtained using both formulations match for $(k,n)=(3,6), (3,7), (4,7)$ and $(3,8)$.

The rest of this paper is organized as follows: Section \ref{sec2} is a review of the standard color decomposition but with a slightly new version of a color factor. Section \ref{sectionordering} defines $(3,n)$ generalized color orderings while Section \ref{sec4} defines $(3,n)$ generalized Feynman diagrams. Combining the results from Sections \ref{sectionordering} and \ref{sec4}, in Section \ref{three}, we introduce color dressed amplitudes. Section \ref{sec6} contains some properties of color ordering and illustrative examples.
In Section \ref{sec7}, we discuss the generalization of the ${\rm U}(1)$ decoupling identities. In Section \ref{newsec} and \ref{poletoGCO}, we show how to bootstrap all GFDs for each color ordering and vice versa. In
Section \ref{sec8}, we discuss the generalization of the single scalar field $\phi^3$ theory.
We generalize the GCOs and GFDs for higher $k$ in
Sections~\ref{sechigherkGCO} and~\ref{sechigherkGFD}, respectively. We end with future directions in Section \ref{sec10}, where we introduce a new family of objects called {\it chirotopal tropical Grassmannians}. The positive tropical Grassmannian is the simplest member of the family.

Most data is presented either in the appendices or in an ancillary file.

\section{Standard color decomposition}\label{sec2}

Tree-level scattering amplitudes of gluons in ${\rm SU}(N)$ Yang--Mills theory can be color decomposed into partial amplitudes as (see \cite{Mangano:1990by} for a review)
\begin{equation}\label{coYMQ}
 {\mathcal A}_n(\{ k_i,\epsilon_i,a_i\} ) = \sum_{\sigma \in S_{n}/{\mathbb{Z}_n}} \operatorname{tr}\left( T^{a_{\sigma (1)}}T^{a_{\sigma (2)}} \cdots T^{a_{\sigma (n)}}\right)A(\sigma(1),\sigma(2),\dots ,\sigma(n)).
\end{equation}
Partial amplitudes have two important properties,
\begin{align*}
A(1,2,\dots ,n-1, n) ={}& A(n,1,\dots ,n-2,n-1),\\
 A(1,2,\dots ,n-1,n) ={}& (-1)^n A(n,n-1,\dots ,2,1).
\end{align*}
The sum in \eqref{coYMQ} is over cyclic orderings and therefore any given partial amplitude appears twice. This motivates the following definitions.

\begin{defn}\label{stdCO}
A {\it color ordering} is an equivalence class of n-tuples $(\sigma(1),\sigma(2),\dots, \sigma(n))$ with $\sigma\in S_n$ such that
\begin{align*}
 (\sigma(1),\sigma(2),\dots, \sigma(n))\sim (\sigma(n),\sigma(1),\dots, \sigma(n-1)),\\ (\sigma(1),\sigma(2),\dots, \sigma(n))\sim (\sigma(n),\sigma(n-1),\dots, \sigma(1)).
\end{align*}
\end{defn}

In the following, we choose a canonical representative to have $\sigma(1)=1$ and $\sigma(2)< \sigma(n)$.

\begin{defn}\label{stdCF}
Given a color ordering $(\sigma(1),\sigma(2),\dots, \sigma(n))$ in its canonical form, its associated {\it color factor} is
\begin{equation*}
 c(\sigma ):= {\rm tr}(T^{a_{\sigma (1)}}T^{a_{\sigma(2)}}\cdots T^{a_{\sigma(n)}}) + (-1)^n {\rm tr}( T^{a_{\sigma (n)}}T^{a_{\sigma(n-1)}}\cdots T^{a_{\sigma(1)}}).
\end{equation*}
	\end{defn}

Note that there are $(n-1)!/2$ such color factors and their orderings $\sigma$ are called {\it planar orderings}. Also, when $n$ is even, the color factor is independent of the representative $\sigma\in S_n$, but when $n$ is odd the color factor can differ by a sign and hence the need to define it using the canonical representative.

In terms of these color factors, \eqref{coYM} and \eqref{coBS} can be written as
\begin{equation}\label{coYM2}
 {\mathcal A}_n(\{ k_i,\epsilon_i,a_i\} ) = \sum_{\sigma \in S_{n}/{\mathbb{Z}_n\times Z_2}} c(\sigma(1),\sigma(2),\dots ,\sigma(n) )A(\sigma(1),\sigma(2),\dots ,\sigma(n))
\end{equation}
and
\begin{align}
 {\mathcal M}_n(\{ k_i,a_i,{\tilde a}_i\} ) ={}& \sum_{\alpha,\beta \in S_{n}/{\mathbb{Z}_n\times Z_2}}c(\alpha(1),\alpha(2),\dots ,\alpha(n))\nonumber\\
 &\times c(\beta(1),\beta(2),\dots ,\beta(n)) m_n(\alpha,\beta).\label{coBS2}
\end{align}

The definition of $m_n(\alpha,\beta)$ in terms of Feynman diagrams is simply given by \cite{Cachazo:2013iea}
\begin{equation*}
 m_n(\alpha,\beta) = (-1)^{w(\alpha, \beta)}\sum_{{\rm T}\in {\mathcal O}(\alpha)\cap {\mathcal O}(\beta)} \frac{1}{\prod_{e\in {\rm T}}s_e} ,
\end{equation*}
where $w(\alpha, \beta)$ is an integer, the winding number, that depends only on the number of relative descents in the pair of cycles \cite{Mafra:2016ltu},
${\mathcal O}(\gamma)$ is the set of all trees that are planar with respect to $\gamma$, the product is over all internal edges $e$ of the tree $T$ and $s_e$ is the standard kinematic invariant associated to a propagator.

\section{Generalized color orderings}\label{sectionordering}

A standard color ordering, $\sigma$, admits a simple pictorial representation: Drawing the points $\{ \sigma(1),\sigma(2), \dots ,\sigma(n) \}$ on the boundary of a disk makes the dihedral symmetry\footnote{Up to a sign when $n$ is odd.} of the color factor $c(\sigma )$ manifest. An equivalent description of a $(2,n)$ color ordering is as an arrangement of~$n$ points \big(or $\mathbb{RP}^0$'s\big) on $\mathbb{RP}^1$.
Our proposal for a $(3,n)$ generalized color ordering is simply given by an arrangement of $n$ lines \big(or $\mathbb{RP}^1$\big) on $\mathbb{RP}^2$. An equivalent definition which parallels more closely that of generalized Feynman diagrams is the following.

\begin{defn}\label{GCO}
A $(k,n)\!=\!(3,n)$ generalized {\it color ordering} is an n-tuple $\Sigma \!=\! \big( \sigma^{(1)},\sigma^{(2)}, \dots ,\sigma^{(n)}\! \big)$, of $(k=2,n-1)$ color orderings such that there exists an arrangement of $n$ lines on $\mathbb{RP}^2$, $\{ L_1,L_2,\dots ,L_n \}$, so that $\sigma^{(i)}$ is the $(2,n-1)$ color ordering on line $L_i$ defined by the points $L_j\cap L_i$ for $j\in [n]\setminus \{i\}$.
\end{defn}

Any $(3,n)$ color ordering, $\Sigma$, admits $n$ natural projections to the set of $(3,n-1)$ color orderings by simply deleting one line from the arrangement. Let us define the projections as
\be
\label{colorprojection}
\pi_i(\Sigma ) = \big( \pi_i\big(\sigma^{(1)}\big),\pi_i\big(\sigma^{(2)}\big), \dots ,\hat\sigma^{(i)},\dots ,\pi_i\big(\sigma^{(n)}\big) \big),
\ee
where $\hat\sigma^{(i)}$ indicates that this ordering is removed, while the projections $\pi_i$ on a $k=2$ ordering remove the label $i$ from the set regardless of its position. For example,
\begin{align*}
 & \pi_3(((23645),(13465),(12456),(15326),(12634),(13524))) \\ & \qquad = ((2645),(1465),(1526),(1264),(1524)).
\end{align*}

It is tempting to think that the projections could lead to a purely combinatorial recursive definition of $k=3$ color orderings. However, as shown in Section~\ref{colorrecursion}, starting at $n=9$, there are examples of collections of $k=2$ orderings with valid projections but which do not correspond to any arrangement of lines.

Before ending this section, we note that there is a set of generalized color orderings with a~very special property.

\begin{defn}\label{globalGCO}
A $(3,n)$ generalized color ordering $\Sigma = \big( \sigma^{(1)},\sigma^{(2)}, \dots ,\sigma^{(n)} \big)$ is said to descend from a $(2,n)$ color ordering $\sigma$ if $\sigma^{(i)}=\pi_i(\sigma)$. We also say that $\Sigma$ is a descendant of $\sigma$.
\end{defn}

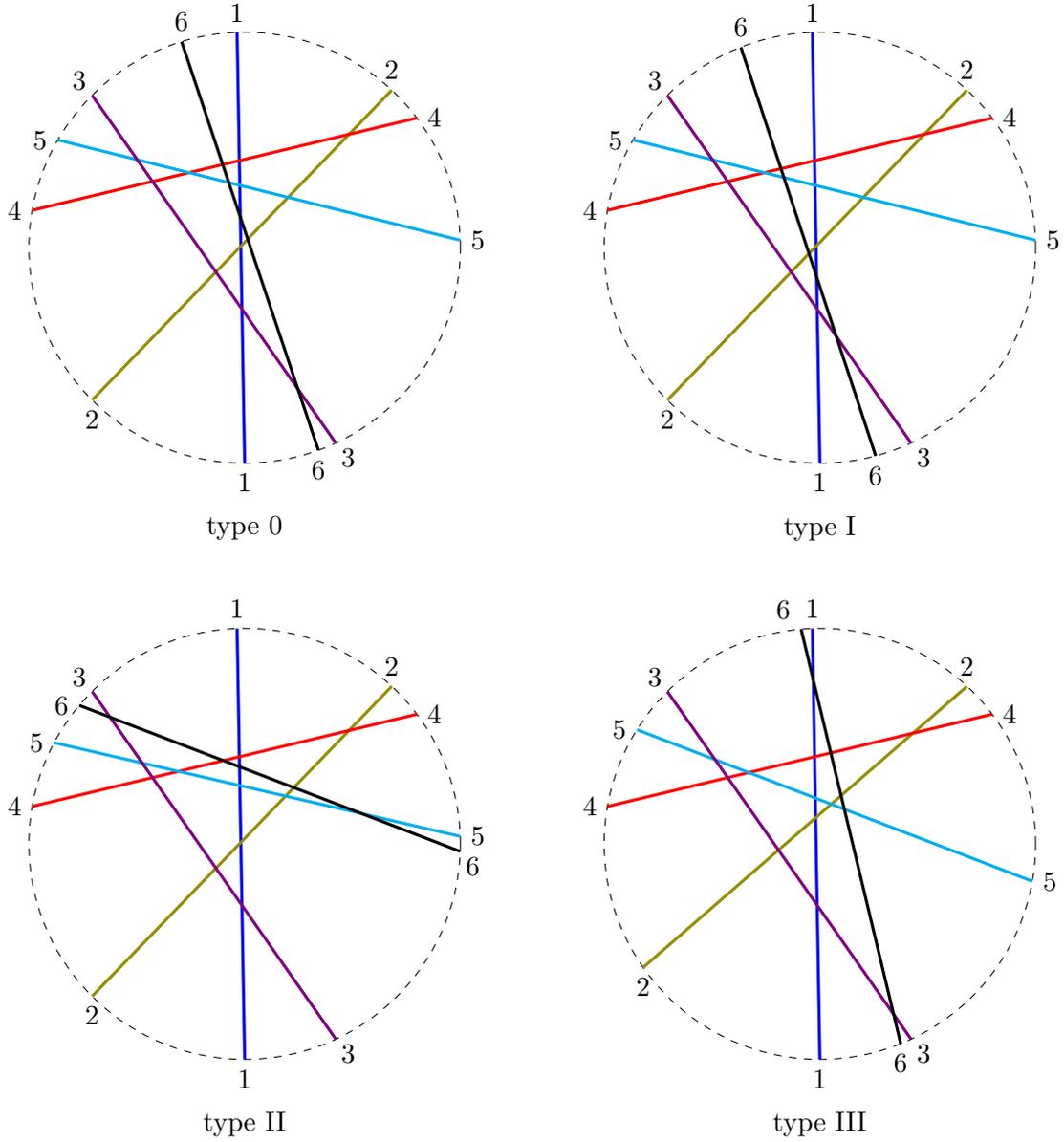
\begin{figure}[!ht]\centering
 \begin{tikzpicture}

\begin{scope}[xshift=-4cm,yshift=4cm,scale=3]

 \draw[ dashed](0,0) circle (1);
 \draw[blue,very thick] (90+2:1) node [above]{\color{black}1} -- (-90:1) node [below]{\color{black}1} ;

 \draw[olive,very thick] (45+2:1) node [above]{\color{black}2} -- (-115-20:1) node [below]{\color{black}2} ;

 \draw[red,very thick] (35+2:1) node [right]{\color{black}4} -- (-180-10:1) node [left]{\color{black}4} ;

 \draw[cyan,very thick] (0+2:1) node [right]{\color{black}5} -- (-180-30:1) node [left]{\color{black}5} ;

 \draw[violet,very thick] (-65:1) node [below right=-2pt]{\color{black}3} -- (-180-45:1) node [ above left=-2pt]{\color{black}3} ;

 \draw[black,very thick] (110-3:1) node [above]{\color{black}6} -- (-70:1) node [below]{\color{black}6} ;

 \node at (-90:1.3) {type 0};

\end{scope}

\begin{scope}[xshift=4cm, yshift=4cm,scale=3]
 \draw[ dashed](0,0) circle (1);
 \draw[blue,very thick] (90+2:1) node [above]{\color{black}1} -- (-90:1) node [below]{\color{black}1} ;

 \draw[olive,very thick] (45+2:1) node [above]{\color{black}2} -- (-115-20:1) node [below]{\color{black}2} ;

 \draw[red,very thick] (35+2:1) node [right]{\color{black}4} -- (-180-10:1) node [left]{\color{black}4} ;

 \draw[cyan,very thick] (0+2:1) node [right]{\color{black}5} -- (-180-30:1) node [left]{\color{black}5} ;

 \draw[violet,very thick] (-65:1) node [below right=-2pt]{\color{black}3} -- (-180-45:1) node [ above left=-2pt]{\color{black}3} ;

 \draw[black,very thick] (110+1.5:1) node [above]{\color{black}6} -- (-75:1) node [below]{\color{black}6} ;

 \node at (-90:1.3) {type I};
\end{scope}

\begin{scope}[xshift=-4cm,yshift=-4.3cm, scale=3]

 \draw[ dashed](0,0) circle (1);
 \draw[blue,very thick] (90+2:1) node [above]{\color{black}1} -- (-90:1) node [below]{\color{black}1} ;

 \draw[olive,very thick] (45+2:1) node [above]{\color{black}2} -- (-115-20:1) node [below]{\color{black}2} ;

 \draw[red,very thick] (35+2:1) node [right]{\color{black}4} -- (-180-10:1) node [left]{\color{black}4} ;

 \draw[cyan,very thick] (0+2:1) node [right]{\color{black}5} -- (-180-30+2:1) node [left]{\color{black}5} ;

 \draw[violet,very thick] (-65:1) node [below right=-2pt]{\color{black}3} -- (-180-45:1) node [ above left=-2pt]{\color{black}3} ;

 \draw[black,very thick] (-2:1) node [below right=-2pt]{\color{black}6} -- (-180-40:1) node [left]{\color{black}6} ;

 \node at (-90:1.3) {type II};
\end{scope}

\begin{scope}[xshift=4cm,yshift=-4.3cm, scale=3]
 \draw[ dashed](0,0) circle (1);
 \draw[blue,very thick] (90+2:1) node [above]{\color{black}1} -- (-90:1) node [below]{\color{black}1} ;

 \draw[olive,very thick] (45+2:1) node [above]{\color{black}2} -- (-115-30:1) node [below]{\color{black}2} ;

 \draw[red,very thick] (35+2:1) node [right]{\color{black}4} -- (-180-10:1) node [left]{\color{black}4} ;

 \draw[cyan,very thick] (0-10:1) node [right]{\color{black}5} -- (-180-30-2:1) node [left]{\color{black}5} ;

 \draw[violet,very thick] (-65:1) node [below right=-2pt]{\color{black}3} -- (-180-45:1) node [ above left=-2pt]{\color{black}3} ;

 \draw[black,very thick] (-65-3:1) node [below=0pt]{\color{black}6} -- (90+5:1) node [above left]{\color{black}6} ;

 \node at (-90:1.3) {type III};

\end{scope}

\end{tikzpicture}

\caption{Representatives of arrangements of lines of different types for $(3,6)$.}	\label{36figure}
\end{figure}

Figure \ref{36figure} shows four examples of color orderings for $(3,6)$ which cannot be related to each other by relabeling (see \cite[Chapter 8]{yoshida2013hypergeometric} for a connection to configuration spaces). This is the first feature in which $k=3$ color orderings differ from $k=2$ ones. Reading the $k=2$ color orderings associated to each line in the arrangements gives the four $k=3$ color orderings in Table~\ref{table36}. Only the first one of the four is a color ordering which descends from a $k=2$ ordering (in this case it is $\sigma = (123456)$).

The $k=3$ color orderings in Table~\ref{table36} can be treated as representatives of four types of $(3,6)$ orderings. Using relabelling, they give rise to all the $(3,6)$ color orderings, totaling $372$. In Section \ref{sec6}, we explain how to find such representatives.
The main tool is a recursive procedure in Section~\ref{colorrecursion} which can also be applied to higher points.

\begin{table}[!htb]\renewcommand{\arraystretch}{1.2}
 \centering
\begin{tabular}{|c|c|c|c|}\hline
Type & Color ordering representative & \# \\ \hline
0 & $((2 3 4 5 6),(1 3 4 5 6),(1 2 4 5 6),(1 2 3 5 6),(1 2 3 4 6),(1 2 3 4 5))$ & 60 \\
I & $((2 5 4 3 6),(1 5 4 3 6),(1 2 4 5 6),(1 2 3 5 6),(1 2 3 4 6),(1 2 5 4 3))$ & 180 \\
II & $((2 3465 ),(1 3465 ),(1 2 4 5 6),(1 2 3 5 6),(1 2 6 3 4),(1 2 5 3 4))$ & 120\\
III & $((2 3645 ),(1 3465 ),(1 2 4 5 6),(1 5 3 2 6),(1 2 6 3 4),(1 3524 ))$ & 12 \\ \hline
\end{tabular}
 \caption{All four types of $(3,6)$ color orderings. 
 The second column provides a representative that can be used to obtain the rest by applying permutations of labels. The last column contains the number of distinct permutations. \label{table36}}
\end{table}

\section{Generalized Feynman diagrams}\label{sec4}

In \cite{Borges:2019csl}, the notion of a generalized Feynman diagram was introduced for $k=3$, building on arrangements of metric trees \cite{herrmann2008draw}, as a collection of Feynman diagrams. The construction in \cite{Borges:2019csl} focused on collections of trees that satisfy a special notion of planarity which, in the terminology introduced in the previous section, turns out to correspond to descendants of $k=2$ color orderings.

Let us review the construction starting with that of a standard Feynman diagram in a biadjoint $\phi^3$ theory. A Feynman diagram is a pair, consisting of a graph together with given kinematic data; the Feynman rules assign to that pair a function. We are interested in tree diagrams with $n$ external vertices (or leaves) and trivalent internal vertices.

The kinematic data consists of a symmetric $n\times n$ matrix, $s_{ab}$, such that its diagonal elements are zero, $s_{aa}=0$, and the sum of its rows vanishes, $\sum_b s_{ab}=0$.

Given a tree, $T$, assign lengths (i.e., positive real numbers) to each of its edges so that $e_a$ denotes the length of the $a^{\rm th}$ external edge and $f_I$ the length of the $I^{\rm th}$ internal edge. Also, denote the minimal distance from leaf $a$ to leaf $b$ by $d_{ab}$. A tree $T$ with a ``metric'' $d_{ab}$ is called a metric tree.

Now consider the integral
\be\label{Scw2}
\int_{\mathbb{R}^+} \prod_{I=1}^{n-3}{\rm d}f_I \exp(-\sum_{a<b}s_{ab}\, d_{ab})=\int_{\mathbb{R}^+} \prod_{I=1}^{n-3}{\rm d}f_I \exp(-\sum_{I}q_I f_I) =\prod_{I=1}^{n-3}\frac{1}{q_I} ,
\ee
where an internal edge $I$ partitions the set of leaves into $L_I\cup R_I = [n]$ and $q_I = \sum_{a\in L_I,\, b\in R_I} s_{ab}.$
It is easy to show that $\sum_{a<b}s_{ab}\, d_{ab} = \sum_{I}q_I f_I $ by using $\sum_b s_{ab}=0$. In particular, note that all external lengths $e_a$ drop out. Also, the middle expression in \eqref{Scw2} is the Laplace transform of the space of internal lengths of the graph, which coincides with the Schwinger parameter representation of Feynman propagators. While the integral representation is only valid for~${q_I>0}$, the rational function on the left defines the value of the Feynman diagram for any $q_I\neq 0$.

Let us now turn to $k=3$ Feynman diagrams. The space of $k=3$ kinematic invariants is given by rank 3 symmetric tensors $\s_{abc}$ such that $\s_{aab}=0$ and $\sum_{bc}\s_{abc}=0$.

\begin{defn}[\cite{herrmann2008draw}]\label{metricTreeArrangement}
An arrangement of metric trees is an n-tuple ${\cal T} = ( T_1,T_2,\dots ,T_n)$ such that $T_i$ is a metric tree with $n-1$ leaves in the set $[n]\setminus \{ i \}$ and metric \smash{$d^{(i)}_{ab}$} so that the following compatibility condition is satisfied
\begin{gather*}
d^{(a)}_{bc} = d^{(b)}_{ac} =d^{(c)}_{ab}, \qquad \forall\{a,b,c\}\subset [n].
\end{gather*}
Denote by $d$ the symmetric tensor with entries $d_{abc}:=d^{(a)}_{bc}$.
\end{defn}

One property with special physical significance is that any arrangement of metric trees admits~$n$ natural projections to arrangements of metric trees with one less leaf. Not surprisingly, the definition follows that for color orderings, i.e.,
\be
\label{projectioncFD}
\pi_i({\cal T}) = \big( \pi_i(T_1),\pi_i(T_2),\dots ,\hat{T_i} ,\dots ,\pi_i(T_n) \big),
\ee
where $\hat{T_i}$ indicates that the $i^{\rm th}$ tree is removed and $\pi_i(T_j)$ means the tree obtained from $T_j$ by pruning the leaf with the label $i$.

Once again, it is tempting to use the recursive property as a way to find arrangements of metric trees. However, not all arrangements of trees satisfying the recursive property admit a~non-degenerate metric. The first example is for $n=9$. Nevertheless, checking that a metric exists is easy since it only involves solving a set of linear equations. It is not obvious that this does not happen for $n<9$ but an exhaustive search shows that to be the case (for a discussion which uses the connection to tropical geometry see \cite{herrmann2008draw}).

In parallel to the discussion of arrangements of lines (see Definition \ref{globalGCO}), we note that there is also a family of arrangements of metric trees with a very special property.

\begin{defn}\label{globalGFD}
A $(3,n)$ arrangement of metric trees ${\cal T} = ( T_1,T_2,\dots ,T_n )$ is said to descend from a metric $T$ if $T^{(i)}=\pi_i(T)$. We also say that ${\cal T}$ is a descendant of $T$.
\end{defn}

Now we are ready to present a very elementary definition of $k=3$ generalized Feynman diagrams for $n<9$. Higher values of $n$ require a more technical definition and it is outside the scope of this work.

\begin{defn}\label{defGFD}
A $(k=3, n<9)$ generalized Feynman diagram is a pair, consisting of given kinematic data, together with an arrangement of metric trees
${\cal T} = ( T_1,T_2,\dots ,T_n )$ that satisfies the following two properties:
\begin{itemize}\itemsep=0pt
 \item There exists at least one generalized color ordering \smash{$\Sigma = \big( \sigma^{(1)},\sigma^{(2)}, \dots ,\sigma^{(n)} \big)$} such that $T_i$ is planar with respect to $\sigma^{(i)}$ for all $i \in [n]$. In this case we say that ${\cal T}$ is compatible with~$\Sigma$.
 \item The arrangement ${\cal T}$ has exactly $2(n-4)$ independent internal edge lengths.\footnote{For higher values of $n$, we suspect that this has to be replaced by the requirement that the metric $d_{abc}$ associated with ${\cal T}$ defines a cone in the tropical Grassmannian $\operatorname{Tr} G(3,n)$.}
\end{itemize}
\end{defn}

Moreover, the rational function associated to ${\cal T}$ is
\be\label{Scw3}
{\cal R}({\cal T}) := \int_{\mathbb{R}^+} \prod_{I=1}^{2(n-4)}{\rm d}f_I \prod_{J=2(n-4)+1}^{n(n-4)}\theta (f_J(f_1,\dots ,f_{2(n-4)})) \exp \bigg(-\sum_{a<b<c}\s_{abc}  d_{abc}\bigg).
\ee
Here $\theta(x):=1$ if $x > 0$ and $\theta(x)=0$ if $x\leq 0$ and $f_I$ with $I\in [2(n-4)]$ are the internal lengths chosen to parameterize all other internal lengths. The conditions in the integrand, $\theta (f_J(f_1,\dots ,f_{2(n-4)}))$, simply enforce that all internal lengths must be non-negative.

Note that the definition does not restrict the kind of trees that participate in a collection. For $(k,n)=(3,6)$, all GFDs are collections of trees with only degree-three internal vertices but starting at $(k,n)=(3,7)$ there can be trees with mixed kinds of internal vertices.

It is important to point out that the notions of descendant color ordering (see Definition~\ref{globalGCO}) and descendant generalized Feynman diagram (using Definition~\ref{globalGFD} for its arrangement) are independent. In other words, there are descendant generalized Feynman diagrams which are compatible with non-descendant color orderings and non-descendant generalized Feynman diagrams which are compatible with descendant color orderings. In fact, most of the GFDs studied in~\cite{Borges:2019csl} are examples of the latter.

Let us end this section with examples that illustrate the Definition~\ref{defGFD}.

In \cite[Section 3.2]{Borges:2019csl}, several GFDs compatible with the $(3,7)$ color ordering which descends from the canonical ordering $(1234567)$ were presented. All of them are collections of seven tree-diagrams with $2(7-4)=6$ independent internal lengths but which evaluate to rational functions~${\cal R}({\cal T})$ with different numbers of poles. We reproduce here the example with seven poles.
 \def\fFD #1,#2,#3,#4,#5,#6 \fFD 	{
 \raisebox{-.85 cm}{
\tikz[scale=.3]{
\draw[]
(0,0)--(3,0)
(0,0)--(-.3,-1.3) node[below]{#1}
(0,0)--(-.3,1.3) node[above]{#2}
(1,0)-- (1,1.3) node[above]{#3}
(2,0)-- (2,1.3) node[above]{#4}
(3,0)--(3.3,1.3) node[above]{#5}
(3,0)--(3.3,-1.3) node[below]{#6}
;
}}
}
\begin{align*}
{\cal T} = \left(
\fFD 3,4,2,5,6,7 \fFD,\fFD 2,4,1,5,6,7 \fFD,\fFD 1,2,4,5,6,7 \fFD,\fFD 1,2,3,7,5,6 \fFD,\fFD 1,2,3,7,4,6 \fFD,\fFD 1,2,7,3,4,5 \fFD,\fFD 1,2,6,3,4,5 \fFD\right).
\end{align*}
If the internal lengths of each tree diagram in the collection are ordered from left to right, then their expressions can be recorded in a $3\times 7$ matrix with $i^{\rm th}$ column \smash{$\big[f_1^{(i)},f_2^{(i)},f_3^{(i)}\big]^{\mathsf{T}}$},
\begin{gather*}
\left[
\begin{matrix}
 x & x & y & y & x+y & z & z \\
 w & w & w+x & p+w+x & p & q & q \\
 u & u & p & v & v & p+v & p
\end{matrix}
\right],
\end{gather*}
with $p+x+y= u+ z$, $q+z= x + y$. Selecting any six independent variables, a straightforward computation using \eqref{Scw3} gives,
\[
{\cal R}({\cal T}) = \frac{\sfW_{1234567}+\t_{34567}}{\s_{456}\t_{1234}\t_{34567}\sfR_{45,67,123}\sfR_{67,12,345}\sfR_{43,21,765}\sfW_{1234567}},
\]
with
\begin{gather}
\t_{A}=\sum_{\{a,b,c\}\in \binom{A}{3}}\s_{abc}, \qquad
\sfR_{ab,cd,efg} = \t_{abefg}+\s_{abc}+\s_{abd},\nonumber\\
 \sfW_{abcdefg} = \t_{abcd} + \t_{fgab} + \s_{abe}.\label{defpoles}
\end{gather}
For later convenience, we define $\sfR_{ab,cd,ef} = \t_{abef}+\s_{abc}+\s_{abd}$.
Let us now give an example of an arrangement of trees in which all trees have degree-three internal vertices and yet it has seven independent internal lengths; because it has more than six independent parameters it is not a~valid generalized Feynman diagram,
 \def\fFD #1,#2,#3,#4,#5,#6 \fFD 	{
 \raisebox{-.65 cm}{
\tikz[scale=.3]{
\draw[]
(0,0)--(90:1)--++(90+45:.7) node[left=-2pt]{#1}
(0,0)--(90:1)--++(90-45:.7) node[right=-2pt]{#2}
(0,0)--(90-120:1)--++(90-120+45:.7) node[right=-2pt]{#3}
(0,0)--(90-120:1)--++(90-120-45:.7) node[below=-2pt]{#4}
(0,0)--(90+120:1)--++(90+120+45:.7) node[below=-2pt]{#5}
(0,0)--(90+120:1)--++(90+120-45:.7) node[left=-2pt]{#6}
;
}}}
\begin{align}
\left(\!\!
 \fFD 3, 7, 4, 5, 6, 2 \fFD,\!\fFD 3, 4, 5, 7, 6, 1 \fFD,\!\fFD 1, 7, 2, 4, 5, 6 \fFD,\!\fFD 7, 6, 5, 1, 2, 3 \fFD,\!\fFD 2, 7, 3, 6, 4, 1 \fFD,\!
\fFD 4, 7, 1, 2, 3, 5 \fFD,\!\fFD 3, 1, 2, 5, 4, 6 \fFD\!\!
\right).\label{37T0}
\end{align}
If internal lengths in each snowflake diagram are labeled clockwise starting with the ``vertical'' edge, then the $3\times 7$ matrix is
\be\label{metricT0}
\left[
\begin{matrix}
 x & w & x & q & r & q & x \\
 y & r & w & y & p & z & r \\
 z & z & p & w & y & p & q \\
\end{matrix}
\right] .
\ee
We have checked that this arrangement of metric trees is not compatible with any $(3,7)$ color orderings and so it fails both conditions in Definition \ref{defGFD}. The arrangement of metric trees \eqref{37T0} was presented in \cite{herrmann2008draw}.\footnote{The data in \cite{herrmann2008draw} is stored at: \url{www.uni-math.gwdg.de/jensen/Research/G3_7/grassmann3_7.html}.
 It is important to mention that the Groebner cone data for $T0$ in their first table is not directly for $T0$ but for one of its co-dimension one boundaries.
}

The next example is a valid GFD given by an arrangement of mixed trees, i.e., trees with both degree-three and degree-four vertices,
\def\fFD #1,#2,#3,#4,#5,#6 \fFD 	{
 \raisebox{-.65 cm}{
\tikz[scale=.3]{
\draw[]
(0,0)--(90:1)--++(90+45:.7) node[left=-2pt]{#1}
(0,0)--(90:1)--++(90-45:.7) node[right=-2pt]{#2}
(0,0)--(90-120:1)--++(90-120+45:.7) node[right=-2pt]{#3}
(0,0)--(90-120:1)--++(90-120-45:.7) node[below=-2pt]{#4}
(0,0)--(90+120:1)--++(90+120+45:.7) node[below=-2pt]{#5}
(0,0)--(90+120:1)--++(90+120-45:.7) node[left=-2pt]{#6}
;
}}}
\def\fFDD #1,#2,#3,#4,#5,#6 \fFDD 	{
 \raisebox{-.65 cm}{
\tikz[scale=.3]{
\draw[]
(0,0)--(90:0)--++(90+45:.7) node[left=-2pt]{#1}
(0,0)--(90:0)--++(90-45:.7) node[right=-2pt]{#2}
(0,0)--(90-120:1)--++(90-120+45:.7) node[right=-2pt]{#3}
(0,0)--(90-120:1)--++(90-120-45:.7) node[below=-2pt]{#4}
(0,0)--(90+120:1)--++(90+120+45:.7) node[below=-2pt]{#5}
(0,0)--(90+120:1)--++(90+120-45:.7) node[left=-2pt]{#6}
;
}}}
\begin{align}
&\left(\!\!
 \fFDD 3, 7, 4, 5, 6, 2 \fFDD,\!
 \fFD 3, 4, 5, 7, 6, 1 \fFD,\!
 \fFDD 1, 7, 2, 4, 5, 6 \fFDD,\!
 \fFD 7, 6, 5, 1, 2, 3 \fFD,\!
 \fFD 2, 7, 3, 6, 4, 1 \fFD,\!
 \fFD 4, 7, 1, 2, 3, 5 \fFD,\!
 \fFDD 3, 1, 2, 5, 4, 6 \fFDD\!\!
\right) .\!\!\!\label{boundaryofT0}
\end{align}
Its contribution to the amplitudes is given by $ {1}/({\s_{126} \s_{145} \s_{234} \s_{257} \s_{356} \s_{467}})$.
One can check that this arrangement is in fact one of the seven possible degenerations of the seven-parameter non-GFD given in \eqref{37T0}. In this case, \eqref{boundaryofT0} is obtained by setting $x=0$ in \eqref{metricT0}. The new arrangement is compatible with $64$ GCOs and two of them are given by,
\begin{align}
\label{37GCOquar}
((2 3 4 5 7 6),(1 3 4 5 7 6),(1 2 4 7 5 6),(1 2 3 7 6 5),(1 2 7 3 6 4),(1 2 3 5 4 7),(1 2 5 3 4 6)),
\\
\label{37GCOquar2}
((2 6 3 5 4 7),(1 5 7 3 4 6),(1 4 2 7 6 5),(1 5 3 2 7 6),(1 3 6 2 7 4),(1 2 7 4 3 5),(1 5 2 3 4 6))
.
\end{align}

The final example is a $(3,8)$ arrangement of metric trees which has $2(8-4)=8$ independent internal lengths but it is not compatible with any GCOs and hence it is not a GFD,
\def\fFD #1,#2,#3,#4,#5,#6,#7 \fFD 	{
 \raisebox{-.65 cm}{
\tikz[scale=.3]{
\draw[]
(0,0)--(90:1)--++(90+45:.7) node[left=-2pt]{#1}
(0,0)--(90:1)--++(90-45:.7) node[right=-2pt]{#2}
(0,0)--(90-100:1)--++(90-100+45:.7) node[right=-2pt]{#3}
(0,0)--(90-100:1)--++(90-100-45:.7) node[below=-2pt]{#4}
(0,0)--(90+100:1)--++(90+120+45:.7) node[below=-2pt]{#5}
(0,0)--(90+100:1)--++(90+100-45:.7) node[left=-2pt]{#6}
(0,0)--(-90:1.5) node[below=-2pt]{#7}
;
}}}
\begin{align}
&\left(\!\!\!\!\!
 \fFD 2, 3, 4, 7, 5, 6, 8 \fFD\!,\!\!\!\!\fFD 1, 3, 5, 7, 6, 8, 4 \fFD\!,\!\!\!\!\fFD 1, 2, 4, 6, 7, 8, 5 \fFD\!,\!\!\!\!\fFD 1, 7, 3, 6, 5, 8, 2 \fFD\!, \!\!\!\!
 \fFD 1, 6, 2, 7, 4, 8, 3 \fFD\!,\!\!\!\!\fFD 1, 5, 3, 4, 2, 8, 7 \fFD\!,\!\!\!\!\fFD 1, 4, 2, 5, 3, 8, 6 \fFD\!,\!\!\!\!\fFD 2, 6, 4, 5, 3, 7, 1 \fFD
\!\!\!\right)\!\!.\!\!\!\!\!\label{37T038}
\end{align}

\section{Color dressed generalized biadjoint amplitudes}\label{three}

An amplitude with the complete color structure is called a {\it color dressed} amplitude in the literature. Here we generalize the $k=2$ color dressed biadjoint amplitude presented in \eqref{coBS2} to~$k=3$. Let ${\rm CO}_{3,n}$ denote the set of all $k=3$ color orderings for $n$ labels and ${\cal N}_{3,n}=|{\rm CO}_{3,n}|$ the number of such orderings. A typical element $\Sigma \in {\rm CO}_{3,n}$ is given by,
\begin{equation*}
 \Sigma =\bigl(\bigl(\sigma^{(1)}_{(2)},\sigma^{(1)}_{(3)},\dots,\sigma^{(1)}_{(n)}\bigr),\bigl(\sigma^{(2)}_{(1)},\sigma^{(2)}_{(3)},\dots,\sigma^{(2)}_{(n)}\bigr),\dots ,\bigl(\sigma^{(n)}_{(1)},\sigma^{(n)}_{(2)},\dots,\sigma^{(n)}_{(n-1)}\bigr)\bigr).
\end{equation*}
To each such color ordering, we associate a color factor ${\bf c}(\Sigma)$. In this work, we treat ${\bf c}(\Sigma)$ purely as a bookkeeping device.

\begin{defn}\label{colorDressed}
The color dressed $(k=3,n)$ biadjoint amplitude is
\begin{equation}\label{coBS3Abs}
 {\mathcal M}^{(3)}_n = \sum_{\Sigma,\tilde{\Sigma}\in {\rm CO}_{3,n}} {\bf c}(\Sigma){\bf c}\big(\tilde\Sigma\big) m_n^{(3)}\big(\Sigma,\tilde\Sigma\big),
\end{equation}
where
\begin{equation*}
 m_n^{(3)}\big(\Sigma,\tilde\Sigma\big) = (-1)^{w(\Sigma,\tilde\Sigma )}\sum_{{\mathcal T}\in {\mathcal O}(\Sigma)\cap {\mathcal O}(\tilde\Sigma)} {\mathcal R}(\mathcal T).
\end{equation*}
Here ${\mathcal O}(\Sigma)$ is the set of all GFDs which are compatible with $\Sigma$. The notion of compatibility is simply that if \smash{${\cal T} = \big(T^{(1)}, T^{(2)},\dots, T^{(n)}\big)$}, then we require \smash{$T^{(i)}\in {\cal O}\big(\sigma^{(i)}\big)$} for all $i$.
\end{defn}

Let us comment on the notion of compatibility. Another way to explain it is to say that a~GFD ${\cal T}$ is compatible or planar with respect to a generalized color ordering $\Sigma$, if the $i^{\rm th}$-tree is planar with respect to the $i^{\rm th}$-ordering. This is a local notion of planarity while the one introduced in \cite{Borges:2019csl} is global.

We do not provide a definition for the sign function here. We suspect that an explicit realization of color factors might be required to obtain a consistent definition.

We have obtained all GFDs for $(3,6)$, $(3,7)$, and $(3,8)$ and provided a Mathematica notebook as an ancillary file to generate any color ordered $(k=3,n<9)$ biadjoint amplitudes with two arbitrary orderings. In the companion paper \cite{Cachazo:2023ltw}, we found all of their integrands needed in the CEGM integrals to compute the color ordered amplitudes. We have verified the amplitudes obtained from both sides match, which is a strong consistency check both for the GFDs and the integrands. We present some brief examples of amplitudes next and more details of GFDs are given in Section~\ref{newsec}.

\subsection{Examples}

We have computed the full color dressed $(3,6)$ biadjoint amplitude,
\begin{equation*}
 {\mathcal M}^{(3)}_{6} = \sum_{I,J=1}^{372} {\bf c}(\Sigma_I){\bf c}(\Sigma_J)\, m_{6}^{(3)}(\Sigma_I,\Sigma_J).
\end{equation*}
There are $1005$ $(3,6)$ generalized Feynman diagrams and the $372\times 372$ matrix of partial amplitudes \smash{$m_{6}^{(3)}(\Sigma_I,\Sigma_J)$} can be obtained directly by listing the generalized Feynman diagrams compatible with both orderings.

\begin{figure}[th!]
	\centering
	\hspace{0.1in}
 \includegraphics[width=0.8\linewidth]{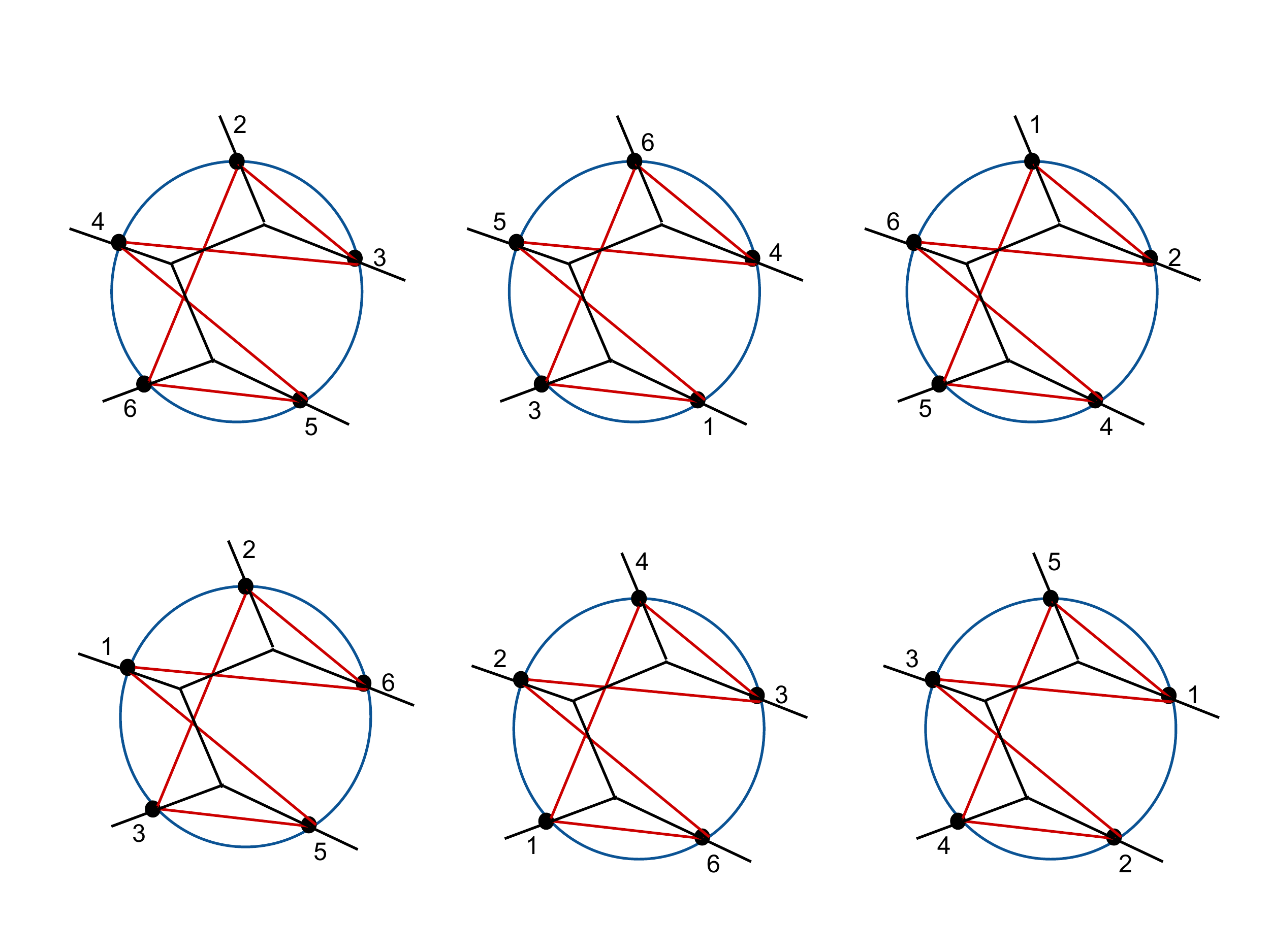}
	\caption{Blue cycles represent the $k=2$ ordering in $\Sigma$, red cycles represent the $k=2$ ordering in $\tilde\Sigma$. Following the CHY diagrammatic description of biadjoint amplitudes, the dual to each red cycle gives rise to the possible Feynman diagrams compatible with both orderings. In this case, there is a single Feynman diagram in each entry and it is drawn in black. They together make up a collection of Feynman diagrams with a non-vanishing metric, i.e., a GFD.}\label{singleGFD}
\end{figure}

Let us start with an example which leads to a single GFD in the set compatible with both orderings.
Considering two color orderings of type I (see Table \ref{table36}),
\begin{align*}
 & \Sigma = ((23564),(13564),(12456),(12653),(12436),(12435)), \\
 &\tilde\Sigma = ((23456),(13645),(12645),(15326),(14326),(13245)),
\end{align*}
the only GFD compatible with both orderings is shown in Figure \ref{singleGFD}. Incidentally, this GFD is not compatible with any type 0 orderings.
The partial amplitude is (see \cite[equation~(4.8)]{Cachazo:2019ngv})
\begin{equation*}
 m_{6}^{(3)}(\Sigma,\tilde\Sigma) = \frac{1}{\s_{123}\s_{345}\s_{561}\s_{246}}.
\end{equation*}
Here and in the remainder of this work we omit the overall sign in \eqref{coBS3Abs}.

Consider another set of two color orderings, this time of type II and III respectively,
\begin{align*}
 \Sigma' ={}& ((23456),(13456),(12465),(12365),(12634),(12534)), \\ \tilde\Sigma' ={}& ((23654),(13564),(12546),(12635),(14326),(13425)).
\end{align*}
There are two GFDs that are compatible with both of these orderings,
\def\fFD #1,#2,#3,#4,#5 \fFD {
 \raisebox{-.61 cm}{
\tikz[scale=.17]{
\draw[]
(0,0)--(2,0)
(0,0)--(-.3,-1.3) node[below]{#1}
(0,0)--(-.3,1.3) node[above]{#2}
(2,0)--(2.3,1.3) node[above]{#3}
(2,0)--(2.3,-1.3) node[below]{#4}
(1,0)-- (1,1.3) node[above]{#5}
;
}}
}
\begin{align*}
&\left(
\fFD2, 3, 4, 5, 6 \fFD,\fFD 1, 3, 5, 6, 4 \fFD,\fFD 1, 2, 4, 6, 5 \fFD,\fFD 3, 6, 1, 5, 2 \fFD,\fFD 2, 6, 1, 4, 3 \fFD,\fFD 2, 5, 3, 4, 1 \fFD
\right)
,
\\
&\left(
\fFD 2, 3, 5, 6, 4 \fFD,\fFD 1, 3, 5, 6, 4 \fFD,\fFD 1, 2, 4, 6, 5 \fFD,\fFD 1, 2, 3, 6, 5 \fFD,\fFD 2, 6, 3, 4, 1 \fFD,\fFD 2, 5, 3, 4, 1 \fFD
\right),
\end{align*}
which gives rise to the partial amplitude,
\begin{equation*}
 m_{6}^{(3)}(\Sigma',\tilde\Sigma') = \frac{1}{\s_{126} \s_{234} \s_{356} \sfR_{34,12,56} }+\frac{1}{\s_{126} \s_{145} \s_{234} \s_{356}}.
\end{equation*}
See the definition of $\sfR$ in \eqref{defpoles}.

In Section \ref{newsec}, we classify all $(3,6)$ GFDs and introduce several operations that connect them.

We end with a $(3,7)$ example. As we have mentioned before, the two GCOs given in \eqref{37GCOquar} and \eqref{37GCOquar2} are both compatible with the generalized Feynman diagram in \eqref{boundaryofT0} which has both cubic and quartic vertices in its trees. In fact, this GFD is the only one that contributes to both GCOs at the same time, resulting in
\begin{equation*}
 m_{7}^{(3)}( \eqref{37GCOquar},\eqref{37GCOquar2}) = \frac{1}{\s_{126} \s_{145} \s_{234} \s_{257} \s_{356} \s_{467}}.
\end{equation*}

\section[Properties of k=3 color orderings and examples]{Properties of $\boldsymbol{k=3}$ color orderings and examples}\label{sec6}

In this section, we study some properties of $k=3$ color orderings as well as bootstrap methods for constructing them. We illustrate the techniques by applying them to $k=3$ with $n\leq 9$.

\subsection[Arrangements of lines vs. arrangements of pseudo-lines]{Arrangements of lines vs.\ arrangements of pseudo-lines}\label{colorrecursion}

While generalized color orderings are constructed out of arrangements of lines on $\mathbb{RP}^2$, this definition is not an effective way to construct them. Instead, we use a result from the theory of uniform oriented matroids which states that if lines are allowed to bend slightly, i.e., become pseudo-lines, then arrangements satisfy a recursive definition (for details see \cite[Chapter 6]{bjorner1999oriented}).

\begin{thm}\label{pseudoLines}
 An $n$-tuple of standard $($or $k=2)$ color orderings such that the $i^{\rm th}$ one is defined on the set $[n]\setminus \{i\}$ can be represented as an arrangement of pseudo-lines on $\mathbb{RP}^2$ if the following holds:
 \begin{itemize}\itemsep=0pt
 \item For $n=5$ the $5$-tuple is one of the $12$ descendants of the $(2,5)$ color orderings.
\item For $n>5$ and for any $j\in [n]$, removing the $j^{\rm th}$ color ordering from the $n$-tuple and then deleting the label $j$ from the $n-1$ remaining ones must produce a $(n-1)$-tuple of color orderings that can be represented by an arrangement of $n-1$ pseudo-lines.
 \end{itemize}
\end{thm}

Note that the second condition is simply the statement that the projection $\pi_j$ defined in \eqref{colorprojection} produces an arrangement of pseudo-lines for all $j\in [n]$.

Let us refer to these as $(3,n)$ pseudo-color orderings. Clearly, the set of pseudo-color orderings contains all $(3,n)$ color orderings. As mentioned in
Section \ref{sectionordering}, the first time one finds a $(3,n)$ pseudo-color ordering which is not a $(3,n)$ color ordering is for $n=9$ \cite{bjorner1999oriented}. This means that we can use the recursive construction of Theorem~\ref{pseudoLines} to search for $(3,n)$-color orderings by constructing all pseudo-color orderings and then discarding the ones that are not valid.

Consider the $n=5$ case. Given that $(3,5)$ generalized biadjoint amplitudes are mapped to~$(2,5)$ amplitudes by simply replacing $\s_{abc}$ with $s_{de}$ with $\{ a,b,c,d,e\} = [5]$, it is not surprising that there exists a bijection between $(3,5)$ and $(2,5)$ color orderings. More explicitly, one can show that each arrangement of five lines in $\mathbb{RP}^2$ is in fact the descendant of a configuration of five points on $\mathbb{RP}^1$. There are $(n-1)!/2 =12$ $(2,5)$ color orderings for $n=5$. The example given in the introduction \eqref{introEx}
\[
\Sigma := \big(\sigma^{(1)},\sigma^{(2)},\sigma^{(3)},\sigma^{(4)},\sigma^{(5)}\big)=((2435),(1435),(1254),(1253),(1234))
\]
is the descendant of $\sigma = (12534)$ since
\begin{align*}
 & \pi_1(\sigma)=(2534)=(2435), \qquad \pi_2(\sigma)=(1534)=(1435),\\
& \pi_3(\sigma)=(1254),\qquad \pi_4(\sigma )=(1253), \qquad \pi_5(\sigma) =(1234).
\end{align*}

The first non-trivial case is $(k,n)=(3,6)$. Here a simple algorithm that starts with an ansatz $\big(\sigma^{(1)},\sigma^{(2)},\dots ,\sigma^{(6)}\big)$, where each $\sigma^{(i)}$ is one of the twelve possible $(2,5)$ orderings with labels in~$[6]\setminus \{ i \}$, and then uses the recursive definition from Theorem~\ref{pseudoLines} to select the $(3,6)$ orderings is fast enough to obtain all $(3,6)$ orderings. We have implemented a slightly more efficient version of this algorithm in \textsc{Mathematica} and found exactly $372$ $(3,6)$ color orderings.

Moreover, the $372$ $(3,6)$ color orderings split into four types modulo relabeling and this is how we obtained the results in Table~\ref{table36}. There are $60$ of type 0, $180$ of type I, $120$ of type II and $12$ of type III.
This means that each type has a symmetry group of order $12$, $4$, $6$ and $60$, respectively.

One advantage of the recursive definition from Theorem~\ref{pseudoLines} is that it is purely combinatorial. Once we get all
(pseudo-)color orderings, they can be turned into figures showing the arrangements of (pseudo-)lines, such as the four figures in Figure \ref{36figure}. In that figure, we manifestly see that all line arrangements reduce to the $(3,5)$ arrangement in Figure~\ref{fiveNaive} when the sixth (black) line is removed.

Note that type 0 has exactly $(6-1)!/2=60$ orderings. This is because all type~0 color orderings are descendants of $(2,6)$-color orderings. This type obviously generalizes to arbitrary~$(3,n)$ where their type~0 has $(n-1)!/2$ generalized color orderings.

In a companion paper \cite{Cachazo:2023ltw}, we explain the connection between these $(3,6)$ color orderings and the $372$ chambers which the space of six points (or six lines) on the real projective plane decomposes into.

\begin{table}[!htb]\renewcommand{\arraystretch}{1.2}
 \centering
\begin{tabular}{|c|c|c|c|c|}\hline
Type & Color ordering representative & \# \\ \hline
0 &
((234567),(134567),(124567),(123567),(123467),(123457),(123456))
& 360
\\
I &
((234567),(134567),(124567),(123567),(123476),(123475),(123465))
& 2520
\\
II &
((234567),(134567),(124567),(123576),(123476),(123745),(123645))
&
5040
\\
III &
((234567),(134567),(124567),(123756),(123746),(123745),(123654))
&
2520
\\
IV &
((234567),(134567),(124576),(123756),(123746),(127345),(126354))
& 2520
\\
V &
((234567),(134576),(124576),(123756),(123746),(154327),(145326))
&
 2520
\\
VI &
((234756),(134576),(124567),(123567),(164327),(127345),(143625))
&
 2520
\\
VII &
((234756),(134756),(124567),(123567),(127346),(127345),(125634))
&
840
\\
VIII &
((234567),(134576),(124576),(123765),(123764),(145327),(145326))
&
1680
\\
IX &
((234567),(134576),(124756),(123765),(127364),(145327),(143526))
& 1680
\\
X &
((234576),(134576),(124756),(123675),(127364),(123574),(125364))
& 5040
\\ \hline
\end{tabular}
 \caption{(3,7) color orderings.
 The last column denotes the number of distinct permutations.}\label{table37}
\end{table}

With the $(3,6)$ color orderings at hand, one can further produce all
$(k,n)=(3,7)$ color orderings using again Theorem~\ref{pseudoLines}. Here however, a brute force search is impractical. Instead, one can start with a given $(3,6)$ color ordering (one each type), and then list all possible ways of adding label $7$ to construct $n=7$ color ordering candidates. Using Theorem \ref{pseudoLines} the valid~$(3,7)$ color orderings are then selected. There are $27\,240$ $(3,7)$ color orderings in total, which split into~$11$ types modulo relabeling. We tally the numbers of their distinct permutations here,
$\{360,1\}$, $\{ 840,1\}$, $\{1680,2\}$, $\{ 2520, 5\}$, $\{5040, 2\}$, e.g., there are five types with $2520$ elements, i.e., with a symmetry group of order $2$. Note that the type 0 GCO has a symmetry group of order $14$, the dihedral group. A representative of each type is given in Table~\ref{table37}. The two GCOs in \eqref{37GCOquar} and \eqref{37GCOquar2} are both of type X.

Further improvements in the algorithm produces $4 \,445\, 640$ $(3,8)$ color orderings, which fall into $135$ types. Their representatives are presented in Appendix~\ref{38colorfactors} whose numbers of distinct permutations are tallied as
$\{2520, 1\}$, $\{2880, 1\}$, $\{5040, 1\}$, $\{10080, 4\}$, $\{20160, 38\}$, $\{40320,
 90\}$.

We also find $4382$ types of $(3,9)$ pseudo-color orderings but one fails to be a color-ordering and so there are $4381$ types, whose representatives are provided in an ancillary file. We also tally their numbers of distinct permutations here,
$\{20160, 1\}$, $\{60480, 6\}$, $\{120960, 24\}$, $\{181440, 158\}$, $\{362880, 4193\}$,
 which in total gives $ 1 \,553\,388\, 480$ $(3,9)$ color ordering.\footnote{Note that up to $(k,n)=(3,9)$, the number of color orderings agrees with the number of uniform matroids over $\mathbb{F}_q$ when continued to $q=-1$, as defined in~\cite{skorobogatov1996number}. This is in contrast to the continuation to $q=1$, which leads to the Euler characteristic of $X(3,n)$ (see \cite[Appendix~A]{Agostini:2021rze}). A connection to the number of realizable oriented uniform matroids is explored in \cite{Cachazo:2023ltw}.}

Finally, let us write down a representative of the type of the $(3,9)$ pseudo color orderings which cannot be represented as arrangements of lines and hence are not $(3,9)$ color orderings (see in \cite[Figure 13]{celaya2020oriented}),
\begin{align*}
&((24738695),(18536749),(14857296), (17385629),(18273469),(15472938),
\\
&\qquad(14853629),(14372596),(15247638)
).
\end{align*}
This type has $120\, 960$ distinct permutations.

\subsection{Connecting color orderings using triangle flips}

Realizing $k=3$ color orderings as arrangements of lines on the projective plane gives rise to figures in which various polygons appear as regions bounded by the lines. If lines $L_i$, $L_j$, $L_k$ bound a triangle then one can deform the arrangement until the triangle shrinks to zero size and can then be opened up in a different configuration (see Figure \ref{triangleFlip}). When these flips are possible, they connect color orderings.
See Figure \ref{36figure} as an explicit example where the two
line arrangements on the top are related by a flip via the triangle bounded by $L_1$, $L_2$ and $L_6$.

\begin{figure}[h!]
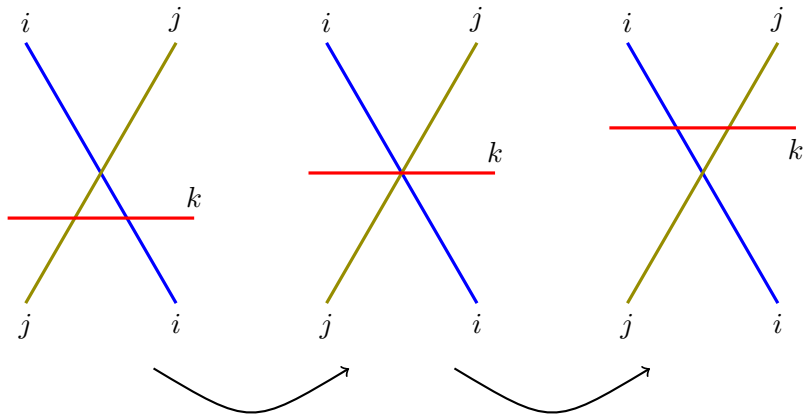

	\centering
	\hspace{0.1in}
 \tikz[scale=2]{

\begin{scope}[xshift=.35cm,yshift=-1.3cm, scale=1.3]
\draw [ thick, ->] (0,0) .. controls (0.5,-0.3) .. (1,0);
 \end{scope}

 \begin{scope}[xshift=2.35cm,yshift=-1.3cm, scale=1.3]
\draw [ thick, ->] (0,0) .. controls (0.5,-0.3) .. (1,0);
 \end{scope}

\begin{scope}[xshift=0cm,yshift=0.cm, scale=1]
 \draw[blue,very thick] (120:1) node [above]{\color{black}$i$} -- (-60:1) node [below]{\color{black}$i$} ;

 \draw[olive,very thick] (60:1) node [above]{\color{black}$j$} -- (-120:1) node [below]{\color{black}$j$} ;

 \draw[red,very thick] (-.62,-.3)
 -- (.62,-.3) node [above]{\color{black}$k$} ;
 \end{scope}

 \begin{scope}[xshift=2cm,yshift=0.cm, scale=1]
 \draw[blue,very thick] (120:1) node [above]{\color{black}$i$} -- (-60:1) node [below]{\color{black}$i$} ;

 \draw[olive,very thick] (60:1) node [above]{\color{black}$j$} -- (-120:1) node [below]{\color{black}$j$} ;

 \draw[red,very thick] (-.62,-.0)
 -- (.62,-0) node [above]{\color{black}$k$} ;
 \end{scope}

 \begin{scope}[xshift=4cm,yshift=0.cm, scale=1]
 \draw[blue,very thick] (120:1) node [above]{\color{black}$i$} -- (-60:1) node [below]{\color{black}$i$} ;

 \draw[olive,very thick] (60:1) node [above]{\color{black}$j$} -- (-120:1) node [below]{\color{black}$j$} ;

 \draw[red,very thick] (-.62,.3)
 -- (.62,.3) node [below]{\color{black}$k$} ;
 \end{scope}
}	
	\caption{Left: Lines $L_i$, $L_j$, $L_k$ bound a triangle with line $L_k$ at the bottom. Center: Line $L_k$ moves up until the triangle becomes a point where all three lines intercept. Right: The three lines bound a~triangle again but with $L_k$ bounding the top.}
	\label{triangleFlip}
\end{figure}

This property can be used to improve algorithms for constructing color ordering. With that in mind, it is convenient to be able to recognize triangles in a color ordering without having to find the arrangement of lines since as $n$ increases the arrangements can become quite complicated.\looseness=1

\begin{claim}
The arrangement of lines associated to a color ordering $\Sigma =\big( \sigma^{(1)},\sigma^{(2)},\dots ,\sigma^{(n)} \big)$ has a~triangle bounded by lines $L_i$, $L_j$, $L_k$ if and only if labels in the sets $\{ i,j\}$, $\{ j,k\} $, and $\{ k,i\}$
are consecutive in $\sigma^{(k)}$, $\sigma^{(i)}$, and $\sigma^{(j)}$, respectively.
\end{claim}

The proof is left as an exercise for the reader.

Not only is it easy to recognize a triangle but it is also simple to perform a flip. However, the result might not be a color ordering but only a pseudo color ordering. The procedure is the following.

Assume that $\Sigma$ has a triangle bounded by lines $L_i$, $L_j$, $L_k$, then a flip sends
\[
\Sigma \longrightarrow \big( \tilde\sigma^{(1)},\tilde\sigma^{(2)},\dots ,\tilde\sigma^{(n)} \big)
\]
with $\tilde\sigma^{(l)}=\sigma^{(l)}$ if $l\notin \{i,j,k\}$, $\tilde\sigma^{(i)} = \sigma^{(i)}|_{j\leftrightarrow k}$, $\tilde\sigma^{(j)} = \sigma^{(j)}|_{i\leftrightarrow k}$, and $\tilde\sigma^{(k)} = \sigma^{(k)}|_{j\leftrightarrow i}$.
One then has to check whether $\big( \tilde\sigma^{(1)},\tilde\sigma^{(2)},\dots ,\tilde\sigma^{(n)} \big)$ can be represented as an arrangement of lines so that it can become a new color ordering $\tilde\Sigma$.

Using triangle flips, one can bootstrap all $(k=3,n)$ color orderings starting with a single one. A convenient choice to start with is a descendant of a $(k=2,n)$ color ordering. We have been able to reproduce all color orderings for $(3,6)$, $(3,7)$, $(3,8)$ and $(3,9)$ derived in Section \ref{colorrecursion} using this technique.

\section{Generalized decoupling identity}\label{sec7}

Color-decomposing amplitudes in the biadjoint scalar theory and in Yang--Mills have the advantage of reducing the complexity of the computation by restricting the set of Feynman diagrams to those that contribute to a given partial amplitude. The price to pay for the simplification is that it obscures some properties of the full amplitude which are simple to verify before performing the decomposition. One such property is the following. When the color of one of the particles, say $T^{a_n}$, is in the commutant of the rest, then the full amplitude trivially vanishes since each color-dressed Feynman diagram contains at least one $su(N)$ structure constant of the form $f_{ijn}$, which vanishes since $[T^{a_n},T^{a_i}]=0$ by assumption. On the other hand, partial amplitudes $A_n(\sigma )$ do not carry individual color information and traces of products of generators do not generically vanish. Instead, some traces that are generically distinct now become identified. The flip side of this is that the vanishing of the full amplitude, ${\cal A}_n$, must then imply identities among partial amplitudes that always hold regardless of the color structure.

The simplest example is when $T^{a_n}=\mathbb{I}$, i.e., the $n^{\rm th}$ particle is a ``photon''. The set of identities so derived is called ${\rm U}(1)$ decoupling identities. Here we treat any $T^{a_n}$ such that $[T^{a_n}, T^{a_i}]=0$ for all $i\in \{ 1,2,\dots ,n-1\}$ as a photon with respect to the rest.

Substituting this into \eqref{coYM2} one has to collect like-terms by noting that, e.g.,
\be\label{k2Prop}
c(1,2,\dots ,i,n,i+1,\dots ,{n-1}) = \pm c(1,2,\dots ,i,i+1,\dots,{n-1},n).
\ee
The sign depends on the parity of $n$ as well as its location in the ordering on the left. In other words, up to a sign, the position of label $n$ in the color factor becomes inconsequential.

Imposing that ${\cal A}_n(\{k_i,\epsilon_i,a_i\})=0$ implies identities such as (see, e.g., \cite{Mangano:1990by})
\be\label{stdU1}
\sum_{i=1}^{n-1} A_n(1,2,\dots ,i,n,i+1,\dots, n-1) = 0.
\ee
Similarly, biadjoint amplitudes satisfy ${\rm U}(1)$-decoupling identities as well,
\[
\sum_{i=1}^{n-1} m_n^{(2)}((1,2,\dots ,i,n,i+1,\dots, n-1),\beta) = 0, \qquad \forall \beta\in {\rm CO}_{2,n},
\]
where ${\rm CO}_{2,n}$ is the set of all $(2,n)$ color orderings.

Note that $m_n^{(2)}(\alpha,\beta)=m_n^{(2)}(\beta,\alpha)$ and therefore we do not get new identities by using the second ordering.

In this section we initiate the study of decoupling identities in generalized biadjoint amplitudes. Here we only scratch the surface since a deeper discussion requires the CEGM formulation and it is presented in the companion paper \cite{Cachazo:2023ltw}.

We do not yet have a satisfactory realization of $k=3$ color factors in terms of Lie algebraic objects, but we can use the $k=2$ case \eqref{k2Prop} as a guide to uncover the decoupling properties of~$k=3$ color factors.

It is clear that decoupling a particle in a $k=3$ generalized amplitude requires making the position of the particle in the $k=2$ orderings in the collection irrelevant. However, this leaves the prescription for what happens to the $k=2$ ordering in the collection which does not contain the particle label undefined. The two natural choices are to either keep the order intact or make it irrelevant. We will study $(3,5)$ and $(3,6)$ to discover that the latter is the correct prescription.

Let us start with the $(3,5)$ case and the prescription where we keep the order intact. Even though it is not the correct prescription for all $n$, it illustrates an important point. Recall that there are exactly $12$ distinct $(3,5)$ color orderings. A decoupling identity is obtained by selecting one, say,
\[
\Sigma_0 = ((2345),(1345),(1245),(1235),(1234)),
\]
and identifying with it any other that differs from $\Sigma_0$ by the position of label $5$ in the first four ordering {\it and} agrees on the fifth. In this case, we find exactly three:
\begin{align*}
& \Sigma_1 = ((2 3 5 4), (1 3 5 4), (1 2 5 4), (1 2 3 5), (1
 2 3 4)) ,\\
& \Sigma_2 = ((2 4 3 5), (1 4 3 5), (1 2 5 4), (1 2 5
 3), (1 2 3 4)) , \\
& \Sigma_3 = ((2 34 5), (1 4 3 5), (1 4 2 5), (1
 3 2 5), (1 2 3 4)) .
\end{align*}
It is easy to see that after making the position of $5$ irrelevant, all three orderings indeed become equal to $\Sigma_0$. We leave it as an exercise for the reader to check that after making the position of~$5$ irrelevant, all $12$ orderings are of the form
\[
((2345),(1345),(1245),(1235),(\bullet )),
\]
and so it is only the ordering in $(\bullet)$ that groups them into three groups of four. Of course, $(\bullet)$ equals one of the three possible $(2,4)$ orderings.

As mentioned in Section \ref{colorrecursion}, all $12$ $k=3$ color orderings are descendants of $(2,5)$ color orderings. This means that we can find a bijection between the two sets. More explicitly,
\[
\Sigma_0 \leftrightarrow (12345), \qquad \Sigma_1 \leftrightarrow (12354), \qquad \Sigma_2 \leftrightarrow (12534), \qquad \Sigma_3 \leftrightarrow (15234)=(14325).
\]
Note that the four $(2,5)$ color orderings in the expression above are exactly the ones appearing in the standard ${\rm U}(1)$ decoupling identity \eqref{stdU1}.

It is clear that had we chosen the second prescription for the decoupling of particle $5$, i.e., in addition to making the position of $5$ irrelevant, also making the choice of $(\bullet)$ irrelevant, we would have found a single $12$-term identity. However, this identity is not irreducible since it is a linear combination of the standard four-term identities.

As it turns out the definition that led to three four-term identities for $(3,5)$ does not generalize to higher points while the one that gives $12$-term identities does.

Here we make a proposal for what the correct $k=3$ color decoupling is and then we apply it to $(3,6)$ and $(3,7)$ amplitudes and discover that indeed it leads to identities among partial amplitudes.

\begin{defn}\label{k3decoupling}
Given a $(3,n)$ color-dressed amplitude, ${\cal M}_n^{(3)}$, the operation that identifies any two $(3,n)$-color orderings, $\Sigma$ and $\tilde\Sigma$, if their $i^{\rm th}$ projections are the same, i.e., $\pi_i (\Sigma ) = \pi_i \big(\tilde\Sigma\big)$, is called decoupling the $i^{\rm th}$-particle.
\end{defn}

Recall that \smash{${\cal M}_n^{(3)}$} is defined in \eqref{coBS3Abs} and the projection operator $\pi_i$ in \eqref{colorprojection}.

This definition has a beautiful combinatorial interpretation given that a $(3,n)$ color ordering is represented by an arrangement of $n$ lines in $\mathbb{RP}^2$. Definition \ref{k3decoupling} implies that line $L_i$ can be freely moved on the plane so that all color orderings generated in the process are identified. Let us refer to the set of these color orderings as a {\it decoupling set} where $L_i$ moves freely.
This is the analog to the $k=2$ interpretation of the ${\rm U}(1)$ decoupling identities in which particle $i$ is free to move on the boundary of the circle defining the ordering.

The decoupling identities have a close relation to the recursive properties of arrangements of pseudo-lines described in
Theorem \ref{pseudoLines} because if we
remove line $L_i$ from the plane, the resulting arrangement of lines still corresponds to a color ordering of lower points according to the theorem. Let us refer to the resulting color ordering as an {\it induced} lower-point GCO, which is useful when classifying the decoupling identities.

\subsection[Decoupling identities for (3,6) color factors]{Decoupling Identities for $\boldsymbol{(3,6)}$ color factors}

Applying Definition \ref{k3decoupling} to $(3,6)$ and decoupling, e.g., particle $6$, we find that the $372$ color orderings are partitioned into $12$ decoupling sets of $31$ orderings each. For example, consider the $k=3$ color ordering that contains the ordering
\[
\Sigma_0 = ((23456),(13456),(12456),(12356),(12346),(12345)).
\]
Note that $\Sigma_0$ belongs to type 0 and it is therefore the descendant of a $k=2$ color factor, in this case, $\sigma= (123456)$. The set of $31$ $(3,6)$ color orderings that contains $\Sigma_0$ is composed of~$5$ type 0,~$15$ type I, $10$ type II and $1$ type III color orderings. See their explicit expressions in Appendix~\ref{31termsec}. It is interesting to note that the $5$ type 0 orderings are descendants of exactly the five $k=2$ orderings that participate in the $k=2$ decoupling identity. This hints the fact that type 0 color orderings are somehow disconnected and that the other types are needed to fill out the ``holes''. In the next section we give more evidence that this picture is correct while in \cite{Cachazo:2023ltw} we show the geometric meaning in terms of the structure of $X(3,6)$ and its fibration over $X(3,5)$ (see \cite[Figure 8.7]{yoshida2013hypergeometric} and \cite[Figure 1]{Cachazo:2019ngv} for two dual descriptions).

Finally, our prescription so far only tells us which color factors to group together but it does not fix the relative signs. From the $(3,5)$ case, we expect that the $31$-term identities are not irreducible. In fact, we find that the set of $372$ equations
\be
\label{decop31terms}
\sum_{\Sigma\in {\cal D}_6} (-1)^{k_\Sigma} m^{(3)}_6\big(\Sigma, \tilde\Sigma \big) = 0, \qquad \forall\, \tilde{\Sigma} \in {\rm CO}_{3,6},
\ee
where ${\cal D}_6$ is the decoupling set of $31$ orderings given in \eqref{31termseceq} has $30$ solutions for $k_\Sigma\in \{0,1\}$ up to an overall rescaling. Every one of the $1005$ $(3,6)$ GFDs either does not appear at all in a decoupling identity \eqref{decop31terms} or appears four times with two positive coefficients of $(-1)^{k_{\Sigma}}$ and two negative ones. In this sense, we say that
these identities hold at the level of GFDs, i.e., the GFDs ${\cal T}$ which contribute to the rational function ${\cal R}({\cal T})$ in the amplitudes cancel pairwise.

In \cite{Cachazo:2023ltw}, we study decoupling identities more deeply and find a geometric interpretation which leads to the irreducible identities that arise from \eqref{decop31terms}.

\subsection[Decoupling identities for (3,7) color factors]{Decoupling identities for $\boldsymbol{(3,7)}$ color factors}

The partition of $(3,6)$ color orderings into $12$ decoupling sets of $31$ orderings each makes it tempting to think that the $27\,240$ $(3,7)$ orderings could also partition evenly. However, the structure we find is much more interesting, revealing that the $27\,240$ GCOs are partitioned inhomogeneously by a decoupling.

One can classify these identities by studying their induced GCOs when decoupling, e.g., particle $7$. There is no doubt that the $27\, 240$ color orderings produce $372$ induced $(3,6)$ color orderings, which themselves are classified into four types as shown in Table~\ref{table36}. Hence, there are four types of decoupling sets for $(3,7)$ as well. As shown in Table~\ref{decoupling37}, it turns out that~$74$~$(3,7)$ color orderings including $6$ of type 0, $18$ of type II, etc., reduce to a $(3,6)$ color ordering of type~0, while $72$ $(3,7)$ color orderings reduce to a $(3,6)$ color ordering of type I, etc. Their explicit expressions are presented in the ancillary file.
In total, the $27\, 240$ color orderings partition into~$372$ decoupling sets including $60$ type 0 decoupling sets of $74$ elements, $180$ type I decoupling sets of $72$ elements, $120$ type II decoupling sets of $74$ elements, and $12$ type III decoupling sets of $80$ elements. We put the representatives of four types of decoupling sets in an ancillary file.\looseness=1

 \begin{table}[!htb]\renewcommand{\arraystretch}{1.2}
 \centering
\noindent
\begin{tabular}{|@{\
\hspace*{13mm}}l||*{11}{c|}|*{1}{c|}}\hline
\multicolumn{1}{|@{}l||}{\backslashbox[0pt][l]{decoupling\!\!\!\!\!\!
\\
~~~~~type}{
\!\!\!\! ordering
\\
type~~}}
&\makebox[1em]{0}&\makebox[1em]{I}&\makebox[1em]{II}
&\makebox[1em]{III}&\makebox[1em]{IV}
&\makebox[1em]{V}
&\makebox[1em]{VI}
&\makebox[1em]{VII}
&\makebox[1em]{VIII}
&\makebox[1em]{IX}
&\makebox[1em]{X}
&\makebox[2em]{Total}
\\\hline\hline
0 & 6& 18& 24& 6& 6& 6& 6& 2& & & & 74 \\\hline
I & & 8& 8& 12& 4& 8& 4& & 8& 4& 16 & 72 \\\hline
II & & & 18& & 12& 6& 6& 6& 2& 6& 18 & 74 \\\hline
III & & & & & & & 60& & & 20& & 80 \\\hline
\end{tabular}
\caption{Four types of decoupling of $(3,7)$. The numbers represent how many $(3,7)$ color orderings of a particular type take part in a certain decoupling. }\label{decoupling37}
\end{table}

As in the $(3,6)$ case, the six color orderings of type 0 which are the descendants of six $(k=2,n=7)$ color orderings that participate in the $k=2$ decoupling identity appear together in a decoupling identity. This time, 68 other color orderings of other types are needed to make up a $(3,7)$ decoupling of type 0.

From Table~\ref{decoupling37}, we also see that a type 0 ordering can only take part in the decoupling of type 0 while a type I ordering can participate in decouplings of both type 0 and I by decoupling different labels. It is worth mentioning that orderings of type VI appear in any type of decoupling while the type III decoupling only contains color orderings of types VI and IX.

A similar analysis reveals 11 types of decoupling identities for $(3,8)$. We do not present them here because, as in previous cases, all of the identities in terms of partial amplitudes they lead to
 are reducible ones and we postpone a more complete discussion to \cite{Cachazo:2023ltw}.

\section[Bootstrapping GFDs via flips compatible with color orderings]{Bootstrapping GFDs via flips compatible \\ with color orderings}\label{newsec}

Listing all generalized Feynman diagrams for any $(k,n)$ is a daunting problem. Even for $k=3$ the number of diagrams grows very fast with $n$. In \cite{Borges:2019csl} and \cite{Cachazo:2019xjx}, a ``bootstrap'' method for producing new GFDs out of old ones using the notion of global planarity was introduced. In other words, their construction was based only on type 0 color orderings, i.e., those that descend from $k=2$ ones. As we will see, there are GFDs that are not compatible with any type 0 color ordering and therefore cannot be generated in that way.

In this section, we extend the bootstrap methods starting from the assumption that we have access to all color orderings. This will allow us to generate all GFDs. We show this for $(3,n)$ with $n\leq 8$.

The approach in this section is the analog of the triangle flip moves explained for generalized color ordering in the previous section. Triangle flips were one of the two techniques for generating color orderings, the other one was based on a recursive property. The analog for GFDs is based on the recursive property \eqref{projectioncFD} but we find the flip moves to be the more efficient option.

\subsection[Flips in k=2 Feynman diagrams]{Flips in $\boldsymbol{k=2}$ Feynman diagrams}

In order to introduce the main idea, let us start with the trivial but illustrative case of standard biadjoint $\phi^3$ Feynman diagrams.
Consider any Feynman diagram, $T$, with $n$ leaves, $n-2$ degree-three vertices, and $n-3$ internal edges. Each internal edge has a length $f_I$, with $I\in \{1,2,\dots ,n-3\}$. If any of the lengths are set to zero, the diagram loses two degree-three vertices and gains a degree-four vertex. There are three ways of ``resolving'' a degree-four vertex into two degree-three ones. One of the three ways leads back to $T$ while the other two lead to two different Feynman diagrams $T'$ and $T''$. This can be done for any internal edge and therefore $T$ is connected to $2(n-3)$ other trees this way.

More generally, we call a {\it flip} the operation of sending a length to zero to produce a degeneration of a diagram to connect it to a different diagram that shares the same degeneration.

It is worth noting that this process is analog to what is known in mathematical physics as a flop transition, in which a cycle in a manifold is sent to zero size (K\"ahler volume) and then replaced by another cycle that grows in size. In fact, when this is done in a toric variety, the description of a flop is identical to that of a mutation in a triangulation of a polygon representation of a Feynman diagram.

Flips of Feynman diagrams are very useful in scattering amplitudes, for example in the study of Bern--Carrasco--Johansson double copy relations \cite{Adamo:2022dcm,Bern:2019prr, Bern:2022wqg,Bern:2008qj} and their geometric explanations~\cite{Arkani-Hamed:2017mur,Gao:2017dek,He:2021lro, He:2018pue,Herderschee:2019wtl}. Here we focus on generalizing them to higher~$k$.

\subsection[Flips in k=3 Feynman diagrams]{Flips in $\boldsymbol{k=3}$ Feynman diagrams}\label{secGFDflips}

Flips of $k=3$ GFDs are also defined using degenerations produced by sending internal edge lengths to zero. Recall that a $(3,n)$ GFD has $n(n-4)$ internal edges, $n-4$ for each tree in the collection, but only $2(n-4)$ internal lengths are independent due to the compatibility condition that produces the $k=3$ metric $d_{abc}$. This notion also gives a natural definition of GFDs connected by a flip for all $(k,n)$.

\begin{defn}
Two GFDs are said to be related by a flip if they have a common codim-1 degenerate GFD.
\end{defn}

A codim-1 degenerate GFD usually contains $k=2$ FDs with cubic or quartic vertices but starting from $n=7$, it may also contain quintic or higher degree vertices.
In practice, to get flips of a GFD, one can degenerate it first and then
blow it up.
However, unlike the $k=2$ case, not all ways of blowing up quartic
 or higher degree vertices in the various diagrams in the collection lead to valid a GFD. The reason is that randomly resolving quartic or higher degree vertices does not guarantee that the new arrangement of trees will have a valid non-degenerate metric. While this might seem to make the problem complicated, it is actually a simplification.

Consider the set of all GCOs a GFD is compatible with. Any of its degenerations must be compatible with a larger set of GCOs.
A necessary condition for the resolution of the degeneration to be allowed is that the new candidate GFD be compatible with at least one of the GCOs of the degenerate GFD. In most cases, the two GFDs connected by flips also share a~common GCO.

For example, consider the following GFDs at $n=6$,
\begin{align}
{\cal T}_A\colon\
\left(\!\!
 {
 \raisebox{-.85 cm}{
\tikz[xscale=.55,yscale=.35]{
\draw[]
(0,0)--(2,0)
(0,0)--(-.3,-1.3) node[below]{2}
(0,0)--(-.3,1.3) node[above]{3}
(2,0)--(2.3,1.3) node[above]{4}
(2,0)--(2.3,-1.3) node[below]{5}
(1,0)-- (1,-1.3) node[below]{6}
;
\draw[blue, very thick] (0,0)--(.5,0) node [above] {$y$}--(1,0);
\draw[red, very thick] (1,0)--(1.5,0) node [above] {$z$}--(2,0);
}}
},\!
 {
 \raisebox{-.85 cm}{
\tikz[xscale=.55,yscale=.35]{
\draw[]
(0,0)--(2,0)
(0,0)--(-.3,-1.3) node[below]{1}
(0,0)--(-.3,1.3) node[above]{3}
(2,0)--(2.3,1.3) node[above]{4}
(2,0)--(2.3,-1.3) node[below]{5}
(1,0)-- (1,-1.3) node[below]{6}
;
\draw[blue, very thick] (0,0)--(.5,0) node [above] {$y$}--(1,0);
\draw[red, very thick] (1,0)--(1.5,0) node [above] {$z$}--(2,0);
}}
},\!
 {
 \raisebox{-.85 cm}{
\tikz[xscale=.55,yscale=.35]{
\draw[]
(0,0)--(2,0)
(0,0)--(-.3,-1.3) node[below]{1}
(0,0)--(-.3,1.3) node[above]{2}
(2,0)--(2.3,1.3) node[above]{4}
(2,0)--(2.3,-1.3) node[below]{5}
(1,0)-- (1,-1.3) node[below]{6}
;
\draw[blue, very thick] (0,0)--(.5,0) node [above] {$~x{+}y$}--(1,0);
\draw[red, very thick] (1,0)--(1.5,0) node [above] {~$z$}--(2,0);
}}
},\!
 {
 \raisebox{-.85 cm}{
\tikz[xscale=.55,yscale=.35]{
\draw[]
(0,0)--(2,0)
(0,0)--(-.3,-1.3) node[below]{1}
(0,0)--(-.3,1.3) node[above]{2}
(2,0)--(2.3,1.3) node[above]{5}
(2,0)--(2.3,-1.3) node[below]{6}
(1,0)-- (1,1.3) node[above]{3}
;
\draw[blue, very thick] (0,0)--(.5,0) node [below] {$x$}--(1,0);
\draw[red, very thick] (1,0)--(1.5,0) node [below] {$w$}--(2,0);
}}
},\!
 {
 \raisebox{-.85 cm}{
\tikz[xscale=.55,yscale=.35]{
\draw[]
(0,0)--(2,0)
(0,0)--(-.3,-1.3) node[below]{1}
(0,0)--(-.3,1.3) node[above]{2}
(2,0)--(2.3,1.3) node[above]{4}
(2,0)--(2.3,-1.3) node[below]{6}
(1,0)-- (1,1.3) node[above]{3}
;
\draw[blue, very thick] (0,0)--(.5,0) node [below] {$x$}--(1,0);
\draw[red, very thick] (1,0)--(1.5,0) node [below] {$w$}--(2,0);
}}
},\!
 {
 \raisebox{-.85 cm}{
\tikz[xscale=.55,yscale=.35]{
\draw[]
(0,0)--(2,0)
(0,0)--(-.3,-1.3) node[below]{1}
(0,0)--(-.3,1.3) node[above]{2}
(2,0)--(2.3,1.3) node[above]{4}
(2,0)--(2.3,-1.3) node[below]{5}
(1,0)-- (1,1.3) node[above]{3}
;
\draw[blue, very thick] (0,0)--(.5,0) node [below] {$x$~~}--(1,0);
\draw[red, very thick] (1,0)--(1.5,0) node [below] {$z{+}w$~~}--(2,0);
}}
}
\!\!
\right).\!\!\!\label{original0}
\end{align}
All twelve internal lengths are explicitly given in term of only four.

${\cal T}_A$ has 4 codimension-1 degenerations, $x\to 0$, $y\to 0$, $z\to 0$, or $w\to 0$. For example, for $w\to 0$, it becomes
 \begin{align}
\label{degenerateGFD}
\left(
 {
 \raisebox{-.85 cm}{
\tikz[xscale=.55,yscale=.35]{
\draw[]
(0,0)--(2,0)
(0,0)--(-.3,-1.3) node[below]{2}
(0,0)--(-.3,1.3) node[above]{3}
(2,0)--(2.3,1.3) node[above]{4}
(2,0)--(2.3,-1.3) node[below]{5}
(1,0)-- (1,-1.3) node[below]{6}
;
\draw[blue, very thick] (0,0)--(.5,0) node [above] {$y$}--(1,0);
\draw[red, very thick] (1,0)--(1.5,0) node [above] {$z$}--(2,0);
}}
},
 {
 \raisebox{-.85 cm}{
\tikz[xscale=.55,yscale=.35]{
\draw[]
(0,0)--(2,0)
(0,0)--(-.3,-1.3) node[below]{1}
(0,0)--(-.3,1.3) node[above]{3}
(2,0)--(2.3,1.3) node[above]{4}
(2,0)--(2.3,-1.3) node[below]{5}
(1,0)-- (1,-1.3) node[below]{6}
;
\draw[blue, very thick] (0,0)--(.5,0) node [above] {$y$}--(1,0);
\draw[red, very thick] (1,0)--(1.5,0) node [above] {$z$}--(2,0);
}}
},
 {
 \raisebox{-.85 cm}{
\tikz[xscale=.55,yscale=.35]{
\draw[]
(0,0)--(2,0)
(0,0)--(-.3,-1.3) node[below]{1}
(0,0)--(-.3,1.3) node[above]{2}
(2,0)--(2.3,1.3) node[above]{4}
(2,0)--(2.3,-1.3) node[below]{5}
(1,0)-- (1,-1.3) node[below]{6}
;
\draw[blue, very thick] (0,0)--(.5,0) node [above] {$~x{+}y$}--(1,0);
\draw[red, very thick] (1,0)--(1.5,0) node [above] {$~z$}--(2,0);
}}
},
 {
 \raisebox{-.85 cm}{
\tikz[xscale=.55,yscale=.35]{
\draw[]
(0,0)--(1,0)
(0,0)--(-.3,-1.3) node[below]{1}
(0,0)--(-.3,1.3) node[above]{2}
(1,0)--(2.3,1.3) node[above]{5}
(1,0)--(2.3,-1.3) node[below]{6}
(1,0)-- (1,1.3) node[above]{3}
;
\draw[blue, very thick] (0,0)--(.5,0) node [below] {$x$}--(1,0);
}}
},
 {
 \raisebox{-.85 cm}{
\tikz[xscale=.55,yscale=.35]{
\draw[]
(0,0)--(1,0)
(0,0)--(-.3,-1.3) node[below]{1}
(0,0)--(-.3,1.3) node[above]{2}
(1,0)--(2.3,1.3) node[above]{4}
(1,0)--(2.3,-1.3) node[below]{6}
(1,0)-- (1,1.3) node[above]{3}
;
\draw[blue, very thick] (0,0)--(.5,0) node [below] {$x$}--(1,0);
}}
},
 {
 \raisebox{-.85 cm}{
\tikz[xscale=.55,yscale=.35]{
\draw[]
(0,0)--(2,0)
(0,0)--(-.3,-1.3) node[below]{1}
(0,0)--(-.3,1.3) node[above]{2}
(2,0)--(2.3,1.3) node[above]{4}
(2,0)--(2.3,-1.3) node[below]{5}
(1,0)-- (1,1.3) node[above]{3}
;
\draw[blue, very thick] (0,0)--(.5,0) node [below] {$x$}--(1,0);
\draw[red, very thick] (1,0)--(1.5,0) node [below] {$z$}--(2,0);
}}
}
\right) .
\end{align}

Now, ${\cal T}_A$ is compatible with $16$ color orderings. Eight of them are of type 0, i.e., descendants of $(k=2,n=6)$ color orderings
like $(123456)$, $(213456)$, and eight type I color orderings such~as,
\begin{align}
\label{blow1ordering}
&( (23456),(13456),(12456),(12365),(12364),(12354) ),
\\
\label{blow2ordering}
&((23645),(13645),(12546),(12653),(12643),(12453)).
\end{align}
For each of the compatible color ordering, the degeneration \eqref{degenerateGFD} can be resolved in exactly two ways, one leads back to ${\cal T}_A$ and the other to a new GFD.

For example, using the ordering \eqref{blow1ordering}, the quartic vertices in \eqref{degenerateGFD} can be resolved in a~unique way (without coming back to \eqref{original0}). This is done by using the fourth and fifth $k=2$ orderings, $(12365)$, $(12364)$, to perform the $k=2$ flip on the fourth and fifth trees respectively, resulting~in\looseness=1
 \def\fFD #1,#2,#3,#4,#5 \fFD {
 \raisebox{-.85 cm}{
\tikz[scale=.3]{
\draw[]
(0,0)--(2,0)
(0,0)--(-.3,-1.3) node[below]{#1}
(0,0)--(-.3,1.3) node[above]{#2}
(2,0)--(2.3,1.3) node[above]{#3}
(2,0)--(2.3,-1.3) node[below]{#4}
(1,0)-- (1,1.3) node[above]{#5}
;
}}
}
 \def\fFDd #1,#2,#3,#4,#5 \fFDd {
 \raisebox{-.85 cm}{
\tikz[scale=.3]{
\draw[]
(0,0)--(2,0)
(0,0)--(-.3,-1.3) node[below]{#1}
(0,0)--(-.3,1.3) node[above]{#2}
(2,0)--(2.3,1.3) node[above]{#3}
(2,0)--(2.3,-1.3) node[below]{#4}
(1,0)-- (1,-1.3) node[below]{#5}
;
}}
}
\begin{align*}
{\cal T}_B\colon\
\left(
\fFDd 2, 3, 4, 5, 6 \fFDd,\fFDd 1, 3, 4, 5, 6 \fFDd,\fFDd 1, 2, 4, 5, 6 \fFDd,\fFDd 1, 2, 3, 6, 5 \fFDd,\fFDd 1, 2, 3, 6, 4 \fFDd,\fFD 1, 2, 4, 5,3 \fFD
\right) .
\end{align*}
${\cal T}_B$ has a non-degenerate metric and hence is indeed a GFD.

Using the other compatible ordering presented in \eqref{blow2ordering}, the original GFD ${\cal T}_A$ has a different flip and results in a different GFD,
\begin{align*}
{\cal T}_C\colon\
\left(
\fFDd 2, 3, 4, 5, 6 \fFDd,
\fFDd 1, 3, 4, 5, 6 \fFDd,
\fFDd 1, 2, 4, 5, 6 \fFDd,
\fFDd 2, 1, 3, 5, 6 \fFDd,
\fFDd 2, 1, 3, 4, 6 \fFDd,
\fFD 1, 2, 4, 5,3 \fFD
\right) .
\end{align*}

By considering all degenerations and all compatible color orderings, one can find eight flips in total and eight new GFDs. The reason why the number is much less than $4\times 16=64$ is that two compatible orderings may lead to the same flips, which is already the case for $k=2$. In order to avoid producing the same GFDs several times, one can also ignore all GCOs at first and blow up a degenerate GFD in all possible topological ways, resulting in many arrangements of metric trees. Then, select those with a correct number of independent internal lengths and verify if they share a common GCO with the original GFD or its degenerations. This is an equivalent way to find all flips of a GFD.

Starting at $n=7$, there are some GFDs whose degenerations contain a $k=2$ Feynman diagram with quintic or higher degree vertices and there is no unique way to resolve these high-degree vertices anymore guided by the compatible GCOs of the degenerate GFDs. One has to consider all possible ways to resolve the high-degree vertices and select the new GFDs from the resulting arrangements of metric trees whose metrics have the correct number of independent internal lengths.\footnote{One can resolve all quartic or higher degree vertices
into purely cubic vertices at first and then degenerate the resulting arrangements of metric trees which contain too many independent internal lengths until they become valid GFDs.} For example, the GFD with quartic vertices in~\eqref{boundaryofT0}, which we present here again for convenience
\def\fFD #1,#2,#3,#4,#5,#6 \fFD 	{
 \raisebox{-.65 cm}{
\tikz[scale=.3]{
\draw[]
(0,0)--(90:1)--++(90+45:.7) node[left=-2pt]{#1}
(0,0)--(90:1)--++(90-45:.7) node[right=-2pt]{#2}
(0,0)--(90-120:1)--++(90-120+45:.7) node[right=-2pt]{#3}
(0,0)--(90-120:1)--++(90-120-45:.7) node[below=-2pt]{#4}
(0,0)--(90+120:1)--++(90+120+45:.7) node[below=-2pt]{#5}
(0,0)--(90+120:1)--++(90+120-45:.7) node[left=-2pt]{#6}
;
}}}
\def\fFDD #1,#2,#3,#4,#5,#6 \fFDD 	{
 \raisebox{-.65 cm}{
\tikz[scale=.3]{
\draw[]
(0,0)--(90:0)--++(90+45:.7) node[left=-2pt]{#1}
(0,0)--(90:0)--++(90-45:.7) node[right=-2pt]{#2}
(0,0)--(90-120:1)--++(90-120+45:.7) node[right=-2pt]{#3}
(0,0)--(90-120:1)--++(90-120-45:.7) node[below=-2pt]{#4}
(0,0)--(90+120:1)--++(90+120+45:.7) node[below=-2pt]{#5}
(0,0)--(90+120:1)--++(90+120-45:.7) node[left=-2pt]{#6}
;
}}}
\begin{align}
&\left(\!\!
 \fFDD 3, 7, 4, 5, 6, 2 \fFDD,\!
 \fFD 3, 4, 5, 7, 6, 1 \fFD,\!
 \fFDD 1, 7, 2, 4, 5, 6 \fFDD,\!
 \fFD 7, 6, 5, 1, 2, 3 \fFD,\!
 \fFD 2, 7, 3, 6, 4, 1 \fFD,\!
 \fFD 4, 7, 1, 2, 3, 5 \fFD,\!
 \fFDD 3, 1, 2, 5, 4, 6 \fFDD\!\!
\right) ,\!\!\!\!\label{boundaryofT0again}
\end{align}
 is flipped to another GFD with pure cubic vertices
\def\fFD #1,#2,#3,#4,#5,#6 \fFD 	{
 \raisebox{-.65 cm}{
\tikz[scale=.3]{
\draw[]
(0,0)--(90:1)--++(90+45:.7) node[left=-2pt]{#1}
(0,0)--(90:1)--++(90-45:.7) node[right=-2pt]{#2}
(0,0)--(90-120:1)--++(90-120+45:.7) node[right=-2pt]{#3}
(0,0)--(90-120:1)--++(90-120-45:.7) node[below=-2pt]{#4}
(0,0)--(90+120:1)--++(90+120+45:.7) node[below=-2pt]{#5}
(0,0)--(90+120:1)--++(90+120-45:.7) node[left=-2pt]{#6}
;
}}}
\def\fFDD #1,#2,#3,#4,#5,#6 \fFDD 	{
 \raisebox{-.65 cm}{
\tikz[scale=.3]{
\draw[]
(0,0)--(90:0)--++(90+45:.7) node[left=-2pt]{#1}
(0,0)--(90:0)--++(90-45:.7) node[right=-2pt]{#2}
(0,0)--(90-120:1)--++(90-120+45:.7) node[right=-2pt]{#3}
(0,0)--(90-120:1)--++(90-120-45:.7) node[below=-2pt]{#4}
(0,0)--(90+120:1)--++(90+120+45:.7) node[below=-2pt]{#5}
(0,0)--(90+120:1)--++(90+120-45:.7) node[left=-2pt]{#6}
;
}}}
\begin{align*}
\left(
 \fFD 2, 3, 4, 5, 7, 6 \fFD,
 %
 	{
 \raisebox{-.85 cm}{
\tikz[scale=.3]{
\draw[]
(0,0)--(3,0)
(0,0)--(-.3,-1.3) node[below]{3}
(0,0)--(-.3,1.3) node[above]{4}
(1,0)-- (1,-1.3) node[below]{1}
(2,0)-- (2,-1.3) node[below]{6}
(3,0)--(3.3,1.3) node[above]{5}
(3,0)--(3.3,-1.3) node[below]{7}
;
}}
},
 %
	{
 \raisebox{-.85 cm}{
\tikz[scale=.3]{
\draw[]
(0,0)--(3,0)
(0,0)--(-.3,-1.3) node[below]{5}
(0,0)--(-.3,1.3) node[above]{6}
(1,0)-- (1,-1.3) node[below]{7}
(2,0)-- (2,1.3) node[above]{1}
(3,0)--(3.3,1.3) node[above]{2}
(3,0)--(3.3,-1.3) node[below]{4}
;
}}
},
 \fFD 7, 6, 5, 1, 2, 3 \fFD,
 \fFD 2, 7, 3, 6, 4, 1 \fFD,
 %
 {
 \raisebox{-.85 cm}{
\tikz[scale=.3]{
\draw[]
(0,0)--(3,0)
(0,0)--(-.3,-1.3) node[below]{3}
(0,0)--(-.3,1.3) node[above]{5}
(1,0)-- (1,-1.3) node[below]{2}
(2,0)-- (2,-1.3) node[below]{1}
(3,0)--(3.3,1.3) node[above]{4}
(3,0)--(3.3,-1.3) node[below]{7}
;
}}
},
%
{
 \raisebox{-.85 cm}{
\tikz[scale=.3]{
\draw[]
(0,0)--(3,0)
(0,0)--(-.3,-1.3) node[below]{4}
(0,0)--(-.3,1.3) node[above]{6}
(1,0)-- (1,1.3) node[above]{1}
(2,0)-- (2,-1.3) node[below]{3}
(3,0)--(3.3,1.3) node[above]{2}
(3,0)--(3.3,-1.3) node[below]{5}
;
}}
}
\right) ,
\end{align*}
guided by their common GCO \eqref{37GCOquar} via their common codim-1 degeneration
\def\fFD #1,#2,#3,#4,#5,#6 \fFD 	{
 \raisebox{-.65 cm}{
\tikz[scale=.3]{
\draw[]
(0,0)--(90:1)--++(90+45:.7) node[left=-2pt]{#1}
(0,0)--(90:1)--++(90-45:.7) node[right=-2pt]{#2}
(0,0)--(90-120:1)--++(90-120+45:.7) node[right=-2pt]{#3}
(0,0)--(90-120:1)--++(90-120-45:.7) node[below=-2pt]{#4}
(0,0)--(90+120:1)--++(90+120+45:.7) node[below=-2pt]{#5}
(0,0)--(90+120:1)--++(90+120-45:.7) node[left=-2pt]{#6}
;
}}}
\def\fFDD #1,#2,#3,#4,#5,#6 \fFDD 	{
 \raisebox{-.65 cm}{
\tikz[scale=.3]{
\draw[]
(0,0)--(90:0)--++(90+45:.7) node[left=-2pt]{#1}
(0,0)--(90:0)--++(90-45:.7) node[right=-2pt]{#2}
(0,0)--(90-120:1)--++(90-120+45:.7) node[right=-2pt]{#3}
(0,0)--(90-120:1)--++(90-120-45:.7) node[below=-2pt]{#4}
(0,0)--(90+120:1)--++(90+120+45:.7) node[below=-2pt]{#5}
(0,0)--(90+120:1)--++(90+120-45:.7) node[left=-2pt]{#6}
;
}
}}
\begin{align*}
\left(
	{
 \raisebox{-.65 cm}{
\tikz[scale=.3]{
\draw[]
(0,0)--(90:0)--++(90+45-25:.9) node[left=-2pt]{3}
(0,0)--(90:0)--++(90-45:.9) node[right=-2pt]{7}
(0,0)--(90-120:1)--++(90-120+45:.7) node[right=-2pt]{4}
(0,0)--(90-120:1)--++(90-120-45:.7) node[below=-2pt]{5}
(0,0)--(90+120:0)--++(90+120+45:.9) node[below=-2pt]{6}
(0,0)--(90+120:0)--++(90+120-45+25:.9) node[left=-2pt]{2}
;
}
}},
 {
 \raisebox{-.65 cm}{
\tikz[scale=.3]{
\draw[]
(0,0)--(90:1)--++(90+45:.7) node[left=-2pt]{3}
(0,0)--(90:1)--++(90-45:.7) node[right=-2pt]{4}
(0,0)--(90-120:1)--++(90-120+45:.7) node[right=-2pt]{5}
(0,0)--(90-120:1)--++(90-120-45:.7) node[below=-2pt]{7}
(0,0)--(90+120:0)--++(90+120+45:.7) node[below=-2pt]{6}
(0,0)--(90+120:0)--++(90+120-45:.7) node[left=-2pt]{1}
;
}}},
 \fFDD 1, 7, 2, 4, 5, 6 \fFDD,
 \fFD 7, 6, 5, 1, 2, 3 \fFD,
 \fFD 2, 7, 3, 6, 4, 1 \fFD,
 %
 {
 \raisebox{-.65 cm}{
\tikz[scale=.3]{
\draw[]
(0,0)--(90:1)--++(90+45:.7) node[left=-2pt]{4}
(0,0)--(90:1)--++(90-45:.7) node[right=-2pt]{7}
(0,0)--(90-120:0)--++(90-120+45:.7) node[right=-2pt]{1}
(0,0)--(90-120:0)--++(90-120-45:.7) node[below=-2pt]{2}
(0,0)--(90+120:1)--++(90+120+45:.7) node[below=-2pt]{3}
(0,0)--(90+120:1)--++(90+120-45:.7) node[left=-2pt]{5}
;
}}},
 \fFDD 3, 1, 2, 5, 4, 6 \fFDD
\right) ,
\end{align*}
which contains a quintic vertex.

There are also few cases where two GFDs connected by flips do not share any GCOs, for instance, the GFD \eqref{boundaryofT0again} and its relabelling
\def\fFD #1,#2,#3,#4,#5,#6 \fFD 	{
 \raisebox{-.65 cm}{
\tikz[scale=.3]{
\draw[]
(0,0)--(90:1)--++(90+45:.7) node[left=-2pt]{#1}
(0,0)--(90:1)--++(90-45:.7) node[right=-2pt]{#2}
(0,0)--(90-120:1)--++(90-120+45:.7) node[right=-2pt]{#3}
(0,0)--(90-120:1)--++(90-120-45:.7) node[below=-2pt]{#4}
(0,0)--(90+120:1)--++(90+120+45:.7) node[below=-2pt]{#5}
(0,0)--(90+120:1)--++(90+120-45:.7) node[left=-2pt]{#6}
;
}}}
\begin{align}
\label{37T0370}
\left(
{
 \raisebox{-.65 cm}{
\tikz[scale=.3]{
\draw[]
(0,0)--(90:1)--++(90+45:.7) node[left=-2pt]{3}
(0,0)--(90:1)--++(90-45:.7) node[right=-2pt]{7}
(0,0)--(90-120:0)--++(90-120+45:.7) node[right=-2pt]{4}
(0,0)--(90-120:0)--++(90-120-45:.7) node[below=-2pt]{5}
(0,0)--(90+120:1)--++(90+120+45:.7) node[below=-2pt]{6}
(0,0)--(90+120:1)--++(90+120-45:.7) node[left=-2pt]{2}
;
}}},
 \fFD 3, 4, 5, 7, 6, 1 \fFD,
 \fFD 1, 7, 2, 4, 5, 6 \fFD,
 {
 \raisebox{-.65 cm}{
\tikz[scale=.3]{
\draw[]
(0,0)--(90:1)--++(90+45:.7) node[left=-2pt]{7}
(0,0)--(90:1)--++(90-45:.7) node[right=-2pt]{6}
(0,0)--(90-120:0)--++(90-120+45:.7) node[right=-2pt]{5}
(0,0)--(90-120:0)--++(90-120-45:.7) node[below=-2pt]{1}
(0,0)--(90+120:1)--++(90+120+45:.7) node[below=-2pt]{2}
(0,0)--(90+120:1)--++(90+120-45:.7) node[left=-2pt]{3}
;
}}}
,
%
{
 \raisebox{-.65 cm}{
\tikz[scale=.3]{
\draw[]
(0,0)--(90:1)--++(90+45:.7) node[left=-2pt]{2}
(0,0)--(90:1)--++(90-45:.7) node[right=-2pt]{7}
(0,0)--(90-120:1)--++(90-120+45:.7) node[right=-2pt]{3}
(0,0)--(90-120:1)--++(90-120-45:.7) node[below=-2pt]{6}
(0,0)--(90+120:0)--++(90+120+45:.7) node[below=-2pt]{4}
(0,0)--(90+120:0)--++(90+120-45:.7) node[left=-2pt]{1}
;
}}}
,
 \fFD 4, 7, 1, 2, 3, 5 \fFD,
 \fFD 3, 1, 2, 5, 4, 6 \fFD
\right),
\end{align}
whose
 common degeneration is given by
\def\fFDD #1,#2,#3,#4,#5,#6 \fFDD 	{
 \raisebox{-.65 cm}{
\tikz[scale=.3]{
\draw[]
(0,0)--(90:0)--++(90+45:.7) node[left=-2pt]{#1}
(0,0)--(90:0)--++(90-45:.7) node[right=-2pt]{#2}
(0,0)--(90-120:1)--++(90-120+45:.7) node[right=-2pt]{#3}
(0,0)--(90-120:1)--++(90-120-45:.7) node[below=-2pt]{#4}
(0,0)--(90+120:1)--++(90+120+45:.7) node[below=-2pt]{#5}
(0,0)--(90+120:1)--++(90+120-45:.7) node[left=-2pt]{#6}
;
}}}
\def\fFD #1,#2,#3,#4,#5,#6 \fFD 	{
 \raisebox{-.65 cm}{
\tikz[scale=.3]{
\draw[]
(0,0)--(90:1)--++(90+45:.7) node[left=-2pt]{#1}
(0,0)--(90:1)--++(90-45:.7) node[right=-2pt]{#2}
(0,0)--(90-120:1)--++(90-120+45:.7) node[right=-2pt]{#3}
(0,0)--(90-120:1)--++(90-120-45:.7) node[below=-2pt]{#4}
(0,0)--(90+120:1)--++(90+120+45:.7) node[below=-2pt]{#5}
(0,0)--(90+120:1)--++(90+120-45:.7) node[left=-2pt]{#6}
;
}}}
\begin{align}
\label{37T0371}
\left({
 \raisebox{-.65 cm}{
\tikz[scale=.3]{
\draw[]
(0,0)--(90:0)--++(90+45:.7) node[left=-2pt]{3}
(0,0)--(90:0)--++(90-30:.7) node[right=-2pt]{7}
(0,0)--(90-120:0)--++(90-120+30:.7) node[right=-2pt]{4}
(0,0)--(90-120:0)--++(90-120-45:.7) node[below=-2pt]{5}
(0,0)--(90+120:1)--++(90+120+45:.7) node[below=-2pt]{6}
(0,0)--(90+120:1)--++(90+120-45:.7) node[left=-2pt]{2}
;
}}},
 \fFD 3, 4, 5, 7, 6, 1 \fFD,
 \fFDD 1, 7, 2, 4, 5, 6 \fFDD,
 {
 \raisebox{-.65 cm}{
\tikz[scale=.3]{
\draw[]
(0,0)--(90:1)--++(90+45:.7) node[left=-2pt]{7}
(0,0)--(90:1)--++(90-45:.7) node[right=-2pt]{6}
(0,0)--(90-120:0)--++(90-120+45:.7) node[right=-2pt]{5}
(0,0)--(90-120:0)--++(90-120-45:.7) node[below=-2pt]{1}
(0,0)--(90+120:1)--++(90+120+45:.7) node[below=-2pt]{2}
(0,0)--(90+120:1)--++(90+120-45:.7) node[left=-2pt]{3}
;
}}}
,
%
{
 \raisebox{-.65 cm}{
\tikz[scale=.3]{
\draw[]
(0,0)--(90:1)--++(90+45:.7) node[left=-2pt]{2}
(0,0)--(90:1)--++(90-45:.7) node[right=-2pt]{7}
(0,0)--(90-120:1)--++(90-120+45:.7) node[right=-2pt]{3}
(0,0)--(90-120:1)--++(90-120-45:.7) node[below=-2pt]{6}
(0,0)--(90+120:0)--++(90+120+45:.7) node[below=-2pt]{4}
(0,0)--(90+120:0)--++(90+120-45:.7) node[left=-2pt]{1}
;
}}}
,
 \fFD 4, 7, 1, 2, 3, 5 \fFD,
 \fFDD 3, 1, 2, 5, 4, 6 \fFDD
\right).
\end{align}
 The GFD \eqref{boundaryofT0} has 64 compatible GCOs such as \eqref{37GCOquar} while its relabeling \eqref{37T0370} has totally different 64 other compatible GCOs such as
\[
((2 4 7 3 5 6),(1 4 3 7 5 6),(1 2 4 7 3 5),(1 2 7 4 3 6),(1 2 3 5 7 6),(1 6 5 4 2 7),(1 3 2 5 4 6)).
\]
The $64+64=128$ GCOs
 together make up the set of compatible GCOs
 for the degeneration in~\eqref{37T0371}.

In general, each $(k=3,n)$ GFD has at least $2(n-4)$ degenerations and $4(n-4)$ flips by considering all compatible color orderings. This has been verified for $(3,n)$ up to $n=8$.

\subsection{Bootstrap algorithms}\label{bootstrapIandII}

Here we present two algorithms for computing GFD. Both are based on the idea that using flips all GFDs can be generated starting from a seed.

{\bf Bootstrap I:}
\begin{enumerate}\itemsep=0pt
\item[(1)] Start with any GFD, such as a descendant of a $(k=2,n)$ Feynman diagram as a seed.
\item[(2)] For each GFD in the list, flip it in all possible ways allowed by GCOs and add any new GFDs to the list.
\item[(3)] Repeat step 2 until no new GFDs are produced.
\end{enumerate}

Of course, making use of relabelling simplifies the procedure and it is a step that can be included in the algorithm.

For example, starting with the seed ${\cal T}_A$ in \eqref{original0},
which is the descendant of a $k=2$ Feynman diagram, \[
 \raisebox{-.85 cm}{
\tikz[scale=.3]{
\draw[]
(0,0)--(3,0)
(0,0)--(-.3,-1.3) node[below]{1}
(0,0)--(-.3,1.3) node[above]{2}
(3,0)--(3.3,1.3) node[above]{4}
(3,0)--(3.3,-1.3) node[below]{5}
(1,0)-- (1,1.3) node[above]{3}
(2,0)-- (2,-1.3) node[below]{6}
;
}}
,\]
a single layer of its flips produces two new GFDs ${\cal T}_B$ and ${\cal T}_C$ and their relabelling. Then flipping~${\cal T}_B$ and ${\cal T}_C$ again guided by their own compatible orderings produces some new ones. Repeating until no new classes of GFDs are produced, we get four other classes of GFDs which we denote as ${\cal T}_D$, ${\cal T}_E$, ${\cal T}_F$ and ${\cal T}_G$,
 \def\fFD #1,#2,#3,#4,#5 \fFD {
 \raisebox{-.61 cm}{
\tikz[scale=.17]{
\draw[]
(0,0)--(2,0)
(0,0)--(-.3,-1.3) node[below]{#1}
(0,0)--(-.3,1.3) node[above]{#2}
(2,0)--(2.3,1.3) node[above]{#3}
(2,0)--(2.3,-1.3) node[below]{#4}
(1,0)-- (1,1.3) node[above]{#5}
;
}}
}
\begin{gather*}
{\cal T}_D\colon\
\left( \fFD 2, 3, 4, 5, 6 \fFD,\fFD 1, 3, 4, 5, 6 \fFD,\fFD 1, 2, 4, 5, 6 \fFD,\fFD 1, 5, 3, 6, 2 \fFD,\fFD 1, 4, 3, 6, 2 \fFD,\fFD 1, 2, 4, 5, 3 \fFD
\right) ,\\
{\cal T}_E\colon\
\left(
\fFD 4, 5, 3, 6, 2 \fFD,\fFD 4, 5, 3, 6, 1 \fFD,\fFD 1, 2, 4, 5, 6 \fFD,\fFD 1, 2, 3, 6, 5 \fFD,\fFD 1, 2, 3, 6, 4 \fFD,\fFD 1, 2, 4, 5, 3 \fFD
\right) ,
\\
{\cal T}_F\colon\
\left(
\fFD 2, 3, 4, 5, 6 \fFD,\fFD 1, 3, 4, 5, 6 \fFD,\fFD 1, 2, 5, 6, 4 \fFD,\fFD 1, 5, 3, 6, 2 \fFD,\fFD 1, 4, 3, 6, 2 \fFD,\fFD 1, 2, 3, 5, 4 \fFD
\right) ,\\
{\cal T}_G\colon\
\left(
\fFD 2, 3, 4, 5, 6 \fFD,\fFD 1, 3, 4, 6, 5 \fFD,\fFD 1, 2, 5, 6, 4 \fFD,\fFD 1, 5, 2, 6, 3 \fFD,\fFD 1, 4, 3, 6, 2 \fFD,\fFD 2, 4, 3, 5, 1 \fFD
\right).
\end{gather*}
The permutation of labels of these seven classes of GFDs gives rise to
$1005$ GFDs in total, which is consistent with the result in the literature \cite{speyer2004tropical}, see also \cite{herrmann2008draw}. For the reader's convenience, we also include additional information for each representative, including their contributions ${\cal R}(\cal T)$ to the amplitudes via \eqref{Scw3}, in Table~\ref{amplitude36con}.

\begin{table}[!ht]\renewcommand{\arraystretch}{1.2}
 \centering
\begin{tabular}{|c|c|c|c|c|c|c|c|}\hline
 & Contribution ${\cal R}(\cal T)$ to the amplitude & \# of flips & \# of perm. \\ \hline
${\cal T}_A $ & ${1}/({\s_{123} \s_{456} \t_{1236} \t_{3456}} ) $ & 8& 90\\ 
${\cal T}_B $ &${1}/({\sfR_{45,12,36} \s_{123} \t_{1236} \t_{3456}})$ & 8 & 180 \\ 
${\cal T}_C $ &${1}/({\s_{123} \s_{345} \t_{1236} \t_{3456}} )$& 8 &90 \\ 
${\cal T}_D $ & $ {1}/({\sfR_{45,12,36} \s_{123} \s_{145} \t_{1236}})$& 8 &360 \\ 
${\cal T}_E $ &$
\displaystyle{({\sfR_{12,45,36}+\sfR_{45,12,36}})/({\sfR_{12,45,36} \sfR_{45,12,36} \t_{1236} \t_{1245} \t_{3456}}} )$ & 12 &15 \\ 
${\cal T}_F $ & $ {1}/({\s_{123} \s_{145} \s_{246} \s_{356}})$& 8 & 30 \\ 
${\cal T}_G $ &${1}/({\sfR_{45,12,36} \s_{123} \s_{145} \s_{356}} )$& 8 & 240 \\ \hline
\end{tabular}
 \caption{Contribution of GFDs of different types to the amplitude.
 The definitions of $\t$, $\sfR$ are given in~\eqref{defpoles}.
 The last two columns denote the numbers of flips and distinct permutations of a GFD, respectively.} \label{amplitude36con}
\end{table}

With these 1005 GFDs at hand, one can pick out those that are compatible with a particular ordering. The four representatives of color orderings in Table~\ref{table36} have 48, 41, 44 and 45 GFDs respectively. Table~\ref{compatible-Table} shows what classes these GFDs belong to.

Some noteworthy facts are the following.

Class F GFDs are not covered by the type 0 color orderings and type III color orderings only contain GFDs of classes F and G.

\begin{table}[!ht]\renewcommand{\arraystretch}{1.2}
 \centering
\begin{tabular}{|c|c|c|c|c|c|c|c|}\hline
 &class A &class B &class C & class D &class E&class F&class G \\ \hline
type 0 & \checkmark & \checkmark & \checkmark & \checkmark & \checkmark & &\checkmark \\
type I & \checkmark & \checkmark & \checkmark & \checkmark &
& \checkmark& \checkmark \\
type II & & \checkmark & \checkmark & \checkmark & \checkmark & \checkmark & \checkmark \\
type III &&&& & & \checkmark & \checkmark \\ \hline
\end{tabular}
 \caption{Compatibilities between different types of color orderings and different classes of GFDs.}\label{compatible-Table}
\end{table}

Any GFD is compatible with the same number of color orderings, $16$ in total, but the number of GCOs in each type can vary. As also shown in Table~\ref{compatible-Table}, only two types of color orderings support GFDs in the same class as ${\cal T}_A$. So do GCOs for ${\cal T}_E$. While GFDs in the classes defined by ${\cal T}_B$, ${\cal T}_C$, ${\cal T}_D$, ${\cal T}_F$ are compatible with three types of GCOs. Finally, ${\cal T}_G$ has the very interesting property of being compatible with color orderings of any type and so it is universal.

If we are interested only in computing a biadjoint partial amplitude $m_n(\Sigma,\Sigma)$, with $\Sigma$ a particular GCO, and a seed is known that is compatible with $\Sigma$, we can adopt a simpler bootstrap.

{\bf Bootstrap II:}
\begin{enumerate}\itemsep=0pt
\item[(1)] Start with a GFD as a seed that is compatible with a certain color ordering $\Sigma$.
 \item[(2)] For every GFD in the list, degenerate it and blow it up in a way compatible with $\Sigma$.
\item[(3)] Repeat step 2 until no new GFDs are produced.
\end{enumerate}

 This bootstrap can be thought of as a generalization of that in \cite{Borges:2019csl} and \cite{Cachazo:2019xjx} to any other type of color orderings.
In Section~\ref{chirotopalsection}, we use this construction as a motivation for introducing {\it chirotopal tropical Grassmannians} as a natural extension of positive tropical Grassmannians.

\subsubsection[Application to (3,7)]{Application to $\boldsymbol{(3,7)}$}

Starting at $(3,7)$, there are GFDs with quartic or higher degree vertices. One example was given in~\eqref{boundaryofT0}. We can still apply the bootstrap algorithms to trees with higher multiplicities but efficiency might be compromised as special attention is needed to find the flips of such GFDs. Fortunately, computations are still within reach with modest computational resources for $(3,7)$ and $(3,8)$ as the number of GFDs with mixed trees is still small.

For $(3,7)$, we reproduced all $211\,155+210 = 211\,365$ GFDs for $(3,7)$ presented in \cite{herrmann2008draw}. Here~$210$ is the number of GFDs with mixed trees. All GFDs fall into $93+1$ classes where 93 of them only contain cubic vertices. The only class with quartic vertices is the one generated by \eqref{boundaryofT0} via relabeling. As mentioned in Section \ref{sec4}, this exceptional GFD is a codim-1 boundary of another arrangement, \eqref{37T0}, whose metric has~7 independent internal lengths. Such an arrangement failed to be a GFD because it is not compatible with any GCOs and has the wrong dimension.

On the one hand, the eleven color orderings in Table~\ref{table37} have 693, 534, 563, 447, 541, 509, $520+2$, $556+2$, 393, $440+1$ and $423+1$ GFDs respectively. Here $423+1$ means that the type X GCO has 423 compatible GFDs with only cubic vertices and one GFD with quartic vertices. These are the numbers of GFDs needed to compute the corresponding \smash{$m_7^{(3)}(\Sigma,\Sigma)$} partial amplitudes. On the other hand, every GFD is compatible with exactly 64 color orderings which may belong to $2,3, \dots $ or $11$ types respectively.

Here is an example of a GFD which is compatible with 11 types of GCOs,
 \def\fFD #1,#2,#3,#4,#5,#6 \fFD 	{
 \raisebox{-.85 cm}{
\tikz[scale=.3]{
\draw[]
(0,0)--(3,0)
(0,0)--(-.3,-1.3) node[below]{#1}
(0,0)--(-.3,1.3) node[above]{#2}
(1,0)-- (1,1.3) node[above]{#3}
(2,0)-- (2,1.3) node[above]{#4}
(3,0)--(3.3,1.3) node[above]{#5}
(3,0)--(3.3,-1.3) node[below]{#6}
;
}}
}
\begin{align*}
&\left(
\fFD 2, 6, 3, 5, 4, 7 \fFD,\fFD 1, 6, 3, 5, 4, 7 \fFD,\fFD 2, 6, 1, 5, 4, 7 \fFD,\fFD 5, 6, 2, 1, 3, 7 \fFD,\fFD 1, 3, 2, 4, 6, 7 \fFD,\fFD 1, 2, 3, 4, 5, 7 \fFD,\fFD 3, 4, 1, 2, 5, 6 \fFD\right),
\end{align*}
whose contribution to the amplitude is given by
${1}/(
\s_{126}
\s_{347}
\s_{567}
\t_{1236}
\sfR_{47, 123, 56}
\sfW_{1347265}
 )$ with $\sfR$,~$\sfW$ defined in \eqref{defpoles}.

All representatives of the $(3,7)$ GFDs, accompanied by one compatible GCO each, their poles, and their contributions to the amplitudes are put in an ancillary file.

\subsubsection[Application to (3,8)]{Application to $\boldsymbol{(3,8)}$}
\label{38GFD1}

For $(3,8)$, as one can imagine, there would be more GFDs with mixed vertices, whose flips are much more complicated than those of GFDs with only cubic vertices. So in practice we apply the bootstrap I more effectively by bootstrapping all GFDs with only cubic vertices first, and then bootstrapping the remaining GFDs with mixed vertices based on them.

Besides, there is another problem for $(3,8)$. Given a GFD it is very time-consuming to determine all GCOs it is compatible with. The reason is that there are $4\, 445\, 640$ $(3,8)$ GCOs. Fortunately, there is a simple method to generate all compatible GCOs for a given GFD based on one of its compatible GCOs. We postpone the explanation of this method to the next section and now we just apply it to our bootstrap. In most cases, two GFDs connected by flips share a common GCO. So in the bootstrap, whenever we generate a new candidate of GFD ${\cal T}'$ by blowing up a degeneration of a GFD ${\cal T}$, we can check whether it is compatible with any of the compatible set of GCOs of~${\cal T}$. Otherwise, we have to check whether it is compatible with any of the $4\, 445\, 640$ $(3,8)$ GCOs.

In this way, we obtained 4734 classes of normal GFDs with pure cubic vertices first. Their flips produce 55 new classes of arrangements of metric trees with mixed vertices and a correct number of independent lengths. 28 of them share a common GCO with the normal GFDs whose flips produce them. For the last 27 candidates, we found that 3 of them are GFDs by checking whether they are compatible with any of the $4\, 445\, 640$ $(3,8)$ GCOs.
It turns out that the flips of the 31 classes of GFDs don't produce any new classes of GFDs.

In one word, we found $4734+31= 4765$ classes of GFDs in total, whose permutations give~${116\,849\,565+ 604800 = 117\,454\,365}$ GFDs.

All representatives of the $(3,8)$ GFDs, accompanied by one compatible GCO each, their poles, and their contributions to the amplitudes are put in an ancillary file.

We have checked many $(3,8)$ GFDs and verified that they are compatible with 256 GCOs. We conjecture it to be true for all $(3,8)$ GFDs.
In the next section, we will explain how to get all compatible GCOs for each given GFD efficiently and here we present some relevant results. It turns out a~$(3,8)$ GFD always contributes to at least 4 but at most 112 types of GCOs.

Let us now assume that Bootstrap I has already been performed and we have obtained all GFDs and their compatible GCOs. Next, we describe an efficient way to compute a given partial amplitude $m(\Sigma,\Sigma)$ which replaces Bootstrap II.

This method works for general cases but for definiteness, let us concentrate on the present interest of this section, $(k,n)=(3,8)$. By assumption, we have all compatible GCOs for a~representative of each of the 4765 classes of GFDs. We can decompose them as a list of pairs, one compatible GCO and one representative GFD, $(\Sigma',{\cal T})$. Then for the GCO of interest, $\Sigma$, we just need to select all pairs such that the GCO $\Sigma'$ is of the same type as $\Sigma$. This means that for each selected pair $(\Sigma',{\cal T})$, there exists one or more permutations of labels $\rho$ such that $\Sigma'\to \Sigma'|_{\rho}= \Sigma$. We relabel the pair $(\Sigma',{\cal T})$ simultaneously under every $\rho$ to get a set of $(\Sigma, {\cal T}|_{\rho})$. Gathering all distinct ${\cal T}|_{\rho}$ obtained in this way, we find all compatible GFDs for every given GCO $\Sigma$.

As shown in \cite{Cachazo:2019xjx}, there are 13612 GFDs contributing to a type~0 GCO, and in this paper, we find that this is the maximum number among all classes of GCOs. In particular, among these, the type 18 GCO given in Table~\ref{38color} has the smallest number 3356 of compatible GFDs. A~general code to give any $(3,8)$ color ordering amplitudes is provided in an ancillary Mathematica notebook file.

Let us point out an interesting contrast. As above, we find that there are 4765 permutation classes of GFDs, each being compatible with some GCO; by using the metric tree parameterization of the Dressian this translates to having 4765 permutation classes of maximal cones in the tropical Grassmannian $\operatorname{Trop} G(3,8)$.

Now, in \cite[Theorem~4.6]{bendle2020parallel}, the authors find 4766 symmetry classes of maximal cones in the tropical Grassmannian of~$(3,8)$. We have identified the extra symmetry class of metric tree arrangement in~\eqref{37T038}. It is not compatible with any GCO. On the other hand, in \cite{bendle2020parallel} the cone parameterized by our metric tree arrangement corresponds to a certain non-binomial saturated initial ideal, see \cite[Remark~4.5]{bendle2020parallel}.

\section{Generating GCOs from GFDs using twists} \label{poletoGCO}

In the previous section, we discussed how to generate new GFDs starting from known ones by using flips guided by generalized color orderings. In every example we studied, all GFDs up to $(3,8)$ have $2^{2(n-4)}$ compatible GCOs, which naturally leads us to conjecture that it holds in general.

\begin{conj}\label{samenumber}
Every $(3,n)$ GFD is compatible with $2^{2(n-4)}$ color orderings.
\end{conj}

This section is devoted to providing what we believe is a promising direction to prove this important conjecture. In fact, we turn things around and use GFDs to generate new generalized color orderings. Very nicely, this also helps us to improve the bootstrap of GFDs as already applied in Section~\ref{38GFD1}.

Let us start with a simple observation.

\begin{prop}\label{k2CO}
 Let $T$ be a $n$-point tree Feynman diagram in a $\phi^3$ scalar field theory. $T$ is compatible with exactly $2^{n-3}$ color orderings.
\end{prop}

Before providing the proof, let us define a useful operation on the planar embeddings of a~tree.

\begin{defn}
Given a tree graph $T$ embedded on a plane, a twist along an internal edge $f$ of $T$ is done by taking one of the two subtrees obtained by deleting $f$ from $T$ and reflecting it along the line defined by $f$. The new embedding of $T$ is said to be a twist of the old one along~$f$.
\end{defn}

Now we can prove Proposition \ref{k2CO}.

\begin{proof}
Any tree, $T$, in a $\phi^3$ theory has $n-3$ internal edges and $n$ leaves. Draw $T$ on a plane and read the order in which the leaves appear. This gives one color ordering. Let us denote it by $\sigma_0$. For each internal edge, there is a twist associated with it. Applying a twist generates a different embedding of $T$ on the plane and hence a different color ordering. The number of all possible compositions of any number of twists is clearly $2^{n-3}$, hence the number of color orderings.
\end{proof}

We would like to generalize the construction above to $(k,n)$ Feynman diagrams. Let us illustrate the procedure for $k=3$.

Consider a $(3,n)$ Feynman diagram, ${\cal T}=(T_1,T_2,\dots ,T_n)$, by definition, ${\cal T}$ is compatible with at least one GCO, $\Sigma_0=(\sigma_1,\sigma_2, \dots ,\sigma_n)$.

Recall that ${\cal R}({\cal T})$ is a rational function in the kinematic invariants $\s_{abc}$. The number of poles in ${\cal R}({\cal T})$ satisfies $n_P\geq 2(n-4)$. Each pole is produced by a one dimensional integral in the space internal lengths along a particular direction. Let $t\in [0,\infty)$ be the parameter along one of such directions. When $t\to \infty$ all internal lengths are either ${\cal O}(t)$ or ${\cal O}(1)$.

Now consider the embedding of each $T_i$ in ${\cal T}$ on a plane according to the $\sigma_i$ ordering. Select one of the $n_P$ possible directions as $t\to \infty$ and twist $T_i$ along the internal edges which are ${\cal O}(t)$. The resulting procedure is a $(3,n)$ twist to ${\cal T}$.

This means that the GFD ${\cal T}$ has $n_P$ twists. However, since the space of internal lengths is~\smash{$\mathbb{R}_+^{2(n-4)}$}, there is a set of $2(n-4)$ twists that generates the rest.

Now we can perform the counting of GCOs associated with ${\cal T}$. For each of the $2(n-4)$ independent twists, one gets a new GCO. The number of all possible compositions of any number of such twists is clearly \smash{$2^{2(n-4)}$}, hence the number of $(3,n)$ color orderings.

This is a strong argument in favor of Conjecture \ref{samenumber}, which states that every $(k=3,n)$ GFD is compatible with \smash{$2^{2(n-4)}$} color orderings.

Let us illustrate the discussion above with a $(3,7)$ GFD which has nine poles. This is an example taken from \cite[Section 3.2]{Borges:2019csl},
 \def\fFD #1,#2,#3,#4,#5,#6 \fFD 	{
 \raisebox{-.85 cm}{
\tikz[scale=.3]{
\draw[]
(0,0)--(3,0)
(0,0)--(-.3,-1.3) node[below]{#1}
(0,0)--(-.3,1.3) node[above]{#2}
(1,0)-- (1,1.3) node[above]{#3}
(2,0)-- (2,1.3) node[above]{#4}
(3,0)--(3.3,1.3) node[above]{#5}
(3,0)--(3.3,-1.3) node[below]{#6}
;
}}
}
\begin{align}\label{T99}
\left(
\fFD 3,4,2,5,6,7 \fFD,\fFD 3,4,1,5,6,7 \fFD,\fFD 1,2,4,7,5,6 \fFD,\fFD 1,2,3,7,5,6 \fFD,\fFD 1,2,7,6,3,4 \fFD,\fFD 1,2,7,5,3,4 \fFD,\fFD 1,2,6,5,3,4 \fFD\right) .
\end{align}
If the internal lengths of each tree diagram in the collection are ordered from left to right, then their expressions can be recorded in a $3\times 7$ matrix with $i^{\rm th}$ column \smash{$\big[f_1^{(i)},f_2^{(i)},f_3^{(i)}\big]^{\mathsf{T}}$},
\[
\left[
\begin{matrix}
 x & x & y & y & p+v & p & p \\
 w & w & r & r & u & u+v & v \\
 v & v & u & u & s & s & s+u \\
\end{matrix}
\right]
\]
with
\begin{gather}\label{conInf}
x + w = r + s + u, \qquad r + y = p + v + w.
\end{gather}

The way to find the directions that correspond to a pole is by taking any of the variables, sending it to infinity, and solving the constraints \eqref{conInf} in all possible ways recalling that all variables must be positive. For example, sending $x\to \infty$ in $x + w = r + s + u$ implies that either~$r$,~$s$, or $u$ must also be sent to infinity. Choosing $r\to \infty$ and using $r + y = p + v + w$ implies that $p$, $v$ or $w$ must be sent to infinity. However, $w$ appears in $x + w = r + s + u$ invalidating the choice. This means that the directions defined by $\{ x, r, p \}$ or $ \{x, r, v\}$ give rise to poles.
A~short exercise reveals only nine possibilities, as expected,\footnote{Sending several internal lengths to infinity is equivalent to shrinking the other internal lengths to zero. Hence the nine directions defined by \eqref{twistingransformationsset}
can also be obtained by finding out all dim-1 degenerations of the GFD, which is another way to get the poles of the amplitudes. See more details in \cite{Guevara:2020lek}.}
\[
\{ \{x, r, p\}, \{x, r, v\}, \{x, s\}, \{x, u\}, \{w, r\}, \{w, s, y\}, \{w, u, y\}, \{y, p\}, \{y, v\} \}.
\]
Now we can think of each variable as defining a $k=2$ twist on the trees where it appears. It is clear that any such twist squares to the identity, i.e., $x^2=\mathbb{I}$, $y^2=\mathbb{I}$, etc. Moreover, any two commute, i.e., $xy=yx$. Each of the $n_P=9$ allowed directions becomes a valid $(3,7)$ twist transformation on the $(3,7)$ color ordering. So, we have
\be
\label{twistingransformationsset}
\{ xrp, xrv, xs, xu, wr, wsy, wuy, yp, yv\}.
\ee
It is easy to show that these nine are in fact generated by only six,
\be\label{twist99}
\{ xrp, xs, xu, wr, yp, yv\}.
\ee
For example, $xrv = (xrp)(yp)(yv)$.

Now, using that the GFD \eqref{T99} is compatible with the $(3,7)$ color ordering that descends from $(1234567)$, one can apply any combination of the operations in \eqref{twist99} to produce a total of~$2^6=64$ GCOs.

In this section, we have explained an efficient way to find all compatible GCOs for a given GFD. As explained at the end of Section~\ref{38GFD1}, once this is done, it is easy to carry out the opposite procedure, i.e., finding all GFDs compatible with a given GCO.

We emphasize here again that all $(3,8)$ partial amplitudes were computed using these techniques and they are consistent with those computed from CEGM integrals \cite{Cachazo:2023ltw}, which in turn provides more support for Conjecture~\ref{samenumber}.

\section[(3,n) minimal scalar amplitudes]{$\boldsymbol{(3,n)}$ minimal scalar amplitudes}\label{sec8}

The standard biadjoint theory can be thought of as a theory of multiple massless scalar fields with interactions constrained by a Lie algebraic structure. A theory with a simpler lagrangian is obtained by considering a single massless scalar field $\phi$ and no Lie algebra structure. In addition, one can set the interaction to be the simplest non-trivial possible one, i.e., a $\phi^3$ term. We call this the $k=2$ minimal scalar theory.

The tree-level amplitudes of this minimal scalar theory are computed by summing over all Feynman diagrams with cubic interactions. For $n$ external points, there are $(2n-5)!!$ such diagrams. It is natural to ask whether there is a way of obtaining \smash{$m_{n}^{(2)\,\min }$} from \smash{$m_{n}^{(2)}(\alpha,\beta )$}. In~\cite{Dolan:2014ega}, Dolan and Goddard noticed that
\be\label{DGdef}
m_{n}^{(2)\, \min } = \frac{1}{2^{n-3}}\sum_{\sigma \in {\rm CO}_{2,n}}m_{n}^{(2)}(\sigma,\sigma),
\ee
where ${\rm CO}_{2,n}$ is the set of all $(2,n)$ color orderings.

In this section, we propose two independent definitions for \smash{$m_{n}^{(3)\, \min}$} which become equivalent assuming Conjecture \ref{samenumber}. The first is as a sum over all $(3,n)$ GFDs while the second is the analog of \eqref{DGdef},
\be\label{ourDef}
m_{n}^{(3)\, \min } = \frac{1}{2^{2(n-4)}}\sum_{\Sigma \in {\rm CO}_{3,n}}m_{n}^{(3)}(\Sigma,\Sigma).
\ee

Recall that Conjecture \ref{samenumber} states that each $(3,n)$ GFD is compatible with exactly $2^{2(n-4)}$ color orderings and so each GFD appears the same number of times in the sum \eqref{ourDef}.

In the remaining of this section we study properties of \smash{$m_{n}^{(3)\, \min }$} for $n=6,7$ which show why it is the correct definition.

\subsection[(n-5)-dimensional residues]{$\boldsymbol{(n-5)}$-dimensional residues}
In \cite{Cachazo:2022vuo}, a simple but surprising connection between \smash{$m^{(3)}_n(\Sigma_0,\Sigma_0)$}, with $\Sigma_0$ the $(3,n)$ color ordering that descends from $\sigma_0 = (1,2,\dots ,n)$, and \smash{$m^{(2)}_n(\sigma_0,\sigma_0)$} was proposed. The proposal is that certain $(n-5)$-dimensional residues of \smash{$m^{(k)}_n(\Sigma_0,\Sigma_0)$} are equal to \smash{$m^{(2)}_n(\sigma_0,\sigma_0)$}.

For $(3,6)$, any residue of \smash{$m^{(3)}_6(\Sigma_0,\Sigma_0)$} at a poles of the form $1/\s_{a,a+1,a+2}$ leads to the $14$ Feynman diagrams that compute \smash{$m^{(2)}_6(\sigma_0,\sigma_0)$}, while for $(3,7)$, a two-dimensional residue defined by the zeroes of $\{ \s_{a,a+1,a+2},\t_{a-1,a,a+1,a+2}\} $ lead to the $42$ Feynman diagrams in \smash{$m^{(2)}_7(\sigma_0,\sigma_0)$}.

Using our definition, \eqref{ourDef}, we found that the residue of \smash{$m_{6}^{(3)\, \min }$} at any pole of the form $1/s_{abc}$, gives rise to the sum over $105$ Feynman diagrams which reproduces \smash{$m_{6}^{(2)\, \min }$}. Even more surprising is that for $n=7$ any two-dimensional residue defined by $\{ \s_{abc},\t_{abcd}\}$ of~\smash{$m_{7}^{(3)\, \min }$} gives rise to the sum over the $945$ Feynman diagrams which reproduces \smash{$m_{7}^{(2)\, \min }$}.

We have checked that if the sum in \eqref{ourDef} were replaced by any subset of color orderings, for example, by only those of type 0, then the residues would not agree with the corresponding~$m_{n}^{(2)\, \min }$ amplitudes.

\subsection{Comments on residues}

One of the most striking properties of \smash{$m^{(3)}_6(\Sigma_0,\Sigma_0)$} is that it contains a new class of poles with no direct analog in $k=2$ amplitudes,
the so-called $\sfR$-pole. Just as the pole is novel to $k=3$, so is the behavior of the residue of \smash{$m^{(3)}_6(\Sigma_0,\Sigma_0)$} at $\sfR=0$. The residue at $\sfR=0$ becomes the product of three \smash{$m^{(2)}_4(\sigma_0,\sigma_0)$} amplitudes. This $3$-split is achieved in codimension one and hence the novelty. Note that \smash{$m^{(2)}_4(\sigma_0,\sigma_0)$} only has two terms (two planar Feynman diagrams) and hence the straightforward computation of the residue from the GFDs gives rise to eight contributions whose sum factors into three amplitudes.

It turns out that \smash{$m_{6}^{(3)\, \min }$} exhibits an even more surprising behavior. The residue where some $\sfR_{a_1a_2,b_1b_2,c_1c_2}=0$ is product of three \smash{$m_{4}^{(2)\, \min }$} amplitudes, each is made out of three terms. This means that $3^3=27$ contributions from GFDs conspire to perfectly produce the factorization.

Once again, this behavior is not observed if any proper subset of color orderings is used in the definition. Note that this is a non-trivial statement, as a sum over all color orderings of a~given type would be permutation invariant and hence a reasonable object by itself.

In light of the recent work \cite{Early:2022mdn} by one of the authors on factorization for the standard globally planar CEGM amplitudes, the behavior of \smash{$m^{(3)}_6$} on the ``$\sfR$ pole'' is expected to generalize very beautifully to higher $k$. Many mysteries remain unanswered about factorization; most relevantly, these include investigating residues of the CEGM amplitudes considered in this work. Such questions are left to the future.

\section[Higher k color orderings]{Higher $\boldsymbol{k}$ color orderings}\label{sechigherkGCO}

In this paper, the main focus is on $k=3$ color orderings
but they can be straightforwardly generalized to higher $k$.

\begin{defn}\label{GCOkkkk}
A $(k,n)$ generalized {\it color ordering} is an ${n\choose k-2}$-tuple
\[
\Sigma^{[k]} = \big\{ \sigma^{(i_1,i_2,\dots,i_{k-2})}\mid \{i_1,\dots, i_{k-2}\} \subset [n] \big\},
\]
where $\sigma^{(i_1,i_2,\dots,i_{k-2})}$ is a $(2,n-k+2)$ color ordering constructed as follows.

Let
$\{ H_{1},H_{2},\dots ,H_{n} \}$
be an arrangement of $n$ projective $({k-}2)$-planes in generic position in~$\mathbb{RP}^{(k-1)}$. Intersecting any $(k-2)$ such $H$'s, $\{ H_{i_1},H_{i_2},\dots,H_{i_{k-2}}\! \}$, produces a line, $L^{(i_1,i_2,\dots,i_{k-2})}$. The line so defined intersects the remaining $(n{-}k{+}2)$ $H$'s each on a point, resulting in a sequence of points on the line which defines a $(2,n{-}k{+}2)$
color ordering $\sigma^{(i_1,i_2,\dots, i_{k-2})}$.
\end{defn}

 By definition, removing a $(k{-}2)$-plane, say $H_{i}$, from the arrangement with $n>k+2$ must result in another arrangement but with $(n-1)$ $(k{-}2)$-planes. Therefore, the operation must give a $(k,n-1)$ color ordering. This
 generalizes the $k$-preserving projection given in \eqref{colorprojection},
\be
\label{projectionkkkk}
\pi_i\bigl(\Sigma ^{[k]}\bigr): = \bigl\{ \pi_i\bigl(
\sigma^{(i_1,i_2,\dots,i_{k-2})}
\bigr)
\mid \{i_1,\dots, i_{k-2}\} \subset [n] \setminus \{i\}
\bigr\}.
\ee

On the other hand, in Definition \ref{GCOkkkk}, we chose to construct $\Sigma^{[k]}$ out of $(2,n{-}k{+}2)$ color orderings. However, it is sometimes convenient to
note that since each $H_i$ is an $\mathbb{RP}^{k-2}\subset \mathbb{RP}^{k-1}$, then $H_i\cap H_j$ is an $({k-}3)$-plane for all $j\in [n]\setminus \{i\}$. This means that we have an arrangement of $n-1$ $({k-}3)$-planes in $\mathbb{RP}^{k-1}$, i.e., a $(k{-}1,n-1)$ color ordering, which we call a $k$-decreasing projection
\begin{gather}
\label{component}
\pi_{(i)}\big(\Sigma ^{[k]}\big)
\equiv \big\{ \sigma^{(i,i_2,\dots,i_{k-2})}\mid \{i_2,\dots, i_{k-2}\} \subset [n] \setminus \{i\} \big\}.
\end{gather}
Clearly, we have\footnote{The union implies that duplicates are not included.}
\begin{gather}
\label{columnnnn}
\Sigma^{[k]}=\bigcup_{i=1}^n
\pi_{(i)}\big(\Sigma ^{[k]}\big).
\end{gather}

Let us also extend the notion of descendant. For any $(k,n)$, there is always a set of generalized color orderings
with a very special property.
\begin{defn}\label{globalGCO2}
A $(k,n)$ generalized color ordering $\Sigma^{[k]} = \bigl\{ \sigma^{(i_1,i_2,\dots,i_{k-2})}\mid \{i_1,\dots, i_{k-2}\} \subset [n] \bigr\}$ is said to descend from a $(2,n)$ color ordering $\sigma$ if $\sigma^{(i_1,i_2,\dots,i_{k-2})}=\pi_{i_1}\pi_{i_2}\cdots \pi_{i_{k-2}}(\sigma)$. We also say that $\Sigma^{[k]}$ is a descendant of $\sigma$.
\end{defn}

\subsection{General pseudo-GCOs}

Equations \eqref{columnnnn} and \eqref{projectionkkkk} suggest a recursive way to get the $(k,n)$ color ordering and similar to $k=3$ pseudo-GCO defined just below Theorem \ref{pseudoLines}, we can define a general $k$ pseudo-GCO.

\begin{defn}\label{generalpseudoGCO}
 A ${n \choose k-2}$-tuple of standard color orderings with $n>k+2$ is said to be a~$(k,n)$ pseudo-GCO if all its $k$-preserving projections are $(k,n-1)$ pseudo-GCOs, while $(k,k+2)$ pseudo-GCOs are all descendants of $(2,k+2)$ color orderings.
\end{defn}

Just as an arrangement of $n$ projective $(k-2)$-planes for a GCO, one can also find an arrangement of $n$ projective $(k-2)$-pseudo-planes for a pseudo-GCO. If there are
$k$ such pseudo-planes intersecting at the same point
when all such pseudo-planes are straightened, we call the corresponding
pseudo-GCO as a non-realizable one. Otherwise, it is a GCO.

One can also define the $k$-decreasing projection of a pseudo-GCO just as \eqref{columnnnn}.
\begin{thm}
\label{theoremcondition1}
Each $k$-decreasing projection of pseudo-GCO is also a pseudo-GCO.
\end{thm}

\begin{proof}
Let's denote an $n$-pt pseudo-GCO as $\Sigma $ and its $k$-decreasing projection as $\pi_{(j)}(\Sigma)$ for any $j\in [n]$.
Obviously, the theorem is true for $n=k+2$.
Suppose it's true for $(n-1)$-pts, which means any $k$-decreasing projection $
\pi_{(j)}\left(
\pi_i\left(\Sigma \right)
\right)
$ with $j\neq i$ of the $k$-preserving projection~$
\pi_i\left(\Sigma \right)
$ is a~pseudo-GCO. Note that the $k$-decreasing projection
$
\pi_{(j)}\left(
\pi_i\left(\Sigma \right)
\right) =
\pi_i\left(\pi_{(j)}(\Sigma)
\right)
$
is also a~$k$-preserving projection of $\pi_{(j)}(\Sigma) $. Hence any $k$-preserving projection of $\pi_{(j)}(\Sigma)$ is a pseudo-GCO, which means $\pi_{(j)}(\Sigma) $ itself is also a pseudo-GCO according to Definition~\ref{generalpseudoGCO}.
By mathematical induction, we proved the theorem.
\end{proof}

\begin{thm}
\label{theoremcondition2}
An ${n \choose k-2}$-tuple of standard color orderings with $k\geq4$ is a pseudo-GCO if any of its $k$-preserving projection is a pseudo-GCO.
\end{thm}
The proof is similar.

\subsection[Duality between (k,n) and (n-k,n) GCOs]{Duality between $\boldsymbol{(k,n)}$ and $\boldsymbol{(n-k,n)}$ GCOs}

Obviously, all $(n-2,n)$ GCOs are just descendants of $(2,n)$ GCOs, based on which we can construct the general duality between $(k,n)$ and $(n-k,n)$ GCOs. Given any $(k,n)$ GCO $\Sigma^{[k]}$, we can get its dual $(n-k,n)$ GCO $\Sigma'^{[n-k]}$ by
\ba
\label{dualityeq}
\Sigma^{[k]} \sim \Sigma'^{[n-k]}=
\bigl\{
{\rm Dual}\big( \pi_{i_1} \pi_{i_2} \cdots \pi_{i_{n-k-2}}
\Sigma^{[k]} \big)
\mid \{i_1,i_2,\dots,i_{n-k-2}\} \subset [n]
\bigr\},
\ea
where we perform a projection of $n-k-2$ labels of $\Sigma^{[k]}$ first, leading to a $(k,k+2)$ GCO, \smash{$\pi_{i_1} \pi_{i_2} \cdots \pi_{i_{n-k-2}}
\Sigma^{[k]} $}, whose dual GCO \smash{${\rm Dual}\big( \pi_{i_1} \pi_{i_2} \cdots \pi_{i_{n-k-2}}
\Sigma^{[k]}\big)$} is just a $(k+2)$-pt standard color ordering \smash{$\sigma^{(i_1i_2,\dots, i_{n-k-2})}$}, which constitutes the dual $(n-k,n)$ GCO \smash{$\Sigma'^{[n-k]}$}.

The dual of non-realizable pseudo-GCO is also a non-realizable pseudo-GCO and can be obtained similarly.

The duality \eqref{dualityeq} can be easily proved by mathematical induction.
\begin{proof}
First, we want to prove \eqref{dualityeq} holds at the level of pseudo-GCOs.
Obviously, it holds for $n=k+2$.
Supposing the duality already holds for $(n-1)$-pts,
i.e.,
\begin{gather*}
\pi_{j}\big(\Sigma^{[k]}\big) \sim \pi_{j}\big(\Sigma'^{[n-k]}\big)\\
\qquad=
\bigl\{
{\rm Dual}\big( \pi_{i_1} \pi_{i_2} \cdots \pi_{i_{n-k-3}}\pi_{j}
\Sigma^{[k]} \big)
\mid \{i_1,i_2,\dots,i_{n-k-3}\} \subset [n] \setminus \{j\}
\bigr\},
\qquad
\forall j\in [n],
\end{gather*}
which means any projection of $\Sigma'^{[n-k]}$ is a valid $(k,n-1)$ pseudo-GCO and hence confirms that $\Sigma'^{[n-k]}$ is pseudo-GCO. Since all $(k,n)$ and all $(n-k,n)$ pseudo-GCOs are in bijection and the dual of non-realizable pseudo-GCO is also a non-realizable pseudo-GCO, we prove the duality~\eqref{dualityeq} holds at the level of GCOs.
\end{proof}

Based on \eqref{dualityeq}, it's clear to see a $k$-preserving projection of $\Sigma^{[k]}$ is dual to a $k$-decreasing projection of its dual $\Sigma'^{[n-k]}$,
\[
\pi_{j}\big(\Sigma^{[k]} \big) \sim \pi_{(j)}\big(\Sigma'^{[n-k]} \big),\qquad \pi_{j}\big(\Sigma'{}^{[n-k]}\big) \sim \pi_{(j)}\big(\Sigma^{[k]}\big) .
\]

The first non-trivial duality is the one between $(3,6)$ GCOs themselves where a GCO of any type is dual to another GCO of the same type.

The duality allows us to get all $(4,7)$ GCOs for free from those of $(3,7)$ GCOs. Here we list a $(4,7)$ GCO as an example,
\ba
\label{eg47}
&((3 5 7 4 6),(2 4 6 7 5),(2 6 3 5 7),(2 3 6 7 4),(2 4 3 7 5),(2 3 6 5 4),(1 4 6 7 5),
\nonumber
\\
&(1 6 3 5 7),(1 3 6 7 4),(1 4 3 7 5),(1 3 6 5 4),(1 2 5 7 6),(1 2 6 7 4),(1 2 7 5 4),
\nonumber
\\
&(1 2 6 5 4),(1 2 3 7 6),(1 2 7 5 3),(1 2 3 5 6),(1 2 3 4 7),(1 2 3 4 6),(1 2 3 4 5)),
\ea
which is dual to the last $(3,7)$ GCO of type X in Table~\ref{table37}. We see the projection of the $(3,7)$ GCO with respect to both 1 and 2 is $((4 6 5 7),(3 5 7 6),(3 6 4 7),(3 4 7 5),(3 5 4 6))$, which is dual the first entry of \eqref{eg47}, $(3 5 7 4 6)$. The dual GCOs for the remaining types are put in the ancillary file.

The $(4,8)$ GCOs are dual to themselves, so we have to work them out independently, which is explained in the next subsection.

\subsection[Bootstrapping k=4 pseudo-GCOs]{Bootstrapping $\boldsymbol{k=4}$ pseudo-GCOs}\label{Bootstrappingk4gco}

The recursive definition of the pseudo-GCOs strongly suggests we use a bootstrap method to generate them. Indeed we reproduced all $(4,7)$ GCOs and generated all $(4,8)$ pseudo-GCOs.

Let us illustrate the idea by starting with an ansatz of $(4,n)$ color orderings of the form,
\[
\Lambda^{[4]}=
\big\{
\lambda^{(1,2)}(34\cdots n), \lambda^{(1,3)}(24\cdots n), \dots, \lambda^{(n-1,n)}(12\cdots n-2)
\big\},
\]
where $\lambda^{(i,j)}$ is a $(2,n-2)$ color ordering. Also, note that we reserve the notation $\Sigma$ and $\sigma$ for valid color orderings and so we have used $\Lambda$ and $\lambda$ for the ansatz.

In \eqref{component}
and \eqref{projectionkkkk}, we have defined two operations which act on generalized color orderings but they can obviously be generalized to act on any ansatz as well. To explain the operations more intuitively,
it is useful to present $\Lambda^{[4]}$ in a slightly redundant way as an $n\times n$ symmetric matrix,
\be
\label{matrixform}
\Lambda^{[4]}\sim
\begin{pmatrix}
0& \lambda^{(2,1)} & \lambda^{(3,1)} & \cdots & \lambda^{(n,1)} \\
\lambda^{(1,2)} & 0 & \lambda^{(3,2)} & \cdots & \lambda^{(n,2)} \\
\lambda^{(1,3)} & \lambda^{(2,3)} & 0 & \cdots & \lambda^{(n,3)} \\
\vdots & \vdots &\vdots &\ddots &\vdots
\\
\lambda^{(1,n)} & \lambda^{(2,n)} & \lambda^{(3,n)} & \cdots & 0
\end{pmatrix},
\ee
where $\lambda^{(i,j)} =\lambda^{(j,i)}$ and we have suppressed the dependence on the $n-2$ labels.

Now one can interpret \eqref{matrixform} as a collection of $k=3$ GCOs, with \smash{$\pi_{(i)}(\Lambda )$} corresponding to the $i^{\rm th}$ row or column of
\eqref{matrixform}, while its projection $\pi_i\left(\Lambda \right)$ corresponds to an $(n-1)\times (n-1)$ submatrix obtained by deleting the $i^{\rm th}$ row and column of \eqref{matrixform} and projecting out all label $i$ in the remaining submatrix.

Now one can impose the following two conditions on the ansatz $\Lambda ^{[4]}$ to get a candidate for a~color ordering:
\begin{enumerate}\itemsep=0pt
 \item[(a)] Its $k$-preserving projection $\pi_i\left(\Lambda\right) $
 for any label $i$
is a valid $(4,n-1)$ pseudo-GCO.
 \item[(b)] Its $k$-decreasing projection $\pi_{(j)}(\Lambda) $ for any label $j$
is a valid $(3,n-1)$ pseudo-GCO.
\end{enumerate}

According to Theorems \ref{theoremcondition1}
 and \ref{theoremcondition2}, there two conditions are equivalent. In practice, the second kind provides a faster way to find out candidates.

For $(4,7)$ color orderings, we start with an ansatz of a $7\times 7$ matrix in the form \eqref{matrixform}.
Requiring that each of its row or column corresponds to a valid $(3,6)$ color ordering leads to exactly $27\, 240$ choices of
\[
\big\{\lambda^{(1,2)},\lambda^{(1,3)},\dots, \lambda^{(6,7)}\big\}.
\]
 Hence we reproduced all $(4,7)$ GCOs. The $27\, 240$ $(4,7)$ color orderings fall into 11 types as expected.

The evaluation of
bootstrapping $(4,8)$ pseudo-GCOs is
combinatorially more involved but does not pose any conceptual challenges. We got 2628 types of pseudo-GCOs in total, which is consistent with the results in the literature \cite{fukuda2013complete} where they claim there are 2,604 types of realizable uniform matroids and 24 types of non-realizable ones.
Comparing results from both sides, we distinguish the realizable and non-realizable pseudo-GCOs and we present them in the ancillary files separately.
The
permutations of 2,604 GCOs give 100 086 840 distinct ones in total. We tally the numbers of their permutations here,
 \begin{gather*}
\{\{1680, 2\}, \{2520, 1\}, \{3360, 1\}, \{5040, 6\}, \{6720, 3\}, \{10080, 16\},\\ \qquad \{13440, 10\}, \{20160, 183\}, \{40320, 2382\}\},
\end{gather*}
 e.g., there are
two types with 1680 elements, i.e., with a symmetry group of order 24. Similarly, we present the numbers of distinct permutations of 24 non-realizable pseudo-GCOs here,
 \begin{gather*}
 \{ \{1680, 1\}, \{3360, 3\}, \{5040, 3\}, \{6720, 3\}, \{10080, 5\},\\
 \qquad \{13440, 3\},
\{20160, 3\}, \{40320, 3\} \} ,
\end{gather*}
which gives 319 200 distinct non-realizable pseudo-GCOs in total.

The 2604 types of $(4,8)$ GCOs are dual to themselves, which is a consistency check of our results. We mention that a $(4,8)$ GCO might be dual to another GCO of a different type. The~24 types of non-realizable pseudo-GCOs are also dual to themselves.

\section[Higher k Feynman diagrams]{Higher $\boldsymbol{k}$ Feynman diagrams}\label{sechigherkGFD}

In parallel to the above section, we discuss the $k$ Feynman diagrams here.
The generalization of Definition~\ref{metricTreeArrangement} is straightforward.

\begin{defn}[\cite{herrmann2008draw}]\label{metricTreeArrangementg}
A $(k,n)$ arrangement of metric trees is an ${n \choose k-2}$-tuple
\begin{align*}
 {\cal T}^{[k]} = \{T_{i_1,i_2,\dots,i_{k-2}}\mid \{i_1,i_2,\dots,i_{k-2} \} \subset [n]\}
\end{align*}
 such that $T_{i_1,i_2,\dots,i_{k-2}}$ is a metric tree with $n-k+2$ leaves in the set $[n]\setminus \{i_1,i_2,\dots,i_{k-2} \} $ and metric \smash{$d^{(i_1,i_2,\dots,i_{k-2})}_{ab}$} so that the following compatibility condition is satisfied
\be\label{compatibleGeneral}
d^{(i_3,i_4,\dots,i_k)}_{i_1,i_2} = d^{(i_2,i_4,\dots,i_k)}_{i_1,i_3} =\cdots =d^{(i_1,i_2,\dots,i_{k-2})}_{i_{k-1},i_k} , \qquad \forall \{i_1,i_2,\dots,i_k\}\subset [n].
\ee
Denote by $d$ the symmetric tensor with entries $d_{i_1,i_2,\dots,i_k}:=d^{(i_3,i_4,\dots,i_k)}_{i_1,i_2}$.
\end{defn}

In what follows, motivated by our construction of $k=3$ GFDs, we present conditions that are necessary in order for an arrangement of metric trees, ${\cal T}^{[k]}$, to define generalized Feynman diagrams for $k\ge 4$.
\begin{itemize}\itemsep=0pt
 \item There are exactly $(k-1)(n-k-1)$ independent internal edge lengths after imposing the compatibility conditions \eqref{compatibleGeneral}.
 \item There exists at least one GCO $\Sigma^{[k]} = \big\{ \sigma^{(i_1,i_2,\dots,i_{k-2})}\mid \{i_1,\dots, i_{k-2}\} \subset [n] \big\}$ such that \smash{$T_{i_1,i_2,\dots,i_{k-2}}$} is planar with respect to \smash{$\sigma^{(i_1,i_2,\dots,i_{k-2})}$} for all $\{i_1,\dots, i_{k-2}\} \subset [n]$. In this case, we say that ${\cal T}^{[k]}$ is compatible with $\Sigma^{[k]}$.
\end{itemize}

Moreover, the rational function associated to ${\cal T}^{[k]}$ is
\begin{align}
{\cal R}({\cal T}) :={}& \int_{\mathbb{R}^+} \prod_{I=1}^{(k-1)(n-k-1)}{\rm d}f_I \prod_{J=(k-1)(n-k-1)+1}^{{n \choose k-2}(n-k-1)}\theta (f_J(f_1,\dots ,f_{(k-1)(n-k-1)}))\nonumber
\\
&
\times
\exp\Biggl(-\sum_{\{i_1,i_2,\dots,i_{k} \} \subset [n] }\s_{i_1,i_2,\dots,i_{k}}\, d_{i_1,i_2,\dots,i_{k}}\Biggr). \label{Scw3ge}
\end{align}
The conditions in the integrand, $\theta (f_J(f_1,\dots ,f_{(k-1)(n-k-1) }))$, simply enforce that all internal lengths must be non-negative.
$\s_{i,i,i_3,\dots,i_{k}}$ are generic completely symmetric rank-$k$ tensors subject to $\s_{i,i,i_3,\dots,i_{k}}=0$ and $\sum_{\{i_2,\dots,i_{k} \} \subset [n] \setminus \{i\}}\s_{i,i_2,\dots,i_{k}}$=0 for any $i$.

A $k$-decreasing projection of a GCO is also a GCO.
However, a $k$-decreasing projection of~${\cal T}^{[k]}$
\begin{gather*}
\pi_{(i)}({\cal T}) = \{ {\cal T}_{i,i_2,\dots,i_{k-2}}\mid \{i_2,\dots, i_{k-2}\} \subset [n] \setminus \{i\} \}.
\end{gather*}
is not necessarily a GFD. On the one hand, if ${\cal T}$ is compatible with $\Sigma$, by definition, $\pi_{(i)}({\cal T}) $ must be compatible with $\pi_{(i)}({\Sigma}) $. On the other hand,
the conditions \eqref{compatibleGeneral} are not sufficient to guarantee that $\pi_{(i)}({\cal T}) $ has enough independent internal lengths. So $\pi_{(i)}({\cal T}) $ could also be a~degenerate GFD.

Similarly, we define the projections of arrangements of metric trees as follows:
\begin{align*}
 \pi_i({\cal T}) = \{\pi_i( T_{i_1,i_2,\dots,i_{k-2}})\mid  \{i_1,\dots,i_{k-2} \} \subset [n] \setminus \{i\}\} ,
\end{align*}
which might be GFDs or their degenerations.

\begin{defn}\label{globalGFD-New}
A $(k,n)$ arrangement of metric trees
\[
{\cal T}^{[k]} = \{T_{i_1,i_2,\dots,i_{k-2}}\mid  \{i_1,i_2,\dots,i_{k-2} \} \subset [n]\}\]
 is said to descend from a metric $T$ if $T_{i_1,i_2,\dots,i_{k-2}}=\pi_{i_1}\pi_{i_2}\cdots \pi_{i_{k-2}} (T)$. We also say that ${\cal T}^{[k]}$ is a degree-$k$ descendant of $T$.
\end{defn}

Note that here the metric tree $T$ could be a cubic tree or its degeneration, i.e., a Feynman diagram with quartic or higher degree vertices.

In particular, we say the standard metric tree $T$ with $n=k+2$ external leaves and its degree-$k$ descendant are dual to each other.

In parallel to \eqref{dualityeq}, we can explain the general duality between $(k,n)$ and $(n-k,n)$ GFDs. Given any $(k,n)$ GFD ${\cal T}^{[k]}$, we conjecture that its dual $(n-k,n)$ GFD ${\cal T}'^{[n-k]}$ is given by
\begin{gather}
\label{dualityeqGFD}
{\cal T}^{[k]} \sim {\cal T}'^{[n-k]}=
\big\{
{\rm Dual}\big( \pi_{i_1} \pi_{i_2} \cdots \pi_{i_{n-k-2}}
{\cal T}^{[k]} \big)
\mid \{i_1,i_2,\dots,i_{n-k-2}\} \subset [n]
\big\},
\end{gather}
where we project out $n-k-2$ labels of ${\cal T}^{[k]}$ first, leading to a $(k,k+2)$ arrangement of metric trees, $\pi_{i_1} \pi_{i_2} \cdots \pi_{i_{n-k-2}}
{\cal T}^{[k]} $, whose dual arrangement, \smash{${\rm Dual}\big( \pi_{i_1} \pi_{i_2} \cdots \pi_{i_{n-k-2}}
{\cal T}^{[k]}\big)$}, is just a $(k+2)$-pt standard Feynman diagram \smash{$T_{(i_1i_2\dots i_{n-k-2})}$} with possible quartic or higher degree vertices, which constitutes the dual $(n-k,n)$ GFD ${\cal T}'^{[n-k]}$.

Based on \eqref{dualityeqGFD}, it is clear to see a projection of ${\cal T}^{[k]}$ is dual to a codim-1 component of its dual ${\cal T}'^{[n-k]}$,
\[
\pi_{i}\big({\cal T}^{[k]} \big) \sim \pi_{(i)}\big({\cal T}'^{[n-k]} \big),\qquad \pi_{i}\big({\cal T}'{}^{[n-k]}\big) \sim \pi_{(i)}\big({\cal T}^{[k]} \big),
\]
which was the original way to define the duality between planar GFDs in \cite{Cachazo:2019xjx}.

\subsection[(4,7) GFDs]{$\boldsymbol{(4,7)}$ GFDs}
Using duality, one can get all $(4,7)$ GFDs from those of $(3,7)$.
For example, the dual $(4,7)$ GFD of the $(3,7)$ one~\eqref{boundaryofT0} with both degree-three and degree-four vertices is given in Figure~\ref{fig:my_labelGFD47}.

\begin{figure}[htb!]
 \centering
\def\fFD #1,#2,#3,#4,#5 \fFD {
 \raisebox{-.61 cm}{
\tikz[scale=.17]{
\draw[]
(0,0)--(2,0)
(0,0)--(-.3,-1.3) node[below]{#1}
(0,0)--(-.3,1.3) node[above]{#2}
(2,0)--(2.3,1.3) node[above]{#3}
(2,0)--(2.3,-1.3) node[below]{#4}
(1,0)-- (1,1.3) node[above]{#5}
;
}}
}
\def\fFDD #1,#2,#3,#4,#5 \fFDD {
 \raisebox{-.61 cm}{
\tikz[scale=.17]{
\draw[]
(0,0)--(1,0)
(0,0)--(-.3,-1.3) node[below]{#1}
(0,0)--(-.3,1.3) node[above]{#2}
(1,0)--(2,0) node[right]{#3}
(1,0)--(1,-1.3) node[below]{#4}
(1,0)-- (1,1.3) node[above]{#5}
;
}}
}
$$
\begin{pmatrix}
0 & \fFD 3, 5, 4, 7, 6 \fFD & \fFD 2, 5, 4, 6, 7 \fFD & \fFD 3, 6, 2, 7, 5 \fFD & \fFD 2, 3, 6, 7, 4 \fFD & \fFD 3, 4, 5, 7, 2 \fFD & \fFD 2, 4, 5, 6, 3 \fFD \\ \fFD 3, 5, 4, 7, 6 \fFD & 0 & \fFD 1, 5, 6, 7, 4 \fFD & \fFDD 1, 7, 3, 5, 6 \fFDD & \fFDD 1, 3, 4, 6, 7 \fFDD & \fFDD 3, 7, 1, 4, 5 \fFDD & \fFD 1, 4, 3, 6, 5 \fFD \\ \fFD 2, 5, 4, 6, 7 \fFD & \fFD 1, 5, 6, 7, 4 \fFD & 0 & \fFD 1, 6, 5, 7, 2 \fFD & \fFD 1, 2, 4, 7, 6 \fFD & \fFD 1, 4, 2, 7, 5 \fFD & \fFD 4, 5, 2, 6, 1 \fFD \\ \fFD 3, 6, 2, 7, 5 \fFD & \fFDD 1, 7, 3, 5, 6 \fFDD & \fFD 1, 6, 5, 7, 2 \fFD & 0 & \fFDD 3, 7, 1, 2, 6 \fFDD & \fFDD 1, 3, 2, 5, 7 \fFDD & \fFD 1, 2, 3, 5, 6 \fFD \\ \fFD 2, 3, 6, 7, 4 \fFD & \fFDD 1, 3, 4, 6, 7 \fFDD & \fFD 1, 2, 4, 7, 6 \fFD & \fFDD 3, 7, 1, 2, 6 \fFDD & 0 & \fFDD 1, 7, 2, 3, 4 \fFDD & \fFD 3, 4, 1, 6, 2 \fFD \\ \fFD 3, 4, 5, 7, 2 \fFD & \fFDD 3, 7, 1, 4, 5 \fFDD & \fFD 1, 4, 2, 7, 5 \fFD & \fFDD 1, 3, 2, 5, 7 \fFDD & \fFDD 1, 7, 2, 3, 4 \fFDD & 0 & \fFD 2, 3, 1, 5, 4 \fFD \\ \fFD 2, 4, 5, 6, 3 \fFD & \fFD 1, 4, 3, 6, 5 \fFD & \fFD 4, 5, 2, 6, 1 \fFD & \fFD 1, 2, 3, 5, 6 \fFD & \fFD 3, 4, 1, 6, 2 \fFD & \fFD 2, 3, 1, 5, 4 \fFD & 0
\end{pmatrix}
$$
 \caption{A $(4,7)$ GFD which is present in a more redundant way by a symmetric matrix such that the Feynman diagram in the $i^{\rm th}$-row and $j^{\rm th}$-column has leaves $i$, $j$ pruned.}
 \label{fig:my_labelGFD47}
\end{figure}

One can check $\pi_1\pi_2\eqref{boundaryofT0}$ is dual to the second Feynman diagram in the first row of the $(4,7)$ GFD. Similarly, a $k$-preserving projection of the $(4,7)$ GFD is dual to a Feynman diagram in~\eqref{boundaryofT0}.
The $(4,7)$ GFD in Figure~\ref{fig:my_labelGFD47}
is compatible with the $(4,7)$ GCO
 dual to \eqref{37GCOquar},
 \begin{align*}
\left(
\begin{matrix}
 0 & \text{(35647)} & \text{(25467)} & \text{(26357)} & \text{(23476)} & \text{(24375)} & \text{(23654)} \\
 \text{(35647)} & 0 & \text{(15467)} & \text{(16357)} & \text{(13476)} & \text{(14375)} & \text{(13654)} \\
 \text{(25467)} & \text{(15467)} & 0 & \text{(12576)} & \text{(12476)} & \text{(12754)} & \text{(12654)} \\
 \text{(26357)} & \text{(16357)} & \text{(12576)} & 0 & \text{(12376)} & \text{(12753)} & \text{(12356)} \\
 \text{(23476)} & \text{(13476)} & \text{(12476)} & \text{(12376)} & 0 & \text{(12347)} & \text{(12346)} \\
 \text{(24375)} & \text{(14375)} & \text{(12754)} & \text{(12753)} & \text{(12347)} & 0 & \text{(12345)} \\
 \text{(23654)} & \text{(13654)} & \text{(12654)} & \text{(12356)} & \text{(12346)} & \text{(12345)} & 0 \\
\end{matrix}
\right),
 \end{align*}
which was also present in a symmetric but redundant way. The contribution of the $(4,7)$ GFD to the amplitudes according to \eqref{Scw3ge} is given by $ {1}/({\s_{3457} \s_{2367} \s_{1567} \s_{1346} \s_{1247} \s_{1235}})$.

Remind that all $(3,6)$ GFDs just contain cubic vertices. So
it's obvious that some rows or columns of the $(4,7)$ GFD in Figure~\ref{fig:my_labelGFD47} are just degenerated $(3,6)$ GFDs, i.e., even though some rows or columns may not have enough independent internal lengths to become a GFD in those rows or columns, the whole matrix have enough independent internal lengths to become a $k=4$ GFD.

Because of this subtlety, it is conceptually more complicated to bootstrap all $k=4$ GFDs than what we have done for $k=4$ GCOs in Section~\ref{Bootstrappingk4gco}.
If we start with a $7\times 7$ matrix of Feynman diagrams as an ansatz and require each row or column to be a $(3,6)$ GFD, we get 93 classes of $(4,7)$ GFDs as well as one extra class of matrices of Feynman diagrams which have seven independent internal lengths and are dual to the class of \eqref{37T0}. Only when we degenerate it, do we get the $(4,7)$ GFD in Figure~\ref{fig:my_labelGFD47}. The 94 classes of $(4,7)$ GFDs obtained in this way are indeed dual to those of $(3,7)$ in the way explained in \eqref{dualityeqGFD} and we present them in an ancillary file.

\section{Future directions}\label{sec10}

In this work, we defined generalized color orderings (GCOs) and started the study of their properties. Combining generalized Feynman diagrams (GFDs) and GCOs, we finally constructed the complete color dressed generalized biadjoint amplitudes. Perhaps a surprising feature of generalized biadjoint amplitudes is that their GFDs are collections of trees that are not necessarily~$\phi^3$ Feynman diagrams. In fact, one should consider Feynman diagrams in a theory that with all possible powers of $\phi$. In mathematical terminology, generic trees in the collection might not be trivalent: a priori any metric tree may be a member of the arrangement leading up to the GFD.

This work is the first one of a series of papers where we start the study of a ``triality'' among partial $(k,n)$ biadjoint amplitudes (as defined in this work using GFDs), CEGM integrands and their integrals on the configuration space of $n$ points in $\mathbb{CP}^{k-1}$ \cite{Cachazo:2023ltw}, and new objects we call {\it chirotopal tropical Grassmannians}.

The mathematics of arrangements of metric trees is naturally connected to tropical geometry, in particular to the Dressian and the tropical Grassmannian \cite{herrmann2008draw}.
We end this work with a~preview of some of the directions on chirotopal tropical Grassmannians that will be explored in the future.

The tropical Grassmannian is a complicated object. However, it contains a relatively simple object known as the positive part, $\operatorname{Trop}^+G(k,n)$. We will argue that $\operatorname{Trop}^+G(k,n)$ is nothing but one of a family of objects, each determined by a chirotope \cite{bjorner1999oriented}, and whose study seems to be within reach.

Given the importance of such a family of objects, we give a preview of their definition here.

\subsection{Chirotopal tropical Grassmannians\label{chirotopalsection}}

The tropical Grassmannian $\operatorname{Trop} G(k,n)$, introduced in \cite{speyer2004tropical}, parametrizes realizable tropical linear spaces; it is the tropical variety of the Plucker ideal of the Grassmannian $G(k,n)$. While $\operatorname{Trop} G(2,n)$ is completely characterized by the tropicalization of the 3-term tropical Plucker relations, for general $(k,n)$ the Plucker ideal contains higher degree generators and to calculate $\operatorname{Trop} G(k,n)$ quickly becomes a completely intractable problem. On the other hand, in \cite{speyer2005tropical}, Speyer--Williams introduced the \textit{positive} tropical Grassmannian, which was later shown \cite{Arkani-Hamed:2020cig,speyer2021positive} to be characterized by 3-term tropical Plucker relations,
\[\pi_{Lac} + \pi_{Lbd} = \min \{\pi_{Lab} + \pi_{Lcd},\pi_{Lad} + \pi_{Lbc}\},\]
which depends on the given global cyclic order $(1,2,\dots, n)$.

In this section, motivated by the observation that the CEGM formula for the generalized biadjoint scalar with integrand the (squared) $k$-Parke--Taylor factor (that is, the canonical function on the nonnegative Grassmannian \cite{Arkani-Hamed:2012zlh}) is equal to the Laplace transform of the positive tropical Grassmannian, we define the \textit{chirotopal} tropical Grassmannian $\operatorname{Trop}^\chi G(k,n)$, by relaxing the requirement that the cyclic order be global; by this we mean that we replace the usual notion of the cyclic order on $G(k,n)$ with certain compatible collections of $\binom{n}{k-2}$ cyclic orders, which we called generalized color orderings in Definition~\ref{GCOkkkk}.

Therefore we not only generalize the positive tropical Grassmannian to other realizable oriented uniform matroids, but we also present two a priori completely different ways to represent~it. The first is purely combinatorial and uses generalized Feynman diagrams as discussed in this work, while the second uses the CEGM formula to reconstitute the cones from higher dimensional residues as in \cite{Early:2022mdn} in the context of factorization.

\begin{defn}[\cite{speyer2004tropical}]
 Given $e = (e_1,\dots, e_N) \in \mathbb{Z}_{\ge 0}^N$, denote $\mathbf{x}^e = x_1^{e_1}\dots x_N^{e_N}$. Let $E \subset\mathbb{Z}^N_{\ge 0}$. If $f = \sum_{e\in E} f_e \mathbf{x}^e$ is nonzero, denote by $\operatorname{Trop} (f)$ the set of all points $(X_1,\dots, X_N)$ such that for the collection of numbers \smash{$\sum_{i=1}^N e_i X_i$} for $e$ ranging over $E$, then the minimum of the collection is achieved at least twice. We say that $\operatorname{Trop} (f)$ is the tropical hypersurface associated to $f$. The \emph{tropical Grassmannian} $\operatorname{Trop} G(k,n)$ is the intersection of all tropical hypersurfaces $\operatorname{Trop} (f)$ where $f$ ranges over all elements in the Plucker ideal.

 The \emph{Dressian} $\operatorname{Dr}_{k,n}$ is the tropical (pre)variety obtained by tropicalizing only the 3-term Plucker relations.\footnote{The term Dressian was coined in \cite{herrmann2008draw}, but $\operatorname{Dr}_{k,n}$ was called the tropical pre-Grassmannian in \cite{speyer2004tropical}.}
\end{defn}

Thus, the Dressian $\operatorname{Dr}(k,n)$ consists of all tropical Plucker vectors; a tropical Plucker vector is said to be realizable if it is in the tropical Grassmannian $\operatorname{Trop}G(k,n)$.

We first define the \textit{chirotopal} Dressian, as a generalization of the positive Dressian and then intersect with the tropical Grassmannian in order to obtain our main definition, the chirotopal tropical Grassmannian.

\begin{defn}\label{defn: chirotopal trop Grass}

Given any generic point in the real Grassmannian $G(k,n)$, that is where all maximal $k\times k$ minors $\Delta_J$ are nonzero, let \smash{$\chi \in \{-1,1\}^{\binom{n}{k}}$} be defined coordinate-wise by{\samepage
\[\chi_J = \operatorname{sgn}(\Delta_J).\]
Any such vector $\chi$ arising in this way is called a realizable, (uniform) \emph{chirotope} \cite{bjorner1999oriented}.}

Given a realizable chirotope $\chi \in \{-1,1\}^{\binom{n}{k}}$, a point $\pi \in \mathbb{R}^{\binom{n}{k}}$ is said to be a \emph{$\chi$-tropical Pl\"{u}cker vector} provided that, for any \smash{$L \in \binom{\lbrack n\rbrack}{k-2}$} and any cyclic order $(j_1,j_2,j_3,j_4)$ in \smash{$\binom{\lbrack n\rbrack\setminus L}{4}$} such that
\[
\left(\frac{\chi_{Lj_1j_2}\chi_{Lj_3j_4}}{\chi_{Lj_1j_3}\chi_{Lj_2j_4}},\frac{\chi_{Lj_1j_4}\chi_{Lj_2j_3}}{\chi_{Lj_1j_3}\chi_{Lj_2j_4}} \right) = (1,1),
\]
	then
\[\pi_{Lj_1j_3} + \pi_{Lj_2j_4} = \min \{\pi_{Lj_1j_2} + \pi_{Lj_3j_4}, \pi_{Lj_1j_4} + \pi_{Lj_2j_3}\}.\]

Here we denote by $\operatorname{Dr}^\chi_{k,n}$ the \emph{$\chi$-Dressian}, consisting of all $\chi$-tropical Plucker vectors.
\end{defn}
\begin{defn}
 Fix a chirotope $\chi \in \{-1,1\}^{\binom{n}{k}}$ as usual. The $\chi$-tropical Grassmannian $\operatorname{Trop}^\chi G(k,n)$ is the set of all realizable $\chi$-tropical Plucker vectors, i.e., it is the intersection of the chirotopal Dressian with the tropical Grassmannian, $\operatorname{Trop}^\chi G(k,n) = \operatorname{Dr}^\chi_{k,n} \cap \operatorname{Trop} G(k,n)$.
\end{defn}
\begin{rem}
 Note that we are not making assertions about the Groebner fan. For that, a~possible related proposal was made in \cite[Conjecture 7.1]{bendle2020parallel}.

 Also we emphasize that both the chirotopal Dressian $\operatorname{Dr}_{k,n}$ and $\chi$-tropical Grassmannian $\operatorname{Trop}^\chi G(k,n)$ depend only on the reorientation class of $\chi$, that is, they are invariant under the torus action, say $t_j\colon \chi_J \mapsto -\chi_J$ whenever $j\in J$, and otherwise $\chi_J \mapsto \chi_J$.
\end{rem}
It is natural to ask whether chirotopal tropical Plucker vectors are always realizable, as is the case for positive tropical Plucker vectors.

Our conjecture, below, proposes to generalize to all realizable chirotopes (as in Definition~\ref{defn: chirotopal trop Grass}) the following recent characterization of the positive tropical Grassmannian.
\begin{thm}[\cite{Arkani-Hamed:2020cig,speyer2021positive}]
 The positive tropical Grassmannian is completely characterized by the $3$-term tropical Plucker relations, that is we have
 \[\operatorname{Trop}^+G(k,n) = \operatorname{Dr}^+_{k,n},\]
 where ``$+$'' is the standard notation for the chirotope $\chi$ with all entries $+1$, that is $\chi\! =\! (1,1,\dots ,1)$.
\end{thm}
Below, in Conjecture \ref{conj: chirotopal equality}, in order to be consistent with the CEGM formula, we are modding out \smash{$\mathbb{R}^{\binom{n}{k}}$} by the lineality space, which consists of all vectors with coordinates
\[\pi_J = \sum_{j\in J} x_j\]
for $x\in \mathbb{R}^n$.

We speculate that chirotopal tropical Grassmannians are better behaved than both the full Dressian and the tropical Grassmannian. Conjecture \ref{conj: chirotopal equality} ventures to assert that every $k=3$ chirotopal tropical Plucker vector is already the tropicalization of a linear space: it is already in the tropical Grassmannian.

\begin{conj}\label{conj: chirotopal equality}
 Each $\operatorname{Dr}^\chi_{3,n}$ is a pure\footnote{A polyhedral complex is pure if all maximal cones have the same dimension.} $(3-1)(n-3-1)$-dimensional polyhedral fan.
 Moreover, we have the equality
 \[\operatorname{Trop}^\chi G(3,n) = \operatorname{Dr}^\chi_{3,n},\]
 that is, chirotopal tropical Plucker vectors are realizable.
 Finally, if we fix a given maximal cone in $\operatorname{Trop} G(3,n)$ then either it is not contained in any chirotopal tropical Grassmannian, or it is contained in exactly $2^{(3-1)(n-3-1)}$ of them.\footnote{Note that this is equivalent to Conjecture \ref{samenumber}.}
\end{conj}
We hope, in stating the conjecture, to stimulate progress around an important question that we believe should be rigorously investigated.

We emphasize that for $k\ge3$, the tropical Grassmannian is \textit{not} in general covered by chirotopal tropical Grassmannians, as we have seen in $\operatorname{Trop}G(3,8)$!

Indeed, as noted in \cite{bendle2020parallel} there is a permutation class of maximal simplicial cones\footnote{For example, with rays spanned by the eight vectors $e^{126},e^{135},e^{178},e^{237},e^{248},e^{346},e^{457},e^{568} \in \mathbb{R}^{\binom{n}{8}}$.} in the tropical Grassmannian $\operatorname{Trop} G(3,8)$ which do not belong to any chirotopal tropical Grassmannian, see the GFD in equation~\eqref{37T038}. These cones are characterized by a (saturated) initial ideal which is not binomial \cite[Remark 4.5]{bendle2020parallel}.

This motivates the following important question.

\begin{question}
 Is there a simple characterization of chirotopal tropical Plucker vectors in terms of properties of the Plucker ideal?
\end{question}

Let us summarize some evidence in support of our conjecture.
\begin{enumerate}\itemsep=0pt
 \item[(1)] Using generalized color orderings, we have confirmed the classification of the realizable uniform chirotopes for $n=6,7,8,9$ which is given in \url{https://www-imai.is.s.u-tokyo.ac.jp/~hmiyata/oriented_matroids/}.
 \item[(2)] Using the data from (1), we have made a highly nontrivial numerical validation: for $k=3$ and $n=6,7$ we find exact agreement between our combinatorial (GFD) expressions and the values obtained from the CEGM integrals for all 4 and 11 types of GCOs, respectively. For $n=8$, the numerical evaluation of CEGM integrals is much more difficult given the large number of solutions to the scattering equations; nonetheless, we find agreement well within the margin of error.
\end{enumerate}

We also have computed the maximal cones, parametrized with generalized Feynman diagrams, or in mathematical terminology, metric tree arrangements subject to the additional requirement they must be compatible with at least one GCO. For example, the latter requirement can be seen to remove the seven-dimensional cones in the Dressian $\operatorname{Dr}_{3,7}$, though for a given seven-dimensional cone, its seven codimension one facets remain: each belongs to some chirotopal tropical Grassmannian.

It is of course natural to ask the even more ambitious question whether (or to what extent) Conjecture~\ref{conj: chirotopal equality} can be extended to larger~$k$.
\begin{question}
 Do the statements in Conjecture~\ref{conj: chirotopal equality} hold if we replace $k=3$ with $k\ge 4$? Are chirotopal tropical Plucker vectors realizable in general?
\end{question}

A natural strategy to approach the proof of Conjecture \ref{conj: chirotopal equality} would be to try to generalize the method of proof used in \cite{Arkani-Hamed:2020cig} and in \cite{speyer2021positive} for the so-called positive configuration space; both proofs rely on the existence of a certain \textit{surjectively positive} parameterization \cite{pachter2004tropical} of the Grassmannian, as in \cite{postnikov2006total}. Unfortunately, it does not seem obvious how to find such parameterizations for components of configuration spaces going beyond the positive configuration space to reorientation classes of other oriented uniform matroids, which suggests that some new ideas may be required.

It is also important to note that when $\chi$ is not isomorphic to the standard positive chirotope, then the finest regular matroid subdivisions that are induced by a $\chi$-tropical Plucker vector do not in general saturate Speyer's f-vector theorem \cite{speyer2009matroid}, as can be seen already in $\operatorname{Trop}^\chi G(3,6)$. For example, the collection of trees $\mathcal{T}_F$ in Table \ref{amplitude36con} induces a matroid subdivision with only five maximal cells, one less than the maximum $\binom{6-2}{2}=6$ for regular matroid subdivisions. This is compatible with the following observation: the vertices for the largest cell label the 16 basis elements of the graphic matroid for the complete graph $K_4$, whose presence as some face of the subdivision has been shown in \cite{speyer2009matroid} to be the condition under which the f-vector of a regular matroid subdivision is not maximized.

In \cite{Cachazo:2023ltw}, we exploit the CEGM formulation and its relation to $X(3,n)$ to construct irreducible decoupling identities. Imposing that such decoupling identities have realizations in terms of GCO is a powerful clue in the quest to finding their explicit realization in terms of some Lie algebraic structure or generalization thereof.

\appendix

 \section[31 GCOs in a (3,6) decoupling set]{31 GCOs in a $\boldsymbol{(3,6)}$ decoupling set} \label{31termsec}

\allowdisplaybreaks

 Below is a list of $31$ GCOs participating a $(3,6)$ decoupling grouped by types:
\begin{align}
&((23456),(13456),(12456),(12356),(12346),(12345)),\nonumber\\
& ((23465),(13465),(12465),(12365),(12346),(12345)),
\nonumber
\\
&((23645),(13645),(12645),(12365),(12364),(12345)),\nonumber\\
& ((25436),(15436),(12645),(12635),(12634),(12345)),
\nonumber
\\
&((23456),(15436),(15426),(15326),(14326),(12345)),
\nonumber
\\[3mm]
&
((25436),(15436),(12456),(12356),(12346),(12543)),
\nonumber\\
& ((23456),(13456),(12456),(12365),(12364),(12354)),
\nonumber
\\
&
((23645),(13645),(12645),(12356),(12346),(12354)),\nonumber\\
& ((23456),(13456),(12645),(12635),(12634),(12543)),
\nonumber
\\
&
((23456),(13645),(12645),(15326),(14326),(13245)),
\nonumber\\
& ((25436),(13645),(15426),(12635),(12634),(13245)),
\nonumber
\\
&
((23645),(15436),(15426),(12365),(12364),(13245)),\nonumber\\
& ((23456),(13465),(12465),(12365),(14326),(14325)),
\nonumber
\\
&
((25436),(13456),(15426),(15326),(14326),(12543)),
\nonumber\\
& ((23465),(13456),(12456),(12356),(14326),(14325)),
\nonumber
\\
&
((25436),(15436),(12465),(12365),(12634),(12435)),
\nonumber\\
& ((23645),(13645),(12465),(12635),(12364),(12435)),
\nonumber
\\
&
((23465),(13465),(12645),(12635),(12346),(12435)),
\nonumber\\
& ((23465),(13465),(12465),(12356),(12364),(12354)),
\nonumber
\\
&
((23465),(15436),(15426),(15326),(12346),(14325)),
\nonumber
\\[3mm]
&
((23465),(13465),(12456),(12356),(12634),(12534)),
\nonumber\\
& ((23456),(13456),(12465),(12365),(12634),(12534)),
\nonumber
\\
&
((25436),(13465),(15426),(15326),(12634),(13425)),
\nonumber\\
& ((23456),(13465),(12645),(12635),(14326),(13425)),
\nonumber
\\
&
((23645),(15436),(15426),(12356),(12346),(13254)),
\nonumber\\
& ((23645),(13645),(12456),(12635),(12634),(12453)),
\nonumber
\\
&
((25436),(15436),(12456),(12365),(12364),(12453)),
\nonumber\\
& ((23645),(13456),(12456),(15326),(14326),(13254)),
\nonumber
\\
&
((23465),(13645),(12645),(15326),(12346),(14235)),
\nonumber\\
& ((23645),(13465),(12465),(15326),(12364),(14235)),
\nonumber
\\[3mm]
&
((23645),(13465),(12456),(15326),(12634),(13524)),
\label{31termseceq}
\end{align}
 where the first GCO in each group is also shown in Table~\ref{table36} whose arrangement of lines is present in Figure~\ref{36figure} which manifestly reduces to the common arrangement of five lines in Figure~\ref{fiveNaive} when line~6 is removed.

\section[(3,8) color orderings]{$\boldsymbol{(3,8)}$ color orderings}\label{38colorfactors}

In this appendix, we complete the
135 types of GCOs for $(3,8)$ in Table~\ref{38color}.
 The second column provides a representative that can be used to obtain the rest by applying permutations of labels. The last column contains the number of distinct permutations.

\begin{center}
\small
\begin{longtable}{|@{\,}c@{\,}|@{\,}c@{\,}|@{\,}c@{\,}|}
\caption{All 135 types of $(3, 8)$ color orderings.} \label{38color}\\

\hline  Type  &  Color ordering representative & $\#$ \\ \hline
\endfirsthead

\multicolumn{3}{c}%
{{\bfseries \tablename\ \thetable{} -- continued from previous page}} \\
\hline Type & Color ordering representative & $\#$ \\ \hline
\endhead

\hline \multicolumn{3}{|r|}{{Continued on next page}} \\ \hline
\endfoot

\hline \hline
\endlastfoot

0 & ((2345678), (1345678), (1245678), (1235678), (1234678), (1234578), (1234568), (1234567)) & 2520 \\ 1 & ((2345678), (1345678), (1245678), (1238576), (1238476), (1237845), (1236845), (1236754)) & 40320 \\ 2 & ((2345867), (1345867), (1284567), (1283567), (1283476), (1234758), (1234658), (1254376)) & 40320 \\ 3 & ((2345867), (1345867), (1284567), (1283756), (1283746), (1237458), (1236548), (1254376)) & 20160 \\ 4 & ((2345678), (1345678), (1245678), (1235678), (1234867), (1234857), (1234856), (1234765)) & 40320 \\ 5 & ((2345678), (1345678), (1284567), (1283567), (1283476), (1283475), (1283465), (1276543)) & 20160 \\ 6 & ((2345678), (1345687), (1245687), (1235687), (1234876), (1234875), (1564328), (1564327)) & 40320 \\ 7 & ((2345678), (1384567), (1284567), (1765328), (1674328), (1574328), (1564328), (1324567)) & 10080 \\ 8 & ((2345678), (1345678), (1245678), (1235678), (1234876), (1234875), (1234865), (1234765)) & 10080 \\ 9 & ((2345678), (1345678), (1245678), (1235678), (1234678), (1234587), (1234586), (1234576)) & 20160 \\ 10 & ((2345678), (1345678), (1245678), (1238567), (1238476), (1238475), (1238465), (1237654)) & 40320 \\ 11 & ((2345678), (1345678), (1248567), (1238567), (1283476), (1283475), (1283465), (1276534)) & 20160 \\ 12 & ((2345678), (1345687), (1245687), (1238567), (1238476), (1238475), (1564328), (1456327)) & 40320 \\ 13 & ((2345678), (1345687), (1245687), (1238576), (1238476), (1238745), (1546328), (1456327)) & 40320 \\ 14 & ((2345678), (1345867), (1245867), (1238756), (1238746), (1547328), (1456328), (1453267)) & 40320 \\ 15 & ((2345678), (1348567), (1248567), (1238576), (1674328), (1547328), (1546328), (1432567)) & 40320 \\ 16 & ((2345678), (1348567), (1248567), (1238756), (1647328), (1547328), (1456328), (1432567)) & 20160 \\ 17 & ((2345687), (1345687), (1284567), (1283567), (1283476), (1283475), (1234658), (1265437)) & 40320 \\ 18 & ((2345687), (1348576), (1248576), (1237865), (1467328), (1453278), (1485326), (1473256)) & 40320 \\ 19 & ((2345867), (1345867), (1284567), (1283576), (1283476), (1237458), (1236458), (1254376)) & 40320 \\ 20 & ((2345867), (1384567), (1284567), (1765328), (1674328), (1234758), (1234658), (1542376)) & 40320 \\ 21 & ((2348567), (1765438), (1765428), (1675328), (1234876), (1237485), (1236485), (1432576)) & 40320 \\ 22 & ((2384567), (1384567), (1284567), (1235678), (1234786), (1234785), (1234658), (1237564)) & 40320 \\ 23 & ((2765438), (1765438), (1245687), (1235876), (1234876), (1238745), (1283645), (1254637)) & 40320 \\ 24 & ((2765438), (1765438), (1245867), (1235876), (1234876), (1283745), (1283645), (1254367)) & 40320 \\ 25 & ((2345678), (1345678), (1245678), (1235687), (1234876), (1234875), (1238465), (1237465)) & 40320 \\ 26 & ((2345678), (1345678), (1245678), (1235867), (1234867), (1238457), (1238456), (1237645)) & 40320 \\ 27 & ((2345678), (1345678), (1245687), (1235687), (1234876), (1234875), (1283465), (1273465)) & 40320 \\ 28 & ((2345678), (1345687), (1245687), (1235687), (1234867), (1234857), (1654328), (1564327)) & 20160 \\ 29 & ((2345678), (1345687), (1284567), (1283567), (1283476), (1283475), (1564328), (1345627)) & 40320 \\ 30 & ((2345678), (1345867), (1284567), (1283567), (1283476), (1574328), (1564328), (1345267)) & 40320 \\ 31 & ((2345678), (1348567), (1284567), (1283567), (1674328), (1574328), (1564328), (1342567)) & 40320 \\ 32 & ((2345687), (1345687), (1245687), (1235867), (1234786), (1238475), (1234685), (1236475)) & 40320 \\ 33 & ((2345687), (1345687), (1245867), (1235867), (1234786), (1283475), (1234685), (1263475)) & 20160 \\ 34 & ((2345867), (1345867), (1245678), (1235678), (1234768), (1283475), (1283465), (1267345)) & 40320 \\ 35 & ((2345867), (1345867), (1245867), (1235678), (1234768), (1238475), (1238465), (1236745)) & 40320 \\ 36 & ((2348567), (1384567), (1284567), (1765328), (1234876), (1234875), (1234865), (1423567)) & 40320 \\ 37 & ((2384567), (1384567), (1245687), (1283576), (1283476), (1283745), (1238645), (1245637)) & 40320 \\ 38 & ((2384567), (1384567), (1284567), (1235678), (1234876), (1234785), (1234685), (1235764)) & 40320 \\ 39 & ((2345678), (1345678), (1245678), (1235876), (1234876), (1237845), (1236845), (1236745)) & 20160 \\ 40 & ((2348567), (1345678), (1245678), (1235678), (1674328), (1574328), (1564328), (1432765)) & 40320 \\ 41 & ((2345678), (1345678), (1245678), (1235678), (1234687), (1234587), (1234856), (1234756)) & 40320 \\ 42 & ((2345678), (1345678), (1245678), (1235687), (1234687), (1234587), (1238456), (1237456)) & 20160 \\ 43 & ((2345678), (1345678), (1245687), (1238567), (1238476), (1238475), (1283465), (1273654)) & 40320 \\ 44 & ((2345678), (1345678), (1245687), (1238576), (1238476), (1238745), (1283645), (1273654)) & 40320 \\ 45 & ((2345678), (1345678), (1245867), (1238567), (1238476), (1283475), (1283465), (1276354)) & 40320 \\ 46 & ((2345678), (1345867), (1245876), (1238756), (1238746), (1543728), (1453628), (1453267)) & 40320 \\ 47 & ((2345687), (1345876), (1248576), (1237865), (1283764), (1453278), (1485326), (1473526)) & 40320 \\ 48 & ((2345687), (1384567), (1284567), (1765328), (1674328), (1574328), (1234658), (1654237)) & 20160 \\ 49 & ((2345867), (1348567), (1284567), (1283567), (1674328), (1234758), (1234658), (1524376)) & 40320 \\ 50 & ((2345867), (1348567), (1284567), (1283756), (1647328), (1237458), (1236548), (1524376)) & 20160 \\ 51 & ((2345867), (1384567), (1284567), (1675328), (1674328), (1237458), (1236458), (1542376)) & 20160 \\ 52 & ((2348567), (1348567), (1245687), (1237568), (1283746), (1283745), (1238654), (1256374)) & 40320 \\ 53 & ((2348567), (1384567), (1284567), (1765328), (1234786), (1234785), (1234658), (1423756)) & 40320 \\ 54 & ((2384567), (1765438), (1765428), (1235786), (1234876), (1237845), (1236485), (1325746)) & 40320 \\ 55 & ((2765438), (1765438), (1245678), (1235768), (1234876), (1237485), (1236485), (1257643)) & 40320 \\ 56 & ((2765438), (1765438), (1245678), (1235786), (1234876), (1237845), (1236485), (1257463)) & 40320 \\ 57 & ((2345678), (1345678), (1245678), (1235687), (1234867), (1234857), (1238456), (1237465)) & 40320 \\ 58 & ((2345678), (1345678), (1245678), (1235867), (1234876), (1238475), (1238465), (1237645)) & 40320 \\ 59 & ((2345678), (1345678), (1245687), (1235687), (1234867), (1234857), (1283456), (1273465)) & 40320 \\ 60 & ((2345678), (1345678), (1245687), (1235876), (1234876), (1238745), (1283645), (1273645)) & 40320 \\ 61 & ((2345678), (1345678), (1245867), (1235867), (1234876), (1283475), (1283465), (1276345)) & 40320 \\ 62 & ((2345678), (1345678), (1245867), (1235876), (1234876), (1283745), (1283645), (1276345)) & 40320 \\ 63 & ((2345678), (1345867), (1245867), (1238576), (1238476), (1547328), (1546328), (1453267)) & 40320 \\ 64 & ((2345678), (1345867), (1284567), (1283576), (1283476), (1547328), (1546328), (1345267)) & 40320 \\ 65 & ((2345678), (1348567), (1284567), (1283576), (1674328), (1547328), (1546328), (1342567)) & 20160 \\ 66 & ((2345867), (1345867), (1248567), (1238576), (1283476), (1237458), (1236458), (1253476)) & 40320 \\ 67 & ((2345867), (1348567), (1248567), (1238576), (1674328), (1237458), (1236458), (1523476)) & 20160 \\ 68 & ((2348567), (1384567), (1284567), (1765328), (1234876), (1234785), (1234685), (1423576)) & 40320 \\ 69 & ((2384567), (1345687), (1245687), (1675328), (1674328), (1547328), (1238645), (1456237)) & 40320 \\ 70 & ((2384567), (1345687), (1245687), (1765328), (1674328), (1574328), (1238465), (1456237)) & 40320 \\ 71 & ((2348567), (1348567), (1245678), (1235678), (1283476), (1283475), (1283465), (1256734)) & 20160 \\ 72 & ((2384567), (1345678), (1245678), (1675328), (1674328), (1547328), (1546328), (1327654)) & 20160 \\ 73 & ((2384567), (1384567), (1245867), (1283576), (1283476), (1238745), (1238645), (1245367)) & 20160 \\ 74 & ((2345867), (1345867), (1245678), (1235678), (1234678), (1283457), (1283456), (1267345)) & 20160 \\ 75 & ((2345678), (1345687), (1245867), (1238567), (1238476), (1283475), (1564328), (1453627)) & 20160 \\ 76 & ((2345678), (1345687), (1245867), (1238756), (1238746), (1283745), (1456328), (1453627)) & 40320 \\ 77 & ((2345678), (1345687), (1248567), (1238576), (1283476), (1283745), (1546328), (1435627)) & 40320 \\ 78 & ((2345678), (1345867), (1248567), (1238576), (1283476), (1547328), (1546328), (1435267)) & 40320 \\ 79 & ((2345678), (1345867), (1248567), (1238756), (1283746), (1547328), (1456328), (1435267)) & 40320 \\ 80 & ((2345687), (1345867), (1245867), (1237856), (1238746), (1547328), (1236584), (1475326)) & 40320 \\ 81 & ((2345687), (1345867), (1284567), (1283567), (1283476), (1574328), (1234658), (1625437)) & 40320 \\ 82 & ((2345687), (1348567), (1284567), (1283567), (1674328), (1574328), (1234658), (1652437)) & 40320 \\ 83 & ((2345867), (1348567), (1284567), (1283576), (1674328), (1237458), (1236458), (1524376)) & 40320 \\ 84 & ((2348567), (1348576), (1245768), (1237568), (1283746), (1543827), (1453826), (1257634)) & 20160 \\ 85 & ((2345678), (1345678), (1245687), (1235867), (1234867), (1238457), (1283456), (1273645)) & 20160 \\ 86 & ((2345678), (1345678), (1245867), (1238576), (1238476), (1283745), (1283645), (1276354)) & 20160 \\ 87 & ((2345678), (1345687), (1245687), (1235867), (1234876), (1238475), (1564328), (1546327)) & 40320 \\ 88 & ((2345678), (1345687), (1245867), (1235867), (1234876), (1283475), (1564328), (1543627)) & 40320 \\ 89 & ((2345687), (1345687), (1248567), (1238576), (1283476), (1283745), (1236458), (1265347)) & 40320 \\ 90 & ((2345687), (1348567), (1248567), (1238576), (1674328), (1547328), (1236458), (1652347)) & 40320 \\ 91 & ((2345867), (1345867), (1245687), (1235678), (1234768), (1283475), (1238465), (1263745)) & 40320 \\ 92 & ((2348567), (1348567), (1245678), (1235768), (1283476), (1283745), (1283645), (1256734)) & 40320 \\ 93 & ((2348567), (1348567), (1245687), (1235768), (1283476), (1283745), (1238645), (1256374)) & 40320 \\ 94 & ((2348567), (1348567), (1245867), (1235786), (1283476), (1237485), (1236845), (1253746)) & 40320 \\ 95 & ((2384567), (1345867), (1245867), (1675328), (1674328), (1238745), (1238645), (1452367)) & 20160 \\ 96 & ((2345687), (1345867), (1284567), (1283576), (1283476), (1547328), (1236458), (1625437)) & 40320 \\ 97 & ((2345687), (1348567), (1284567), (1283576), (1674328), (1547328), (1236458), (1652437)) & 20160 \\ 98 & ((2348567), (1345867), (1245678), (1237568), (1647328), (1283745), (1283654), (1437625)) & 40320 \\ 99 & ((2348567), (1345867), (1245687), (1237568), (1647328), (1283745), (1238654), (1473625)) & 40320 \\ 100 & ((2348567), (1384567), (1284576), (1657328), (1237468), (1273458), (1263548), (1423765)) & 5040 \\ 101 & ((2345678), (1345678), (1245687), (1235867), (1234876), (1238475), (1283465), (1273645)) & 40320 \\ 102 & ((2345678), (1345687), (1245867), (1238576), (1238476), (1283745), (1546328), (1453627)) & 40320 \\ 103 & ((2384567), (1345687), (1245678), (1765328), (1674328), (1574328), (1283465), (1372654)) & 40320 \\ 104 & ((2384567), (1345867), (1245678), (1765328), (1674328), (1283475), (1283465), (1376254)) & 40320 \\ 105 & ((2384567), (1345867), (1245687), (1675328), (1674328), (1283745), (1238645), (1452637)) & 40320 \\ 106 & ((2384567), (1345867), (1245687), (1765328), (1674328), (1283475), (1238465), (1452637)) & 20160 \\ 107 & ((2384567), (1348567), (1245678), (1765328), (1283476), (1283475), (1283465), (1376524)) & 20160 \\ 108 & ((2384567), (1348567), (1245687), (1675328), (1283476), (1283745), (1238645), (1425637)) & 40320 \\ 109 & ((2384567), (1348567), (1245867), (1675328), (1283476), (1237845), (1236845), (1425376)) & 40320 \\ 110 & ((2345687), (1348567), (1248567), (1238756), (1647328), (1547328), (1236548), (1652347)) & 10080 \\ 111 & ((2348567), (1345687), (1245678), (1235678), (1674328), (1574328), (1283465), (1437265)) & 40320 \\ 112 & ((2348567), (1345687), (1245687), (1235678), (1674328), (1574328), (1238465), (1473265)) & 40320 \\ 113 & ((2348567), (1345867), (1245678), (1235678), (1674328), (1283475), (1283465), (1437625)) & 40320 \\ 114 & ((2348567), (1345867), (1245867), (1235768), (1674328), (1238745), (1238645), (1476325)) & 20160 \\ 115 & ((2348567), (1348567), (1245867), (1235768), (1283476), (1238745), (1238645), (1253674)) & 40320 \\ 116 & ((2345867), (1345687), (1245678), (1235678), (1234678), (1754328), (1283456), (1543726)) & 20160 \\ 117 & ((2345867), (1345687), (1245687), (1235678), (1234678), (1754328), (1238456), (1547326)) & 20160 \\ 118 & ((2345867), (1345867), (1245687), (1235678), (1234678), (1283457), (1238456), (1263745)) & 20160 \\ 119 & ((2348567), (1345867), (1245876), (1237568), (1647328), (1273845), (1263854), (1467325)) & 20160 \\ 120 & ((2345678), (1345867), (1248576), (1238756), (1283746), (1543728), (1453628), (1435267)) & 20160 \\ 121 & ((2345687), (1345867), (1248567), (1238576), (1283476), (1547328), (1236458), (1625347)) & 40320 \\ 122 & ((2348567), (1345867), (1245678), (1235768), (1674328), (1283745), (1283645), (1437625)) & 40320 \\ 123 & ((2348567), (1345867), (1245687), (1235768), (1674328), (1283745), (1238645), (1473625)) & 40320 \\ 124 & ((2348567), (1345876), (1245768), (1237568), (1647328), (1543827), (1453826), (1436725)) & 40320 \\ 125 & ((2348567), (1348567), (1245786), (1237568), (1283746), (1273845), (1268354), (1257364)) & 20160 \\ 126 & ((2384567), (1345687), (1245678), (1675328), (1674328), (1547328), (1283645), (1372654)) & 40320 \\ 127 & ((2384567), (1345867), (1245678), (1675328), (1674328), (1283745), (1283645), (1376254)) & 20160 \\ 128 & ((2384567), (1348567), (1245867), (1675328), (1283476), (1238745), (1238645), (1425367)) & 40320 \\ 129 & ((2348567), (1345867), (1245786), (1237568), (1647328), (1273845), (1268354), (1463725)) & 20160 \\ 130 & ((2345687), (1345867), (1248567), (1238756), (1283746), (1547328), (1236548), (1625347)) & 20160 \\ 131 & ((2348567), (1345867), (1245687), (1235678), (1674328), (1283475), (1238465), (1473625)) & 20160 \\ 132 & ((2348567), (1345867), (1245687), (1235678), (1764328), (1283457), (1238456), (1473625)) & 2880 \\ 133 & ((2348567), (1345876), (1245786), (1237568), (1647328), (1548327), (1453826), (1463725)) & 20160 \\ 134 & ((2345687), (1345867), (1248576), (1238756), (1283746), (1543728), (1263548), (1625347)) & 10080

\end{longtable}\vspace{-2mm}

\end{center}

\subsection*{Acknowledgements}

The authors thank B.~Schroeter, B.~Sturmfels and B.~Umbert for useful correspondence and discussions. This research was supported in part by a grant from the Gluskin Sheff/Onex Freeman Dyson Chair in Theoretical Physics and by Perimeter Institute. Research at Perimeter Institute is supported in part by the Government of Canada through the Department of Innovation, Science and Economic Development Canada and by the Province of Ontario through the Ministry of Colleges and Universities. This research received funding from the European Research Council (ERC) under the European Union's Horizon 2020 research and innovation programme (grant agreement no.~725110), Novel structures in scattering amplitudes.

\pdfbookmark[1]{References}{ref}
\LastPageEnding

\end{document}